\documentclass[journal,comsoc]{IEEEtran}
\usepackage{graphicx}
\usepackage{amssymb}
\usepackage{amsmath}
\usepackage{cite}
\usepackage{subfigure}
\usepackage{mathrsfs}
\usepackage[displaymath,mathlines]{lineno}
\usepackage{color}
\usepackage{tabulary}
\usepackage{multirow}
\usepackage{algpseudocode}
\usepackage{algorithm,algpseudocode}
\usepackage{algorithmicx}
\usepackage{pbox}
\usepackage{multicol}
\usepackage{lipsum}
\usepackage{amsthm}
\usepackage{relsize}
\usepackage{lipsum}
\usepackage{epstopdf}
\usepackage{mathtools}
\usepackage{calc}
\usepackage{makecell}

\makeatletter
\newcommand{\doublewidetilde}[1]{{%
		\mathpalette\double@widetilde{#1}}}
\newcommand{\double@widetilde}[2]{%
		\sbox\z@{$\m@th#1\widetilde{#2}$}%
		\ht\z@=.5\ht\z@
		\widetilde{\box\z@}}
\makeatother

\newtheorem{theorem}{Theorem}
\newtheorem{lemma}{Lemma}
\newtheorem{corollary}{Corollary}

\newtheorem{remark}{Remark}
\newtheorem{assumption}{Assumption}

\newcounter{eqnback}
\newcounter{eqncnt}

\begin{document}
%
\title{\huge Reconfigurable Intelligent Surface-Assisted Cell-Free Massive MIMO Systems Over Spatially-Correlated Channels \vspace*{-0.1cm}}

\author{\normalsize Trinh~Van~Chien, \textit{Member}, \textit{IEEE}, Hien~Quoc~Ngo, \textit{Senior Member}, \textit{IEEE},  Symeon~Chatzinotas, \textit{Senior Member}, \textit{IEEE}, Marco~Di~Renzo, \textit{Fellow}, \textit{IEEE}, and Bj\"{o}rn~Ottersten, \textit{Fellow}, \textit{IEEE} \vspace*{-0.4cm}
\thanks{The work of T. V. Chien, S. Chatzinotas, and B. Ottersten was supported by
	RISOTTI - Reconfigurable Intelligent Surfaces for Smart Cities under project
	FNR/C20/IS/14773976/RISOTTI. The work of H. Q. Ngo was supported by
	the UK Research and Innovation Future Leaders Fellowships under Grant
	MR/S017666/1. The work of M. Di Renzo was supported in part by the
	European Commission through the H2020 ARIADNE project under grant
	agreement number 871464 and through the H2020 RISE-6G project under
	grant agreement number 101017011. Parts of this paper were presented at IEEE GLOBECOM 2021 \cite{Chien2021a}. The associate editor coordinating the review of this paper and approving it for publication was Z. Zhang. (Corresponding author: Trinh Van Chien.)}
\thanks{T. V. Chien, S. Chatzinotas, B. Ottersten are with the Interdisciplinary Centre for Security, Reliability and Trust (SnT), University of Luxembourg, L-1855 Luxembourg, Luxembourg (email: vanchien.trinh@uni.lu, symeon.chatzinotas@uni.lu, and bjorn.ottersten@uni.lu).}
\thanks{H. Q. Ngo is with the School of Electronics, Electrical Engineering and Computer Science, Queen's University Belfast, Belfast BT7 1NN, United Kingdom (email: hien.ngo@qub.ac.uk).}
\thanks{M.  Di  Renzo  is  with  Universit\'e  Paris-Saclay,  CNRS,  CentraleSup\'elec, Laboratoire  des  Signaux  et  Syst\`emes, 3 Rue
	Joliot-Curie, 91192  Gif-sur-Yvette,  France (email: marco.di-renzo@universite-paris-saclay.fr).}
}

\maketitle

\begin{abstract}
Cell-Free Massive multiple-input multiple-output (MIMO) and reconfigurable intelligent surface (RIS) are two promising technologies for application to beyond-5G networks. This paper considers Cell-Free Massive MIMO systems with the assistance of an RIS for enhancing the system performance under the presence of spatial correlation among the engineered scattering elements of the RIS. Distributed maximum-ratio processing is considered at the access points (APs). We introduce an \textit{aggregated channel} estimation approach that provides sufficient information for data processing with the main benefit of reducing the overhead required for channel estimation. The considered system is studied by using asymptotic analysis which lets the number of APs and/or the number of RIS elements grow large. A lower bound for the channel capacity is obtained for a finite number of APs and engineered scattering elements of the RIS, and closed-form expressions for the uplink and downlink ergodic net throughput are formulated in terms of only the channel statistics. Based on the obtained analytical frameworks, we unveil the impact of channel correlation, the number of RIS elements, and the pilot contamination on the net throughput of each user. In addition, a simple control scheme for optimizing the configuration of the engineered scattering elements of the RIS is proposed, which is shown to increase the channel estimation quality, and, hence, the system performance. Numerical results demonstrate the effectiveness of the proposed system design and performance analysis. In particular, the performance benefits of using RISs in Cell-Free Massive MIMO systems are confirmed, especially if the direct links between the APs and the users are of insufficient quality with high probability. 
\end{abstract}


%
\IEEEpeerreviewmaketitle

\vspace*{-0.5cm}
\section{Introduction}
\vspace*{-0.1cm}
In the last few decades, we have witnessed an exponential growth of the demand for wireless communication systems that provide reliable communications and ensure ubiquitous coverage, high spectral efficiency and low latency \cite{giordani2020toward}. To meet these requirements, several new technologies have been  incorporated  in $5$G communication standards, which include Massive multiple-input multiple-output (MIMO) \cite{Chien2017a}, millimeter-wave communications \cite{rappaport2017overview}, and network densification \cite{Bjornson2016c}. Among them, Massive MIMO has gained significant attention since it can offer a good service to many users in the network. Moreover, the net throughput offered by a Massive MIMO system is close to the Shannon capacity, in many scenarios, by only employing simple linear processing techniques, such as maximum ratio (MR) or zero forcing (ZF) processing. Since the net throughput can be computed in a closed-form expression that only depends on the channel statistics, the optimized designs are applicable for a long period of time \cite{van2020power}. The colocated Massive MIMO architecture has the advantage of low backhaul requirements since the base station antennas are installed in a compact array. Conventional cellular networks, however, are impaired by  intercell interference. In particular, the users at the cell boundaries are impaired by high intercell interference and path loss, and hence, they may experience insufficient performance. More advanced signal processing methods are necessary to overcome the inherent intercell interference that characterizes conventional cellular network deployments. 

Cell-Free Massive MIMO has recently been introduced to reduce the intercell interference that characterizes colocated Massive MIMO architectures. Cell-Free Massive MIMO is a network deployment where a large number of access points (APs) are located in a given coverage area to serve a small number of users \cite{interdonato2019ubiquitous,ngo2017cell,9136914,interdonato2021enhanced}. All APs collaborate with each other via a backhaul network and serve all the users in the absence of cell boundaries. The system performance is enhanced in Cell-Free Massive MIMO systems because they inherit the benefits of the distributed MIMO and network MIMO architectures, but the users are also close to the APs. When each AP is equipped with a single antenna, MR processing results in a good net throughput for every user, while ensuring a low computational complexity and offering a distributed implementation that is convenient for scalability purposes \cite{9064545}. However, Cell-Free Massive MIMO cannot guarantee a good quality of service under harsh propagation conditions, such as in the presence of poor scattering environments or high attenuation due to the presence of large obstacles. 

Reconfigurable intelligent surface (RIS) is an emerging technology that is capable of shaping the radio waves at the electromagnetic level without applying digital signal processing methods and without requiring power amplifiers \cite{wu2019intelligent,le2020robust, 9140329}. Each element of the RIS scatters (e.g., reflects) the incident signal without using radio frequency chains and power amplification \cite{9326394}. Integrating an RIS into  wireless networks introduces digitally controllable links that scale up with the number of engineered scattering elements of the RIS, whose estimation is, however, challenged by the lack of digital signal processing units at the RIS \cite{9198125,9200578,abrardo2020intelligent,8937491,wei2021channel}. For simplicity, the main attention has so far been concentrated on designing the phase shifts under the assumption of perfect channel state information (CSI) \cite{9198125,perovic2021achievable} and the references therein. In \cite{wei2021channel}, the authors have recently discussed the fundamental issues of performing channel estimation in RIS-assisted wireless systems. The impact of the channel estimation overhead and reporting on the spectral efficiency, energy efficiency, and their tradeoff has recently been investigated in \cite{9200578}. In \cite{9198125} and \cite{abrardo2020intelligent}, to reduce the impact of the channel estimation overhead, the authors have investigated the design of RIS-assisted communications in the presence of statistical CSI. As far as the integration of Cell-Free Massive MIMO and RIS is concerned, recent works have formulated and solved optimization problems with different communication objectives under the assumption of perfect (and instantaneous) CSI  \cite{zhou2020achievable, zhang2020capacity, bashar2020performance, 9286726, zhang2021beyond}. Recent results in the context of single-input single-output (SISO) and multi-user MIMO systems have, however, shown that designs for the engineered scattering
	elements of the RIS that are based on statistical CSI may be of practical interest and provide good performance  \cite{abrardo2020intelligent,9195523,hou2020reconfigurable,van2021outage}.

In the depicted context, no prior work has analyzed the performance of RIS-assisted Cell-Free Massive MIMO systems in the presence of spatially-correlated channels. In this work, motivated by these considerations, we introduce an analytical framework for analyzing and optimizing the uplink and downlink transmissions of RIS-assisted Cell-Free Massive MIMO systems under spatially correlated channels
and in the presence of direct links subject to the presence of blockages. In particular, the main contributions made by this paper can be summarized as follows:
\begin{itemize}
    \item We consider an RIS-assisted Cell-Free Massive MIMO under spatially correlated channels. All APs estimate the instantaneous channels in the uplink pilot training phase. We exploit a channel estimation scheme that estimates the aggregated channels including both the direct and indirect links, instead of every individual channel coefficient as in previous works \cite{bashar2020performance,wei2021channel}. For generality, the pilot contamination is assumed to originate from an arbitrary pilot reuse pattern. 
    \item We analytically show that, even by using a low complexity MR technique, the non-coherent interference, small-scale fading effects, and additive noise are averaged out when the number of APs and RIS elements increases. The received signal includes, hence, only the desired signal and the coherent interference. In addition, we show that the indirect links become dominant if the number of engineered scattering elements of the RIS increases.
    \item We derive a closed-form expression of the net throughput for both the uplink and downlink data transmissions. The impact of the array gain, coherent joint transmission, channel estimation errors, pilot contamination, spatial correlation, and phase shifts of the RIS, which determine the system performance, are explicitly observable in the obtained analytical expressions.
    \item With the aid of numerical simulations, we verify the effectiveness of the proposed channel estimation scheme and the accuracy of the closed-form expressions of the net throughput. The obtained numerical results show that the use of RISs  significantly enhances the net throughput per user, especially when the direct links are blocked with high probability.   
\end{itemize}
The rest of this paper is organized as follows: Section~\ref{Sec:SysModel} presents the system model, the channel model, and the channel estimation protocol. The uplink data transmission protocol and the asymptotic analysis by assuming a very large number of APs and engineered scattering
elements at the RIS are discussed in Section~\ref{Sec:UL}. A similar analysis for the downlink data transmission is reported in  Section~\ref{Sec:Downlink}. Finally, Section~\ref{Sec:NumRes} illustrates several numerical results, while the main conclusions are drawn in Section~\ref{Sec:Conclusion}. 

\textit{Notation}: Upper and lower bold letters are used to denote matrices and vectors, respectively. The identity matrix of size $N \times N$ is denoted by $\mathbf{I}_N$. The imaginary unit of a complex number is denoted by $j$ with $\sqrt{j} = -1$. The superscripts $(\cdot)^{\ast},$ $(\cdot)^T,$ and $(\cdot)^H$ denote the complex conjugate, transpose, and Hermitian transpose, respectively. $\mathbb{E}\{ \cdot\}$ and $\mathsf{Var} \{ \cdot \}$ denote the expectation and variance of a random variable. The circularly symmetric Gaussian distribution is denoted by  $\mathcal{CN}(\cdot, \cdot)$ and  $\mathrm{diag} (\mathbf{x})$ is the diagonal matrix whose main diagonal is given by $\mathbf{x}$. $\mathrm{tr}(\cdot)$ is the trace operator. The Euclidean norm of vector $\mathbf{x}$ is $\| \mathbf{x}\|$, and $\| \mathbf{X} \|_2$ is the spectral norm of matrix $\mathbf{X}$. Finally, $\mathrm{mod}(\cdot,\cdot)$ is the modulus operation and $\lfloor \cdot \rfloor$ denotes the truncated argument.
\begin{figure}[t]
	\centering
	\includegraphics[trim=0.0cm 0.1cm 0.00cm 0.1cm, clip=true, width=3.0in]{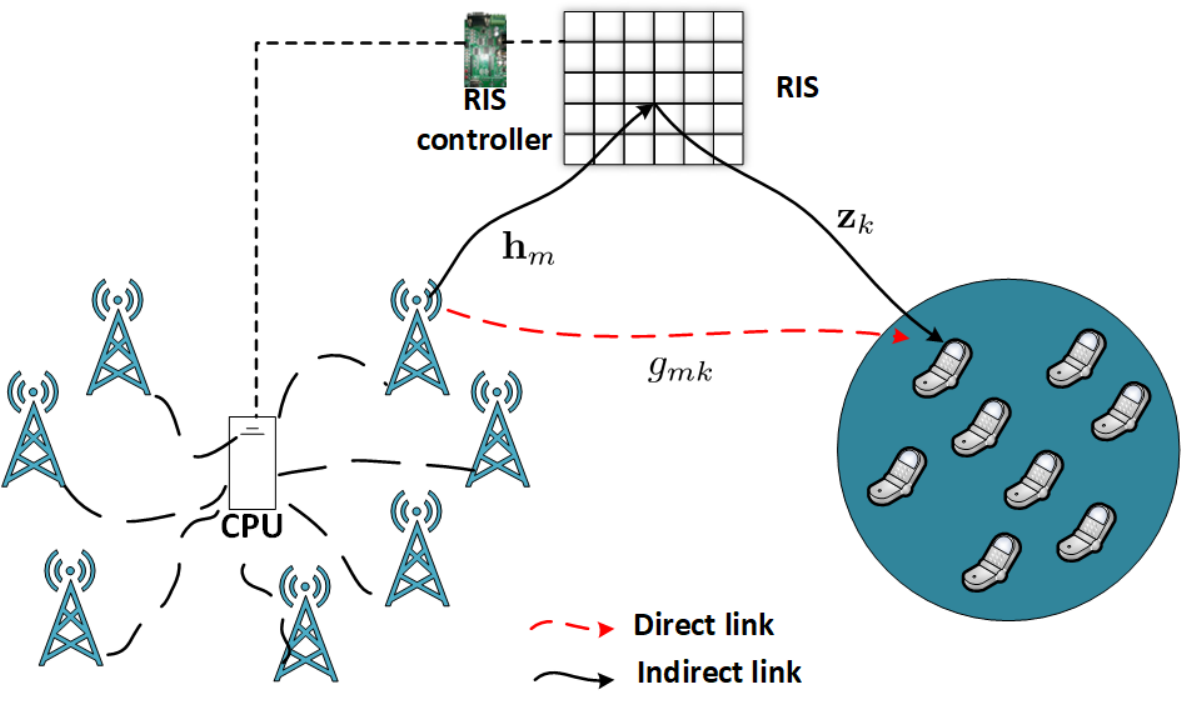} \vspace*{-0.2cm}
	\caption{An RIS-assisted Cell-Free Massive MIMO system where $M$ APs collaborate with each other to serve $K$ distant users.}
	\label{FigSysModel}
	\vspace*{-0.5cm}
\end{figure}
\vspace*{-0.2cm}
\section{System Model, Channel Estimation, and RIS Phase Shift Control} \label{Sec:SysModel}
\vspace*{-0.1cm}
We consider an RIS-assisted Cell-Free Massive MIMO system, where $M$ APs connected to a central processing unit (CPU) serve $K$  users on the same time and frequency resource, as schematically illustrated in Fig.~\ref{FigSysModel}. All APs and users are equipped with a single antenna and they are randomly located in the coverage area. Since the considered users are far away from the APs, the communication is assisted by an RIS that comprises $N$ engineered scattering elements that can modify the phases of the incident signals.\footnote{In general, a completed system should include many users randomly located in the coverage area. There are some favorable users where their links to some APs are strong. But there are also some unfavorable users where the their links to the APs are weak. This may come from the large path loss (long distances) or heavy shadowing. In our work, we consider the cases where RISs are deployed to improve the performance of these unfavorable users, and hence, the coverage of the whole system can be increased.} The matrix of phase shifts of the RIS is denoted by $\pmb{\Phi} = \mathrm{diag} \left( [ e^{j\theta_1}, \ldots, e^{j\theta_N}]^T \right)$, where $\theta_n \in [-\pi, \pi]$ is the phase shift applied by the $n$-th element of the RIS. The phase shifts are adjusted by a controller which exchanges information with the APs via a backhaul link (see Fig.~\ref{FigSysModel}). As a canonical form of Cell-Free Massive MIMO systems, we assume that the system operates in time-division duplexing (TDD) mode. Thus, we assume that channel reciprocity holds in the consisted system model.
\vspace*{-0.3cm}
\subsection{Channel Model}
\vspace*{-0.1cm}
We assume a quasi-static block fading model where the channels are static and frequency flat in each coherence interval comprising $\tau_c$ symbols. We assume that the APs estimate the channel during the uplink pilot training phase. Thus, $\tau_p$ symbols ($\tau_p < \tau_c$) in each coherence interval are dedicated to the channel estimation phase and the remaining $(\tau_c - \tau_p)$ symbols are utilized for the uplink and downlink data transmission phases.

The following notation is used: $g_{mk}$ is the channel between the user~$k$ and the AP~$m$, which is the direct link \cite{wu2019intelligent};  $\mathbf{h}_m \in \mathbb{C}^{N}$ is the channel between the AP~$m$ and the RIS; and $\mathbf{z}_{k} \in \mathbb{C}^{N}$ is the channel between the RIS and the user~$k$. The pair $\mathbf{h}_m$ and $\mathbf{z}_{k}$, which constitutes the cascaded channel, results in an indirect link (virtual line-of-sight link), which enhances the communication reliability between the AP~$m$ and the user~$k$ \cite{di2019smart}. The majority of existing works assume that the wireless channels undergo uncorrelated Rayleigh fading. In this paper, we consider a more realistic channel model by taking into account  the spatial correlation among the engineered scattering elements of the RIS, which is due to their sub-wavelength size, sub-wavelength inter-distance, and geometric layout. In an isotropic propagation environment, in particular, $g_{mk}$, $\mathbf{h}_{m}$, and $\mathbf{z}_{k}$ can be modeled as follows
\begin{equation} \label{eq:Channels}
g_{mk} \sim \mathcal{CN}(0, \beta_{mk}), \mathbf{h}_{m} \sim \mathcal{CN}(\mathbf{0}, \mathbf{R}_{m}), \mathbf{z}_{k}  \sim \mathcal{CN}(\mathbf{0}, \widetilde{\mathbf{R}}_{k}),
\end{equation}
where $\beta_{mk}$ is the large-scale fading coefficient; $\mathbf{R}_{m} \in \mathbb{C}^{N \times N}$ and $\widetilde{\mathbf{R}}_{k} \in \mathbb{C}^{N \times N}$ are the covariance matrices that characterize the spatial correlation among the channels of the RIS elements. The covariance matrices in \eqref{eq:Channels} correspond to a general model, which can be further particularized for application to typical RIS designs and propagation environments. For example, a simple exponential model was used to describe the spatial correlation among the engineered scattering elements of the RIS in \cite{alwazani2020intelligent}. Another recent model that is applicable to isotropic scattering with  uniformly distributed multipath components in the half-space in front of the RIS was recently reported in \cite{bjornson2020rayleigh}, whose covariance matrices are\footnote{This paper considers a spatial correlation model between engineered scattering elements in the far-field scenarios that is applicable and when the covariance matrices among the users differ only in terms of large-scale channel coefficients. Other scenarios where the users have covariance matrices with different phases \cite{clerckx2008correlated} are of interest and are left for future work. } 
\begin{equation}\label{eq:CovarMa}
\mathbf{R}_m = \alpha_{m} d_Hd_V\mathbf{R} \mbox{ and } \widetilde{\mathbf{R}}_{k} = \tilde{\alpha}_{k} d_Hd_V\mathbf{R},
\end{equation}
where $\alpha_{m}, \tilde{\alpha}_{k} \in \mathbb{C}$ are the large-scale channel coefficients, which, for example, model the signal attenuation due to large objects and due to the transmission distance. The matrices in \eqref{eq:CovarMa} assume that the size of each element of the RIS is $d_H \times d_V$, with $d_H$ being the horizontal width and $d_V$ being the vertical height of each RIS element. In particular, the $(m',n')-$th element of the spatial correlation matrix $\mathbf{R} \in \mathbb{C}^{N \times N }$ in \eqref{eq:CovarMa} is
$[\mathbf{R}]_{m'n'}= \mathrm{sinc} (2 \|\mathbf{u}_{m'} - \mathbf{u}_{n'} \|/ \lambda)$,
where $\lambda$ is the wavelength and $\mathrm{sinc}(x) = \sin(\pi x) / (\pi x)$ is the $\mathrm{sinc}$ function. The vector $\mathbf{u}_{x}, x \in \{ m',n'\}$ is given by
$\mathbf{u}_{x} = [0, \mod(x-1,N_H)d_H, \lfloor (x-1)/N_H\rfloor d_V]^T$,
where $N_H$ and $N_V$ denote the total number of RIS elements in each row and column, respectively. The channel model in \eqref{eq:Channels}  is significantly distinct from related works since the small-scale fading and the spatial correlation matrices are included in both links of the virtual line-of-sight link that comprises the RIS. In \cite{alwazani2020intelligent}, by contrast, the channels between the transmitters and the RIS are assumed to be deterministic, for analytical tractability.
\vspace*{-0.2cm}
\subsection{Uplink Pilot Training Phase}
\vspace*{-0.1cm}
The channels are independently estimated from the $\tau_p$ pilot sequences transmitted by the $K$ users. All the users share the same $\tau_p$ pilot sequences. In particular, $\pmb{\phi}_k \in \mathbb{C}^{\tau_p}$ with $\| \pmb{\phi}_k \|^2 = 1$ is defined as the pilot sequence allocated to the user~$k$. We denote by $\mathcal{P}_k$ the set of indices of the users (including the user $k$) that share the same pilot sequence as the user $k$. The pilot sequences are assumed to be mutually orthogonal such that the pilot reuse pattern is 
\begin{equation}
\pmb{\phi}_{k'}^H \pmb{\phi}_k = \begin{cases}
1, & \mbox{if } k' \in \mathcal{P}_k,\\
0, & \mbox{if } k' \notin \mathcal{P}_k.
\end{cases}
\end{equation}
During the pilot training phase, all the $K$ users transmit the pilot sequences to the $M$ APs simultaneously. In particular, the user~$k$ transmits the pilot sequence $\sqrt{\tau_p} \pmb{\phi}_k$. The received training signal at the AP~$m$, $\mathbf{y}_{pm} \in \mathbb{C}^{\tau_p}$, can be written as  
\begin{equation} \label{eq:ReceivedPilot}
\mathbf{y}_{pm} = \sum_{k=1}^K \sqrt{p \tau_p}  g_{mk} \pmb{\phi}_k + \sum_{k=1}^K \sqrt{p \tau_p} \mathbf{h}_{m}^H \pmb{\Phi} \mathbf{z}_{k} \pmb{\phi}_k  + \mathbf{w}_{pm},
\end{equation}
where $p$ is the normalized signal-to-noise ratio (SNR) of each pilot symbol, and $\mathbf{w}_{pm} \in \mathbb{C}^{\tau_p}$ is the additive noise at the AP~$m$, which is  distributed as $\mathbf{w}_{pm} \sim \mathcal{CN} (\mathbf{0}, \mathbf{I}_{\tau_p})$. In order for the AP~$m$ to estimate the desired channels from the user~$k$, the received training signal in \eqref{eq:ReceivedPilot} is projected on $\pmb{\phi}_k^H$ as 
\begin{multline} \label{eq:ReceivedPilotv1}
	y_{pmk} = \pmb{\phi}_k^H \mathbf{y}_{pm} =  \sqrt{p \tau_p}  \left(g_{mk} +  \mathbf{h}_{m}^H \pmb{\Phi} \mathbf{z}_{k} \right) + \\ 
	\sum_{k' \in \mathcal{P}_k \setminus \{k\} } \sqrt{p\tau_p}  \left(g_{mk'} +  \mathbf{h}_{m}^H \pmb{\Phi} \mathbf{z}_{k'} \right) + w_{pmk},
\end{multline}
where $w_{pmk} = \pmb{\phi}_k^H \mathbf{w}_{pm} \sim \mathcal{CN}(0, 1)$. We emphasize that the co-existence of the direct and indirect channels due to the presence of the RIS results in a complicated channel estimation process. In particular, the cascaded channel in \eqref{eq:ReceivedPilotv1} results in a nontrivial procedure for applying the minimum mean-square error (MMSE) estimation method, as reported in previous works, for processing the projected signals \cite{ngo2017cell,9136914}. Based on the specific signal structure in \eqref{eq:ReceivedPilotv1}, we denote the channel between the AP~$m$ and the user~$k$ through the RIS as
\begin{equation} \label{eq:umk}
u_{mk} = g_{mk} + \mathbf{h}_{m}^H \pmb{\Phi} \mathbf{z}_{k}, 
\end{equation}
which is referred to as the \textit{aggregated channel} that comprises the direct and indirect link between the user~$k$ and the AP~$m$. In contrast to  previous works, where the matrix $\pmb{\Phi}$ of the RIS phase shifts is optimized based on instantaneous CSI, in this paper, $\pmb{\Phi}$ is optimized based on statistical CSI. This is detailed in Section~\ref{sec:PS_optimization}. 
 By capitalizing on the definition of the aggregated channel in \eqref{eq:umk}, the required channels can be estimated in an effective manner even in the presence of the RIS. In particular, the aggregated channel in \eqref{eq:umk} is given by the product of weighted complex Gaussian and spatially correlated random variables, as given in \eqref{eq:Channels}. Despite the complex analytical form, the following lemma gives information on the statistics of the aggregated channels.
\begin{lemma} \label{lemma:ChannelProperty}
The second and fourth moments of the aggregated channel $u_{mk}$ can be formulated as follows
\begin{align}
\mathbb{E} \{ |u_{mk}|^2 \} &= \delta_{mk}, \label{eq:2Order} \\
\mathbb{E} \{ |u_{mk}|^4 \}  &= 2 \delta_{mk}^2 + 2 \mathrm{tr} ( \pmb{\Theta}_{mk}^2),\label{eq:4Order}
\end{align}
where $\pmb{\Theta}_{mk} = \pmb{\Phi}^H \mathbf{R}_m \pmb{\Phi} \widetilde{\mathbf{R}}_{k} $ and $\delta_{mk} = \beta_{mk} + \mathrm{tr}(\pmb{\Theta}_{mk})$. Moreover, the aggregated channels are mutually independent for $m \neq m'$ and $k \neq k'$, i.e.,
\begin{equation}
 \mathbb{E} \{ u_{mk} u_{m'k'}^\ast \} = \mathbb{E} \{ u_{mk} \} \mathbb{E} \{ u_{m'k'}^\ast \} = 0.
\end{equation}
In addition, the aggregated channels $u_{mk}$ and $u_{m'k}, \forall m \neq m',$ and  the aggregated channels $u_{mk}$ and $u_{mk'}, \forall k \neq k',$ are mutually uncorrelated, i.e.,
\begin{align} 
& \mathbb{E} \{ u_{mk} u_{m'k}^\ast \} = 0 \mbox{ and } \mathbb{E} \{ u_{mk'} u_{mk}^\ast \} = 0.  \label{eq:Uncorrelated}
\end{align}
Besides, the aggregated channels $u_{mk}$, $u_{mk'}$, $u_{m'k}$, and $u_{m'k'},$ fulfill the following conditions
\begin{align}
& \mathbb{E} \{ |u_{mk} u_{m'k'}^\ast|^2 \} =\delta_{mk} \delta_{m'k'}, m \neq m', k \neq k', \label{eq:2IndeAggre}  \\
& \mathbb{E}\{ u_{mk}^\ast u_{mk'} u_{m'k'}^{\ast} u_{m'k}  \} = \mathrm{tr}( \pmb{\Theta}_{mk'} \pmb{\Theta}_{m'k}), m\neq m', k \neq k', \label{eq:4Chan} \\
& \mathbb{E} \{ | u_{mk} u_{m'k}^\ast|^2 \} = \delta_{mk} \delta_{m'k} + \mathrm{tr}( \pmb{\Theta}_{mk} \pmb{\Theta}_{m'k} ), m \neq m',  \label{eq:2UncorreChan}\\
& \mathbb{E} \{ | u_{mk} u_{mk'}^\ast|^2 \} = \delta_{mk} \delta_{mk'} + \mathrm{tr}( \pmb{\Theta}_{mk} \pmb{\Theta}_{mk'} ), k \neq k'.
\end{align}
\end{lemma}
\begin{proof}
See Appendix~\ref{appendix:ChannelProperty}.
\end{proof}
The moments in Lemma~\ref{lemma:ChannelProperty} are employed next for analyzing the channel estimation and the net throughput performance. We note, in addition, that the odd moments of $u_{mk}$, e.g., the first and third moments, are equal to zero. Conditioned on the phase shifts, we employ the linear MMSE method for estimating $u_{mk}$ at the AP. In spite of the complex structure of the RIS-assisted cascaded channel, Lemma~\ref{lemma:ChannelEst} provides analytical expressions of the estimated channels.
\begin{lemma} \label{lemma:ChannelEst}
By assuming that the AP~$m$ employs the linear MMSE estimation method based on the signal observation in \eqref{eq:ReceivedPilotv1}, the estimate of the aggregated channel ${u}_{mk}$ can be formulated as
\begin{equation} \label{eq:ChannelEst}
	\hat{u}_{mk} =  \big(\mathbb{E}\{ y_{pmk}^\ast u_{mk} \} y_{pmk} \big)/ \mathbb{E} \{ | y_{pmk} |^2 \}   = c_{mk} y_{pmk},
\end{equation}
where $c_{mk} =  \mathbb{E}\{ y_{pmk}^\ast u_{mk} \} / \mathbb{E} \{ | y_{pmk} |^2 \}$ has the following closed-form expression
\begin{equation} \label{eq:cmk}
c_{mk} =  \frac{\sqrt{p\tau_p} \delta_{mk} }{p\tau_p \sum_{k' \in \mathcal{P}_k} \delta_{mk'} + 1}.
\end{equation}
The estimated channel in \eqref{eq:ChannelEst} has zero mean and variance $\gamma_{mk}$ equal to
\begin{equation} \label{eq:gammamk}
\gamma_{mk} = \mathbb{E} \{ |\hat{u}_{mk}|^2 \} = \sqrt{p\tau_p}  \delta_{mk}c_{mk}.
\end{equation}
Also, the channel estimation error $e_{mk} = u_{mk} - \hat{u}_{mk}$ and the channel estimate $\hat{u}_{mk}$ are uncorrelated. The channel estimation error has zero mean and variance equal to
\begin{equation} \label{eq:EstError}
 \mathbb{E}\big\{ |e_{mk} |^2 \big\} = \delta_{mk}  - \gamma_{mk}.
\end{equation}
\end{lemma}
\begin{proof}
It is similar to the proof in \cite{Kay1993a}, and is obtained by applying similar analytical steps to the received signal in \eqref{eq:ReceivedPilotv1} and by taking into account the structure of the RIS-assisted cascaded channel and the spatial correlation matrices in \eqref{eq:Channels}.
\end{proof}
Lemma~\ref{lemma:ChannelEst} shows that, by assuming $\pmb{\Phi}$ fixed, the aggregated channel in \eqref{eq:umk} can be estimated without increasing the pilot training overhead, i.e., only $\tau_p$ symbols in each coherence interval are needed for channel estimation, which is the same as for conventional Cell-Free Massive MIMO systems. If the user~$k'$ uses the same pilot sequence as the user~$k$ does, then $\hat{u}_{mk'} = c_{mk'} y_{pmk}$ from \eqref{eq:ChannelEst}. Consequently, we obtain the relation $\hat{u}_{mk'}  = \frac{c_{mk'}}{c_{mk}}\hat{u}_{mk}$, which implies that, because of pilot contamination, it may be difficult to distinguish the signals of two generic users. In that regard it is worth noting that, to get rid of pilot contamination, one can assign mutually orthogonal pilot signals to all the users in the network (if  the coherence time is long enough so that $\tau_p \geq K$). Under mutually orthogonal pilot sequences, $c_{mk}$ and $\gamma_{mk}$ simplify to $c_{mk}^o$ and $\gamma_{mk}^o$, respectively, as follows
\begin{equation}
c_{mk}^o = \frac{\sqrt{p\tau_p} \delta_{mk}}{p\tau_p \delta_{mk} +1}, \, \gamma_{mk}^o =\sqrt{p\tau_p}\delta_{mk} c_{mk}^o.
\end{equation}
This implies that, in the absence of pilot contamination, we have $\gamma_{mk}^o \rightarrow \delta_{mk}$ as $\tau_p \rightarrow \infty$, i.e., the variance of the channel estimation error in \eqref{eq:gammamk} is equal to zero. The channel estimates given in Lemma~\ref{lemma:ChannelEst} can be applied to an arbitrary set of phase shifts and covariance matrices. To facilitate the performance analysis presented next, we introduce the following corollary that characterizes the correlation between the aggregated channels and their estimates.
\begin{corollary} \label{CorollaryReal}
Let us consider the two aggregated channels $u_{mk}$ and $u_{m'k}$ with $m\neq m'$, and let us denote
\begin{align}
	o_{mk} &= \sqrt{\omega_{mk}} \hat{u}_{mk}^\ast u_{mk} - \sqrt{\omega_{mk}} \mathbb{E}\{\hat{u}_{mk}^\ast u_{mk}\}, \\  o_{m'k} &= \sqrt{\omega_{m'k}} \hat{u}_{m'k}^\ast u_{m'k} -  \sqrt{\omega_{mk}} \mathbb{E}\{\hat{u}_{mk}^\ast u_{mk}\},
\end{align}
where $\omega_{mk}$ and $\omega_{m'k}$ are two non-negative deterministic scalars. Then, the following holds
\begin{equation} \label{eq:omkmk}
\mathbb{E}\{ o_{mk} o_{m'k}^{\ast} \} =   p \tau_p \sqrt{\omega_{mk} \omega_{m'k}} c_{mk} c_{m'k} \sum_{k' \in \mathcal{P}_k} \mathrm{tr}( \pmb{\Theta}_{mk} \pmb{\Theta}_{m'k'}  ).
\end{equation}
\end{corollary}
\begin{proof}
See Appendix~\ref{AppendixReal}.
\end{proof}
Both Lemma~\ref{lemma:ChannelProperty} and Corollary~\ref{CorollaryReal} indicate that the presence of an RIS makes the channel statistics more complex, compared to a conventional Cell-Free Massive MIMO system, due to the correlation among the propagation channels.  In the next sections, the analytical expression of the channel estimates in Lemma~\ref{lemma:ChannelEst} and the properties in Corollary~\ref{CorollaryReal} and Lemma~\ref{lemma:ChannelProperty} 
are employed for signal detection in the uplink and for beamforming in the downlink. Also, the same results are used to optimize the phase shifts of the RIS in order to minimize the channel estimation error and to evaluate the corresponding ergodic net throughput.
\vspace{-0.2cm}
\subsection{RIS Phase Shift Control and Optimization}\label{sec:PS_optimization}
\vspace{-0.1cm}
Channel estimation is a critical aspect in Cell-Free Massive MIMO. As discussed in previous text, in many scenarios, non-orthogonal pilots have to be used. This causes pilot contamination, which may significantly reduce the system performance. In this section, we design an RIS-assisted phase shift control scheme that is aimed to improve the quality of channel estimation. To this end, we introduce the normalized mean square error (NMSE) of the channel estimate of the user~$k$ at the AP~$m$, as follows
\begin{equation} \label{eq:NMSEmk}
\mathrm{NMSE}_{mk} = \frac{\mathbb{E}\{|e_{mk}|^2\}}{\mathbb{E}\{|u_{mk}|^2\}} =1 - \frac{p \tau_p \delta_{mk} }{p \tau_p \sum_{k' \in \mathcal{P}_{k}}\delta_{mk'} + 1},
\end{equation}
where the last equality is obtained from \eqref{eq:2Order} and \eqref{eq:EstError}. The NMSE is a suitable metric to evaluate the channel estimation quality and to measure the relative channel estimation error per AP. By definition, the NMSE lies in the range $[0, 1]$. In particular, the NMSE tends to zero if orthogonal pilot signals are used for every user and the pilot power is sufficiently large. In general, however, the NMSE is greater than zero if the $K$ users reuse the pilot signals, i.e., $\tau_p < K$, since $\mathrm{NMSE}_{mk} \rightarrow 1 - \delta_{mk}/(\sum_{k' \in \mathcal{P}_k} \delta_{mk'})$ as $p \rightarrow \infty$. We propose to optimize the phase shift matrix $\pmb{\Phi}$ of the RIS so as to minimize the total NMSE obtained from all the users and all the APs, as follows
\begin{equation} \label{Prob:NMSEk}
\begin{aligned}
& \underset{\{ \theta_n \} }{\mathrm{minimize}}
&&   \sum_{m=1}^M \sum_{k=1}^K \mathrm{NMSE}_{mk} \\
& \,\,\mathrm{subject \,to}
& & -\pi \leq \theta_n \leq \pi, \forall n.
\end{aligned}
\end{equation}
We emphasize that the optimal phase shifts obtained by solving the problem in~\eqref{Prob:NMSEk} are independent of the instantaneous CSI and depend only on the statistical CSI, i.e., the large-scale fading coefficients and the channel covariance matrices. Problem~\eqref{Prob:NMSEk} is a fractional program, whose globally-optimal solution is not simple to be obtained for an RIS with a large number of independently tunable elements. Nonetheless, in the special network setup where the direct links from the APs to the users
are weak enough to be negligible with respect to the RIS-assisted links, the optimal solution to problem~\eqref{Prob:NMSEk} is available in a closed-form expression as in Corollary~\ref{corollary:EqualPhase}.
\begin{corollary} \label{corollary:EqualPhase}
If the direct links are weak enough to be negligible and the RIS-assisted channels are spatially correlated as formulated in \eqref{eq:CovarMa}, the optimal minimizer of the optimization problem in \eqref{Prob:NMSEk} is $\theta_1 = \ldots = \theta_N$, i.e., the equal phase shift design is optimal.
\end{corollary}
\begin{proof}
See Appendix~\ref{appendix:CorMSEk}. 
\end{proof}
If the direct links are completely blocked and the spatial correlation model in \eqref{eq:CovarMa} holds, Corollary~\ref{corollary:EqualPhase} provides a simple but effective option to design the phase shifts of the RIS while ensuring the optimal estimation of the aggregated channels according to the sum-NMSE minimization criterion. Therefore, an efficient channel estimation protocol can be designed even in the presence of an RIS equipped with a large number of engineered scattering elements. The numerical results  in Section~\ref{Sec:NumRes}
show that the phase shifts design obtained in Corollary~\ref{corollary:EqualPhase} offers good gains in terms of net throughput even if the direct links cannot be completely ignored.
\begin{remark}\label{RemarkProb}
The proposed optimization method of the phase shifts of the RIS is based on the minimization of the sum-NMSE, and it is, therefore, based on improving the channel estimation quality. This is a critical objective in Massive MIMO systems, since improving the accuracy of channel estimation results in a noticeable enhancement of the uplink and downlink net throughput \cite{mai2018pilot,van2018joint}. If the direct links are not weak enough, the equal phase shift design is not optimal anymore, and the optimal solution to problem \eqref{Prob:NMSEk} may be obtained numerically. For example, one can compute the first-order derivative of the sum-NMSE in \eqref{Prob:NMSEk} with respect to each reflecting element, and the gradient descent algorithm may be utilized to obtain a locally-optimal solution of  \eqref{Prob:NMSEk} in an iterative manner. Another option would be to optimize the phase shifts of the RIS based on the maximization of the uplink or downlink ergodic net throughput. The solution of the corresponding optimization problem is, however, challenging and depends on whether the uplink or the downlink transmission phases are considered. Due to space limitations, therefore, we postpone this latter criterion for optimizing the phase shifts of the RIS to a future research work.
\end{remark}

\vspace*{-0.2cm}
\section{Uplink Data Transmission and Performance Analysis With MR Combining}\label{Sec:UL}
\vspace*{-0.1cm}
In this section, we first introduce a procedure to detect the uplink transmitted signals by capitalizing on the channel estimation method introduced in the previous section. Then, we derive an asymptotic closed-form expression of the ergodic net throughput.
\vspace*{-0.2cm}
\subsection{Uplink Data Transmission Phase}
\vspace*{-0.1cm}
In the uplink, all the $K$ users  transmit their data to the $M$ APs simultaneously. Specifically, the user~$k$ transmits a modulated symbol $s_k$ with $\mathbb{E}\{|s_k|^2\} =1$. This symbol is weighted by a power control factor $\sqrt{\eta_k}$, $0 \leq \eta_k \leq 1$, which enhances the spectral efficiency by, for example, compensating the near-far effects and mitigating the mutual interference among the users. Then, the received baseband signal, $y_{um} \in \mathbb{C},$ at the AP~$m$ is
\begin{equation} \label{eq:yum}
\begin{split}
y_{um} & = \sqrt{\rho_u} \sum_{k=1}^K \sqrt{\eta_k} \big(g_{mk} +  \mathbf{h}_{m}^H \pmb{\Phi} \mathbf{z}_{k} \big)s_k + w_{um} \\
&= \sqrt{\rho_u} \sum_{k=1}^K \sqrt{\eta_{k}} u_{mk} s_k + w_{um},
\end{split}
\end{equation}
where $\rho_u$ is the normalized uplink SNR of each data symbol, which is defined as the maximum transmit power divided by the noise variance, and $\rho_u \eta_k$ is the corresponding SNR of the user~$k$; and $w_{um} \sim \mathcal{CN}(0,1)$ is the normalized additive noise. For data detection, the MR combining method is used at the CPU, i.e., $\hat{u}_{mk}, \forall m,k,$ in \eqref{eq:ChannelEst} is employed to detect the data transmitted by the user~$k$. In mathematical terms, the corresponding decision statistic is
\begin{equation} \label{eq:ruk}
\begin{split}
r_{uk} &= \sum_{m=1}^M \hat{u}_{mk}^\ast y_{um}\\
&= \sqrt{\rho_u} \sum_{m=1}^M \sum_{k'=1}^K  \sqrt{\eta_{k'}} \hat{u}_{mk}^\ast u_{mk'} s_{k'} + \sum_{m=1}^M \hat{u}_{mk}^\ast w_{um}.
\end{split}
\end{equation}
Based on the observation $r_{uk}$, the uplink ergodic net throughput of the user~$k$ is analyzed in the next subsection. 
\vspace*{-0.2cm}
\subsection{Asymptotic Analysis of the Uplink Received Signal} \label{subsec:Asymul}
\vspace*{-0.1cm}
In the considered system model, the number of APs, $M$, and the number of tunable elements of the RIS, $N$, can be large. Therefore, we analyze the performance of two case studies: $(i)$ $N$ is fixed and $M$ is large; and $(ii)$ both $N$ and $M$ are large. The asymptotic analysis is conditioned upon a given setup of the CSI, which includes the large-scale fading coefficients, the covariance matrices, and the power utilized for the pilot and data transmission phases. To this end, the uplink weighted signal in \eqref{eq:ruk} is split into three terms based on the pilot reuse set $\mathcal{P}_k$, as follows
\begin{equation} \label{eq:rukv1}
\begin{split}
r_{uk} =&  \underbrace{\sqrt{\rho_u} \sum_{k' \in \mathcal{P}_k } \sum_{m=1}^M  \sqrt{\eta_{k'}} \hat{u}_{mk}^\ast u_{mk'} s_{k'}}_{\mathcal{T}_{k1}} + \\
& \underbrace{\sqrt{\rho_u} \sum_{k' \notin \mathcal{P}_k} \sum_{m=1}^M   \sqrt{\eta_{k'}} \hat{u}_{mk}^\ast u_{mk'} s_{k'}}_{\mathcal{T}_{k2}} + \underbrace{\sum_{m=1}^M \hat{u}_{mk}^\ast w_{um}}_{\mathcal{T}_{k3}},
\end{split}
\end{equation}
where $\mathcal{T}_{k1}$ accounts for the signals received from all the users in $\mathcal{P}_k$, and $\mathcal{T}_{k2}$ accounts for the mutual interference from the users that are assigned orthogonal pilot sequences. The impact of the additive noise obtained after applying MR combining is given by $\mathcal{T}_{k3}$. From \eqref{eq:ReceivedPilotv1}, \eqref{eq:umk}, and \eqref{eq:ChannelEst}, we obtain the following identity
\begin{equation} \label{eq:Termv1}
\begin{split}
& \sum_{m=1}^M \sqrt{\eta_{k'}} \hat{u}_{mk}^\ast u_{mk'}\\
& = \sum_{m=1}^M \sqrt{\eta_{k'}} c_{mk} u_{mk'} \left(  \sum_{k'' \in \mathcal{P}_k} \sqrt{p \tau_p} u_{mk''}^\ast + w_{pmk}^\ast \right) \\
& =  \sum_{m=1}^M \sqrt{p \tau_p \eta_{k'} }  c_{mk} |u_{mk'}|^2 + \sum_{k'' \in \mathcal{P}_k \setminus \{k' \} } \sum_{m=1}^M \sqrt{p\tau_p\eta_{k'}}   \\
& \quad \times c_{mk} u_{mk'}  u_{mk''}^{\ast} +  \sum_{m=1}^M \sqrt{\eta_{k'}} c_{mk} u_{mk'}   w_{pmk}^\ast.
\end{split}
\end{equation}

\subsubsection{Case I} $N$ is fixed and $M$ is large, i.e., $M \rightarrow \infty$. In this case, we divide both sides of \eqref{eq:Termv1} by $M$ and exploits Tchebyshev's theorem \cite{cramer2004random}\footnote{Let $X_1, \ldots, X_n$ be independent random variables such that $\mathbb{E}\{ X_i \} = \bar{x}_i$ and $\mathsf{Var}\{ X_i\} \leq c < \infty$. Then, Tchebyshev's theorem states $\frac{1}{n}\sum_{n'=1}^n X_{n'} \xrightarrow[n \rightarrow \infty]{P} \frac{1}{n} \sum_{n'} \bar{x}_{n'}.$} and \eqref{eq:2Order} to obtain
\begin{equation} \label{eq:Tchev1}
\frac{1}{M} \sum_{m=1}^M \sqrt{\eta_{k'}} \hat{u}_{mk}^\ast u_{mk'} \xrightarrow[M \rightarrow \infty ]{P} \frac{1}{M} \sum_{m=1}^M \sqrt{ p \tau_p \eta_{k'}}  c_{mk} \delta_{mk'},
\end{equation}
where $\xrightarrow{P}$ denotes the convergence in probability.\footnote{A sequence $\{ X_n \}$ of random variables converges in probability to the random variable $X$ if, for all $\epsilon > 0$, it holds that $\lim_{n \rightarrow \infty} \mathrm{Pr}(|X_n - X| > \epsilon ) = 0$, where $\mathrm{Pr}(\cdot)$ denotes the probability of an event.} Specifically, the second and third terms in \eqref{eq:Termv1} converge to zero due to the so-called favorable propagation conditions and since the aggregated channel and the noise are mutually independent \cite{van2021reconfigurable}. By inserting \eqref{eq:Tchev1} into the decision variable in \eqref{eq:rukv1}, we obtain the following deterministic value
\begin{equation} \label{eq:Asympt1}
\frac{1}{M}r_{uk} \xrightarrow[M\rightarrow \infty]{P} \frac{1}{M} \sum_{k' \in \mathcal{P}_k} \sum_{m=1}^M \sqrt{ p \tau_p \rho_u \eta_{k'} }  c_{mk} \delta_{mk'} s_{k'},
\end{equation}
because $\mathcal{T}_{k2}/M \rightarrow 0$ and $\mathcal{T}_{k3}/M \rightarrow 0 $ as $M \rightarrow \infty$. The result in \eqref{eq:Asympt1} unveils that, for a fixed $N$,  the channels become asymptotically orthogonal. In particular, the small-scale fading, the non-coherent interference, and the additive noise vanish. The only residual impairment is the pilot contamination caused by the users that employ the same pilot sequence. This result is the evidence that, due to pilot contamination, the system performance cannot be improved by adding more APs if MR combining is used. The contributions of both the direct and RIS-assisted indirect channels explicitly appear  in \eqref{eq:Asympt1} through the terms $\beta_{mk'}$ and $\mathrm{tr}(\pmb{\Theta}_{mk'})$, respectively.

\subsubsection{Case II}
Both $N$ and $M$ are large, i.e., $N\rightarrow \infty$ and $M \rightarrow \infty$. To analyze this case study, we need some assumptions on the covariance matrices $\mathbf{R}_m$ and $\widetilde{\mathbf{R}}_{k}$, as summarized as follows.
\begin{assumption} \label{Assumption1}
For $m= 1,\ldots,M$ and $k=1,\ldots,K,$  the covariance matrices $\mathbf{R}_m$ and $\widetilde{\mathbf{R}}_{k}$ are assumed to fulfill the following properties
\begin{align}
&\underset{N}{\limsup} \, \| \mathbf{R}_m\|_2 < \infty, \underset{N}{\liminf} \, \frac{1}{N} \mathrm{tr} ( \mathbf{R}_m) > 0, \label{eq:Asymp1}\\ 
&\underset{N}{\limsup} \, \| \widetilde{\mathbf{R}}_{k} \|_2 < \infty, \underset{N}{\liminf} \,  \frac{1}{N} \mathrm{tr} ( \widetilde{\mathbf{R}}_{k}) > 0. \label{eq:Asymp2v1}
\end{align}
\end{assumption}
The assumptions in \eqref{eq:Asymp1} and \eqref{eq:Asymp2v1} imply that the largest singular value and the sum of the eigenvalues (counted with their mutiplicity) of the $ N \times N $ covariance matrices that characterize the spatial correlation among the channels of the RIS elements are finite and positive. Dividing both sides of \eqref{eq:Termv1} by $MN$ and then applying Tchebyshev's theorem and \eqref{eq:2Order}, we obtain 
\begin{equation}\label{eq:AsymMNUL}
\frac{1}{MN} \sum_{m=1}^M \sqrt{\eta_{k'}} \hat{u}_{mk}^\ast u_{mk'} \xrightarrow[\substack{M \rightarrow \infty\\ N \rightarrow \infty} ]{P} \frac{1}{MN} \sum_{m=1}^M \sqrt{ p\tau_p \eta_{k'}}  c_{mk}  \mathrm{tr}( \pmb{\Theta}_{mk'}).
\end{equation}
We first observe that $\pmb{\Theta}_{mk'}$ is similar to $\widetilde{\mathbf{R}}_{k'}^{1/2} \pmb{\Phi} \mathbf{R}_{m} \pmb{\Phi}^H \widetilde{\mathbf{R}}_{k'}^{1/2}$, which is a positive semi-definite matrix.\footnote{Two matrices $\mathbf{A}$ and $\mathbf{B}$ of size $N \times N$ are similar if there exists an invertible $N \times N$ matrix $\mathbf{U}$ such that $\mathbf{B} = \mathbf{U}^{-1} \mathbf{A} \mathbf{U}$.} Since similar matrices have the same eigenvalues, it follows that $\mathrm{tr}( \pmb{\Theta}_{mk'}) >0$. Based on Assumption~\ref{Assumption1}, we obtain the following inequalities 
\begin{equation} \label{eq:traceConv}
\begin{split}
\frac{1}{N}\mathrm{tr}(\pmb{\Theta}_{mk'} ) &\stackrel{(a)}{\leq} \frac{1}{N} \| \pmb{\Phi} \|_2 \mathrm{tr}\big( \mathbf{R}_m \pmb{\Phi} \widetilde{\mathbf{R}}_{k'} \big) \\
&\stackrel{(b)}{=}  \frac{1}{N}\mathrm{tr}\big(  \pmb{\Phi} \widetilde{\mathbf{R}}_{k'} \mathbf{R}_m \big) \stackrel{(c)}{\leq} \frac{1}{N} \| \widetilde{\mathbf{R}}_{k'} \|_2 \mathrm{tr}(\mathbf{R}_{m}),
\end{split}
\end{equation}
where $(a)$ is obtained from Lemma~\ref{lemma:trace} in Appendix~\ref{Appendix:UsefulLemmas}; $(b)$ follows because $\|\pmb{\Phi}\|_2 = 1$; and  $(c)$ is obtained by applying again Lemma~\ref{lemma:trace}. Based on Assumption~\ref{Assumption1}, the last inequality in \eqref{eq:traceConv} is bounded by a positive constant. From \eqref{eq:AsymMNUL}, therefore, the decision variable in \eqref{eq:rukv1} can be formulated as
\begin{equation} \label{eq:Asymp2}
\frac{1}{MN}r_{uk} \xrightarrow[\substack{M\rightarrow \infty\\ N \rightarrow \infty}]{P} \frac{1}{MN} \sum_{k' \in \mathcal{P}_k} \sum_{m=1}^M \sqrt{p\tau_p \rho_u \eta_{k'}  }  c_{mk} \mathrm{tr}(\pmb{\Theta}_{mk'}) s_{k'},
\end{equation}
which is bounded from above thanks to \eqref{eq:traceConv}. The expression obtained in \eqref{eq:Asymp2} reveals that, as $M,N \rightarrow \infty$, the post-processed signal at the CPU consists of the desired signal of the intended user $k$ and the interference from the other users in $\mathcal{P}_k$. Compared with \eqref{eq:Asympt1}, we observe that \eqref{eq:Asymp2} is independent of the direct links and depends only on the RIS-assisted indirect links. This highlights the potentially promising contribution of the RIS, in the limiting regime $M,N \to \infty$, for enhancing the system performance.
\vspace*{-0.5cm}
\subsection{Uplink Ergodic Net Throughput with a Finite Number of APs and RIS Elements}
\vspace*{-0.15cm}
In this section, we focus our attention on the practical setup in which $M$ and $N$ are both finite. By utilizing the user-and-then forget channel capacity bounding method \cite{Marzetta2016a}, the uplink ergodic net throughput of the user $k$ can be written as follows
\begin{equation} \label{eq:ULRate}
R_{uk} = B \nu_u \left( 1 - \tau_p/\tau_c \right) \log_2 \left( 1 + \mathrm{SINR}_{uk} \right), \mbox{[Mbps]},
\end{equation}
where $B$ is the system bandwidth measured in MHz and $0\leq \nu_u \leq 1$ is the portion of each coherence time interval that is dedicated to the uplink data transmission phase. The effective uplink signal-to-noise-plus-interference ratio (SINR), which is denoted by $\mathrm{SINR}_{uk}$, is defined as follows
\begin{multline} \label{eq:ULSINR}
\mathrm{SINR}_{uk} = \\
 \frac{|\mathsf{DS}_{uk}|^2}{\mathbb{E} \{|\mathsf{BU}_{uk}|^2 \} + \sum_{k'=1, k' \neq k}^K \mathbb{E} \{|\mathsf{UI}_{uk'k}|^2 \} + \mathbb{E}\{ |\mathsf{NO}_{uk}|^2 \}},
\end{multline}
where the following definitions hold
\begin{align}
&\mathsf{DS}_{uk} = \sqrt{\rho_u\eta_k} \mathbb{E}  \left\{ \sum_{m=1}^M  \hat{u}_{mk}^\ast u_{mk} \right\}, \\
& \mathsf{BU}_{uk} = \sqrt{\rho_u\eta_k} \left( \sum_{m=1}^M  \hat{u}_{mk}^\ast u_{mk}  - \mathbb{E}  \left\{ \sum_{m=1}^M  \hat{u}_{mk}^\ast u_{mk}  \right\} \right), \\
&\mathsf{UI}_{uk'k} = \sqrt{\rho_u \eta_{k'}} \sum_{m=1}^M \hat{u}_{mk}^\ast u_{mk'}, \\
& \mathsf{NO}_{uk} = \sum_{m=1}^M \hat{u}_{mk}^\ast w_{um}.
\end{align}
In particular, $\mathsf{DS}_{uk}$ denotes the (average) strength of the desired signal, $\mathsf{BU}_{uk}$ denotes the beamforming uncertainty, which represents the randomness of the effective channel gain for a given beamforming method, $\mathsf{UI}_{uk'k}$ denotes the interference caused by the user~$k'$ to the user~$k$, and $\mathsf{NO}_{uk}$ denotes the additive noise. We emphasize that the net throughput in \eqref{eq:ULRate} is achievable since it is a lower bound of the channel capacity. A closed-form expression for \eqref{eq:ULRate} is given in Theorem~\ref{theorem:ULMR}. 
\setcounter{eqnback}{\value{equation}} \setcounter{equation}{44}
\begin{figure*}
	\begin{equation} \label{eq:SINRULv1}
		\mathrm{SINR}_{uk} = \frac{\rho_u \eta_k \left( \sum\limits_{m=1}^M \gamma_{mk}^o \right)^2}{ \sum\limits_{m=1}^M \gamma_{mk}^o +  \rho_u \sum\limits_{k'=1}^K \sum\limits_{m=1}^M \eta_{k'} \gamma_{mk}^o \delta_{mk'} + p \tau_p \rho_u \sum\limits_{k' =1 }^K\sum\limits_{m=1}^M  \sum\limits_{m'=1}^M\eta_{k'} c_{mk}^o c_{m'k}^o \mathrm{tr}( \pmb{\Theta}_{mk'} \pmb{\Theta}_{m'k}) +  p \tau_p \rho_u \eta_{k} \sum\limits_{m=1}^M   c_{mk}^2 \mathrm{tr}(\pmb{\Theta}_{mk}^2)}
	\end{equation}
	\hrule
	\vspace*{-0.4cm}
\end{figure*}
\setcounter{eqncnt}{\value{equation}}
\setcounter{equation}{\value{eqnback}}
\begin{theorem} \label{theorem:ULMR}
If the CPU utilizes the MR combining method, a lower-bound closed-form expression for the uplink net throughput of the user~$k$ is given by \eqref{eq:ULRate}, where the SINR is
\begin{equation} \label{eq:ClosedFormSINR}
\mathrm{SINR}_{uk} = \frac{\rho_u \eta_{k} \left( \sum_{m=1}^M \gamma_{mk} \right)^2}{\mathsf{MI}_{uk} + \mathsf{NO}_{uk}},
\end{equation}
where $\mathsf{MI}_{uk}$ is the mutual interference and $\mathsf{NO}_{uk}$ is the noise, which are formulated as follows
\begin{align} \label{eq:MIuk}
&\mathsf{MI}_{uk} = \rho_u \sum_{k'=1}^K \sum_{m=1}^M \eta_{k'} \gamma_{mk} \delta_{mk'} \notag \\
&+ p \tau_p \rho_u \sum_{k' =1 }^K \sum_{k'' \in \mathcal{P}_k}    \sum_{m=1}^M  \sum_{m'=1}^M\eta_{k'} c_{mk}c_{m'k} \mathrm{tr}( \pmb{\Theta}_{mk'} \pmb{\Theta}_{m'k''}) \\
 &+  p \tau_p \rho_u \sum_{k' \in  \mathcal{P}_k} \sum_{m=1}^M  \eta_{k'} c_{mk}^2 \mathrm{tr}(\pmb{\Theta}_{mk'}^2) \notag \\
 & + p \tau_p \rho_u \sum_{k' \in \mathcal{P}_k \setminus \{ k\} } \eta_{k'} \left(\sum_{m=1}^M c_{mk} \delta_{mk'} \right)^2,\notag\\
& \mathsf{NO}_{uk} =\sum_{m=1}^M \gamma_{mk},
\end{align}
with $\delta_{mk'} =\beta_{mk'} + \mathrm{tr}( \pmb{\Theta}_{mk'} )$, $c_{mk}$ given in \eqref{eq:cmk}, and $\gamma_{mk}$ given in \eqref{eq:gammamk}.
\end{theorem}
\begin{proof}
See Appendix~\ref{appendix:ULMR}.
\end{proof}
By direct inspection of the SINR in \eqref{eq:ClosedFormSINR}, we observe that the numerator increases with the square of the sum of the variances of the channel estimates, $\gamma_{mk},\forall m$ thanks to the joint signal processing, i.e., the received signals form the $M$ APs are sent to the CPU for centralized data detection. On the other hand, the first term in \eqref{eq:MIuk} represents the power of the mutual interference. The use of an RIS to support multiple users introduce the extra interference shown in the second and third terms in \eqref{eq:MIuk}. Due to the limited and finite number of orthogonal pilot sequences being used, the fourth term in \eqref{eq:MIuk} dominates the impact of pilot contamination. The second term in the denominator in \eqref{eq:ClosedFormSINR} is the additive noise. If the coherence time is sufficiently large that every user can utilize its own orthogonal pilot sequence, the uplink net throughput of the user~$k$ can still be obtained from \eqref{eq:ULRate}, but the effective SINR simplifies to \eqref{eq:SINRULv1}. The SINR in \eqref{eq:ClosedFormSINR} is a multivariate function of the matrix of phase shifts of the RIS and of the channel statistics, i.e., the channel covariance matrices. Table~\ref{Table:CompareUL} gives a comparison of the obtained uplink SINR  of the user $k$ with and without the presence of the RIS. By direct inspection of $|\mathsf{DS}_{uk}|^2$, we evince that the strength of the desired signal increases thanks to the assistance of the RIS. However, the mutual interference becomes more severe as well, due to the need of estimating both the direct and indirect links in the presence of the RIS. By assigning orthogonal pilot signals to all the $K$ users, the coherent interference can be completely suppressed. In Section~\ref{Sec:NumRes}, the performance of Cell-Free Massive MIMO and RIS-assisted Cell-Free Massive MIMO is compared with the aid of numerical simulations.

\begin{table*}[t]
	\caption{Comparison of the uplink SINR between Cell-Free Massive MIMO and RIS-Assisted Cell-Free Massive MIMO} \label{Table:CompareUL}
	\centering
\begin{tabular}{|c|c|c|c|c|c|c|c|c|}
	\hline
	\multicolumn{2}{|c|}{Uplink SINR}  & Cell-Free Massive MIMO    & RIS-Assisted Cell-Free Massive MIMO  \\ 
	\hline
	\multirow{9}{*}{\eqref{eq:ClosedFormSINR}} &  $\delta_{mk}$ & $\beta_{mk}$ & $\beta_{mk} + \mathrm{tr}\big( \pmb{\Phi}^H \mathbf{R}_{m} \pmb{\Phi} \widetilde{\mathbf{R}}_{k} \big)$      \\  
	\cline{2-4}
	&  $c_{mk}$ & $\frac{\sqrt{p\tau_p}  \beta_{mk}}{p\tau_p \sum\limits_{k' \in \mathcal{P}_k}  \beta_{mk'} + 1}$ & $\frac{\sqrt{p\tau_p} \delta_{mk} }{p\tau_p \sum\limits_{k' \in \mathcal{P}_k} \delta_{mk'} + 1}$   \\  
	\cline{2-4}
	&  $\gamma_{mk}$ & $ \sqrt{p\tau_p} \beta_{mk}c_{mk}$ & $ \sqrt{p\tau_p}  \delta_{mk} c_{mk}$       \\  
	\cline{2-4}
	& $|\mathsf{DS}_{uk}|^2$ & $\rho_u \eta_{k} \left( \sum\limits_{m=1}^M \gamma_{mk} \right)^2$ & $\rho_u \eta_{k} \left( \sum\limits_{m=1}^M \gamma_{mk} \right)^2$\\
	\cline{2-4}
	& $\mathsf{MI}_{uk}$ & \makecell{$\rho_u \sum\limits_{k'=1}^K \sum\limits_{m=1}^M \eta_{k'} \gamma_{mk} \delta_{mk'} +$ \\$ p \tau_p \rho_u \sum\limits_{k' \in \mathcal{P}_k \setminus \{ k\} } \eta_{k'} \left(\sum\limits_{m=1}^M c_{mk} \delta_{mk'} \right)^2$} & \makecell{$\rho_u \sum\limits_{k'=1}^K \sum\limits_{m=1}^M \eta_{k'} \gamma_{mk} \delta_{mk'} + $ \\ $p \tau_p \rho_u \sum\limits_{k' =1 }^K \sum\limits_{k'' \in \mathcal{P}_k}    \sum\limits_{m=1}^M  \sum\limits_{m'=1}^M\eta_{k'} c_{mk}c_{m'k} \mathrm{tr}( \pmb{\Theta}_{mk'} \pmb{\Theta}_{m'k''}) + $ \\ $ p \tau_p \rho_u \sum\limits_{k' \in  \mathcal{P}_k} \sum\limits_{m=1}^M  \eta_{k'} c_{mk}^2 \mathrm{tr}(\pmb{\Theta}_{mk'}^2) + p \tau_p \rho_u \sum_{k' \in \mathcal{P}_k \setminus \{ k\} } \eta_{k'} \left(\sum\limits_{m=1}^M c_{mk} \delta_{mk'} \right)^2$} \\
	\cline{2-4}
	& $\mathsf{NO}_{uk}$ & $ \sum\limits_{m=1}^M \gamma_{mk}$ & $ \sum\limits_{m=1}^M \gamma_{mk}$\\
	\cline{1-4}
	\multirow{6}{*}{\eqref{eq:SINRULv1}} &  $c_{mk}^o$ & $\frac{\sqrt{p\tau_p}  \beta_{mk}}{p\tau_p   \beta_{mk} + 1}$ & $\frac{\sqrt{p\tau_p} \delta_{mk} }{p\tau_p \delta_{mk} + 1}$  \\  
	\cline{2-4}
	&  $\gamma_{mk}^o$ & $ \sqrt{p\tau_p} \beta_{mk}c_{mk}^o$ & $ \sqrt{p\tau_p}  \delta_{mk} c_{mk}^o$  \\  
	\cline{2-4}
	&
	$|\mathsf{DS}_{uk}|^2$ & $\rho_u \eta_{k} \left( \sum\limits_{m=1}^M \gamma_{mk}^o \right)^2$ & $\rho_u \eta_{k} \left( \sum\limits_{m=1}^M \gamma_{mk}^o \right)^2$ \\
	\cline{2-4}
	& $\mathsf{MI}_{uk}$ & \makecell{ $ \rho_u  \sum\limits_{k'=1}^{K} \sum\limits_{m=1}^M \eta_{k'} \gamma_{mk}^o \delta_{mk'} + \sum\limits_{m=1}^M \gamma_{mk}^o$} &   \makecell{$ \rho_u \sum\limits_{k'=1}^K \sum\limits_{m=1}^M \eta_{k'} \gamma_{mk}^o \delta_{mk'} + p \tau_p \rho_u \sum\limits_{k' =1 }^K\sum\limits_{m=1}^M  \sum\limits_{m'=1}^M\eta_{k'} c_{mk}^o c_{m'k}^o \mathrm{tr}( \pmb{\Theta}_{mk'} \pmb{\Theta}_{m'k}) + $ \\ $  p \tau_p \rho_u \eta_{k} \sum\limits_{m=1}^M   c_{mk}^2 \mathrm{tr}(\pmb{\Theta}_{mk}^2)$} \\  
	\cline{2-4}
	& $\mathsf{NO}_{uk}$ & $ \sum\limits_{m=1}^M \gamma_{mk}^o$ & $ \sum\limits_{m=1}^M \gamma_{mk}^o$\\
	\cline{1-4}
\end{tabular}
\vspace*{-0.2cm}
\end{table*}

\vspace*{-0.2cm}
\section{Downlink Data Transmission and Performance Analysis With MR Precoding} \label{Sec:Downlink}
\vspace*{-0.1cm}
In this section, we consider the downlink data transmission phase and analyze the received signal at the users when the number of APs is large or when the numbers of RIS elements and APs are both large. A closed-form expression of the downlink ergodic net throughput that is attainable with MR precoding and for an arbitrary phase shift matrix of the RIS elements is provided.
\vspace*{-0.2cm}
\subsection{Downlink Data Transmission Phase}
\vspace*{-0.1cm}
By exploiting channel reciprocity, the AP~$m$ treats the channel estimates obtained in the uplink as the true channels in order to construct the beamforming coefficients. Accordingly, the downlink signal transmitted from  the AP~$m$ is\footnote{In this paper, the downlink data transmission is conducted based on the uplink channel estimates, which depend on the RIS phase shift matrix. Since the MR processing is used for downlink data transmission based on the uplink channel estimates, closed-form analytical expressions for the downlink can be obtained if the same phase shift matrix is  utilized in the uplink and in the downlink.}
\setcounter{eqnback}{\value{equation}} \setcounter{equation}{45}
\begin{equation} \label{eq:TransSig}
x_m = \sqrt{\rho_{d}}\sum_{k=1}^K \sqrt{\eta_{mk}} \hat{u}_{mk}^\ast q_k,
\end{equation}
where $\rho_d$ is the normalized SNR in the downlink; $q_k$ is the complex data symbol that is to be sent (cooperatively by virtue of the considered coherent joint transmission scheme) by all the $M$ APs to the user~$k$, with $\mathbb{E}\{ |q_k|^2 \} = 1$; and $\eta_{mk}$ is the power control coefficient of  the AP~$m$, which  satisfies the power budget constraint as follows
\begin{equation} \label{eq:PowerConst}
\mathbb{E} \{|x_m|^2\} \leq \rho_d \Rightarrow 
\sum_{k=1}^K \eta_{mk} \gamma_{mk} \leq 1.
\end{equation} 
The cooperation among the $M$ APs for jointly transmitting the same data symbol to a particular user creates the major distinction between the downlink and uplink data transmission phases. Based on \eqref{eq:TransSig}, the received signal at the user~$k$ is the superposition of the signals transmitted by the $M$ APs as
\begin{equation} \label{eq:rdk}
\begin{split}
r_{dk} &= \sum_{m=1}^M u_{mk} x_m + w_{dk} \\
&= \sqrt{\rho_{d}} \sum_{m=1}^M \sum_{k'=1}^K \sqrt{\eta_{mk'}} u_{mk} \hat{u}_{mk'}^\ast q_{k'} + w_{dk}.
\end{split}
\end{equation}
where $w_{dk} \sim \mathcal{CN}(0,1)$ is the additive noise at the user $k$. The user~$k$ decodes the desired data symbol based on the observation in \eqref{eq:rdk}.
\vspace*{-0.2cm}
\subsection{Asymptotic Analysis of the Downlink Received Signal}
\vspace*{-0.1cm}
In contrast with the uplink data processing where the CPU needs only the channel estimate $\hat{u}_{mk}$ for detecting the data of the user~$k$, as displayed in \eqref{eq:rukv1}, the received signal in \eqref{eq:rdk} depends on the channel estimates of the $K$ users in the network, since the channel estimates from the $K$ users are used for MR precoding. Therefore, the analysis of the uplink and downlink data transmission phases are different. First, we split \eqref{eq:rdk} into three terms, by virtue of the pilot reuse pattern $\mathcal{P}_k$, as follows
\begin{multline} \label{eq:rdkv1}
r_{dk} = \sqrt{\rho_{d}} \sum_{k'\in \mathcal{P}_k} \sum_{m=1}^M  \sqrt{\eta_{mk'}} u_{mk} \hat{u}_{mk'}^\ast q_{k'} + \\
\sqrt{\rho_{d}} \sum_{k'\notin \mathcal{P}_k } \sum_{m=1}^M  \sqrt{\eta_{mk'}} u_{mk} \hat{u}_{mk'}^\ast q_{k'} +w_{dk}.
\end{multline}
Then, we investigate the two asymptotic regimes for $M \to \infty$ and $M,N \to \infty$. In particular, the first term in \eqref{eq:rdkv1} can be rewritten as
\begin{equation} \label{eq:rdkAsymtotic}
\begin{split}
& \sum_{m=1}^M \sqrt{\eta_{mk'}} u_{mk} \hat{u}_{mk'}^\ast\\
& \stackrel{(a)}{=} \sum_{m=1}^M \sqrt{\eta_{mk'}} c_{mk'} u_{mk} \left(  \sum_{k'' \in \mathcal{P}_k} \sqrt{p\tau_p}  u_{mk''}^\ast + w_{pmk'}^\ast \right) \\
&\stackrel{(b)}{=} \sum_{m=1}^M \sqrt{\eta_{mk'}p\tau_p}  c_{mk'} |u_{mk}|^2 +\sum_{k'' \in \mathcal{P}_k \setminus \{k\} } \sum_{m=1}^M \sqrt{\eta_{mk'}p \tau_p}    \\
&\quad \times c_{mk'}u_{mk}  u_{mk''}^\ast + \sum_{m=1}^M \sqrt{\eta_{mk'}} c_{mk'} u_{mk} w_{pmk'}^\ast,
\end{split}
\end{equation}
where $(a)$ is obtained by utilizing the channel estimates in \eqref{eq:ChannelEst} and $(b)$ is obtained by extracting the aggregated channel of the user~$k$ from the summation. By letting $M$ and/or $N$ be large, similar to the uplink data transmission phase, we obtain the following asymptotic results
\begin{equation}
\frac{1}{M} \sum_{m=1}^M \sqrt{\eta_{mk'}} u_{mk} \hat{u}_{mk'}^\ast \xrightarrow[M \rightarrow \infty]{P} \frac{1}{M} \sum_{m=1}^M \sqrt{\eta_{mk'} p\tau_p}  c_{mk'} \delta_{mk},
\end{equation}
\begin{multline}
\frac{1}{MN} \sum_{m=1}^M \sqrt{\eta_{mk'}} u_{mk} \hat{u}_{mk'}^\ast \xrightarrow[\substack{M \rightarrow \infty \\
N \rightarrow \infty}]{P} \\ \frac{1}{MN} \sum_{m=1}^M \sqrt{\eta_{mk'} p \tau_p}  c_{mk'}  \mathrm{tr} ( \pmb{\Theta}_{mk} ),
\end{multline}
which are bounded from above based on Assumption~\ref{Assumption1}. Consequently, the received signal at the user~$k$ converges to a deterministic equivalent as the number of APs is large, i.e., $M \rightarrow \infty$, and as the number of APs and RIS elements are large, i.e., $M,N \rightarrow \infty$. More precisely, the received signal converges (asymptotically) to
\begin{align}
&\frac{1}{M} r_{dk} \xrightarrow[M \rightarrow \infty]{P} \frac{1}{M} \sum_{k' \in \mathcal{P}_k} \sum_{m=1}^M \sqrt{ p \tau_p \rho_d \eta_{mk'} }  c_{mk'} \delta_{mk}s_{k'},\label{eq:DeEquiv1}\\
&\frac{1}{MN} r_{dk} \xrightarrow[M \rightarrow \infty]{P} \frac{1}{MN} \sum_{k' \in \mathcal{P}_k} \sum_{m=1}^M \sqrt{p \tau_p \rho_d \eta_{mk'} }  c_{mk'} \mathrm{tr}( \pmb{\Theta}_{mk} )s_{k'},\label{eq:DeEquiv2}
\end{align}
which indicates the inherent coexistence of the users in $\mathcal{P}_k$. The deterministic equivalents in \eqref{eq:DeEquiv1} and \eqref{eq:DeEquiv2} unveil that the impact of the channel estimation accuracy and the channel statistics is different between the uplink and the downlink. In particular, the asymptotic received signal in the uplink only depends on the channel estimation quality of each individual user, which is manifested by the coefficient $c_{mk}$. The asymptotic received signal in the downlink depends, on the other hand, on the channel estimation quality of all the users that share the same orthogonal pilot sequences, i.e.,  $c_{m,k'}, \forall k' \in \mathcal{P}_{k}$.
\setcounter{eqnback}{\value{equation}} \setcounter{equation}{61}
\begin{figure*}
	\begin{equation} \label{eq:DLOrt}
		\fontsize{9}{9}{\mathrm{SINR}_{dk} = \frac{\rho_{d} \left(   \sum\limits_{m=1}^M  \sqrt{\eta_{mk}} \gamma_{mk} \right)^2}{1+\rho_d \sum\limits_{k' =1}^K \sum\limits_{m=1}^M \eta_{mk'}  \gamma_{mk'}^o \delta_{mk} +  p \tau_p \rho_d \sum\limits_{k' =1}^K \sum\limits_{m=1}^M \sum\limits_{m'=1}^M \sqrt{\eta_{mk'} \eta_{m'k'}} c_{mk'}  c_{m'k'} \mathrm{tr} (\pmb{\Theta}_{mk} \pmb{\Theta}_{m'k'} ) + p \tau_p \rho_d  \sum\limits_{m=1}^M \eta_{mk}  c_{mk}^2 \mathrm{tr}(\pmb{\Theta}_{mk}^2)}}
	\end{equation}
	\hrule
	\vspace{-0.3cm}
\end{figure*}
\setcounter{eqncnt}{\value{equation}}
\setcounter{equation}{\value{eqnback}}
\vspace*{-0.2cm}
\subsection{Downlink Ergodic Net Throughput with  a Finite Number of APs and RIS Elements}
\vspace*{-0.1cm} 
By utilizing the channel capacity bounding technique \cite{Marzetta2016a}, similar to the analysis of the uplink data transmission phase,  the downlink ergodic net throughput of the user $k$ can be written as follows
\begin{equation} \label{eq:DLRate}
R_{dk} = B\nu_d \left( 1- \tau_p/\tau_c \right) \log_2 \left( 1 + \mathrm{SINR}_{dk} \right),\mbox{[Mbps]},
\end{equation}
where $0\leq \nu_d \leq 1$ is the portion of each coherence time interval dedicated to the downlink data transmission phase, with $\nu_u + \nu_d = 1$, and the effective downlink SINR is defined as
\begin{equation} \label{eq:DLSINR}
\mathrm{SINR}_{dk} = \frac{|\mathsf{DS}_{dk}|^2}{\mathbb{E} \{|\mathsf{BU}_{dk}|^2 \} + \sum_{k'=1, k' \neq k}^K \mathbb{E} \{|\mathsf{UI}_{dk'k}|^2 \} + 1},
\end{equation}
where the following definitions hold
\begin{align}
& \mathsf{DS}_{dk} = \sqrt{\rho_{d}}  \mathbb{E}\left\{ \sum_{m=1}^M  \sqrt{\eta_{mk}} u_{mk}\hat{u}_{mk}^\ast \right\}, \\
& \mathsf{UI}_{dk'k} = \sqrt{\rho_{d}} \sum_{m=1}^M \sqrt{\eta_{mk'}} u_{mk}\hat{u}_{mk'}^\ast,\\
& \mathsf{BU}_{dk}  =   \sqrt{\rho_{d}} \left( \sum_{m=1}^M  \sqrt{\eta_{mk}} u_{mk}\hat{u}_{mk}^\ast -  \mathbb{E}\left\{ \sum_{m=1}^M  \sqrt{\eta_{mk}} u_{mk}\hat{u}_{mk}^\ast \right\}  \right),
\end{align}
In particular, $\mathsf{DS}_{dk}$ denotes the (average) strength of the desired signal received by the user~$k$, $\mathsf{BU}_{dk}$ denotes the beamforming uncertainty, $\mathsf{UI}_{dk'k}$ denotes the interference caused to the user~$k$ by the signal intended to the user~$k'$. The downlink ergodic net throughput in \eqref{eq:DLRate} is achievable since it is a lower bound of the channel capacity, similar to the uplink data transmission phase. In contrast to the uplink ergodic net throughput, which only depends on the combining coefficients of each individual user, the downlink net throughput of the user~$k$ depends on the precoding coefficients of all the $K$ users. A closed-form expression for \eqref{eq:DLRate} is given in Theorem~\ref{theorem:DLMR}.
\begin{theorem} \label{theorem:DLMR}
If the CPU utilizes the MR precoding method, a lower-bound closed-form expression for the downlink net throughput of the user~$k$ is given by \eqref{eq:DLRate}, where the SINR is
\begin{equation} \label{eq:DLSINRMRT}
\mathrm{SINR}_{dk} = \frac{\rho_{d} \big(   \sum_{m=1}^M  \sqrt{\eta_{mk}} \gamma_{mk} \big)^2}{\mathsf{MI}_{dk} + 1},
\end{equation}
where $\mathsf{MI}_{dk}$ is the mutual interference, which is defined as follows
\begin{equation} \label{eq:MIdk}
\begin{split}
 &  \mathsf{MI}_{dk}=  \rho_d \sum_{k' =1}^K \sum_{m=1}^M \eta_{mk'}  \gamma_{mk'} \delta_{mk} +  p \tau_p \rho_d \times \\
& \sum_{k' =1}^K \sum_{k'' \in \mathcal{P}_{k'}} \sum_{m=1}^M \sum_{m'=1}^M \sqrt{\eta_{mk'} \eta_{m'k'}} c_{mk'}  c_{m'k'} \mathrm{tr} (\pmb{\Theta}_{mk} \pmb{\Theta}_{m'k''} ) \\
&+ p \tau_p \rho_d  \sum_{k' \in \mathcal{P}_k } \sum_{m=1}^M \eta_{mk'}  c_{mk'}^2 \mathrm{tr}(\pmb{\Theta}_{mk}^2) + \\
& p \tau_p \rho_d  \sum_{k' \in \mathcal{P}_k \setminus \{ k\}} \left(\sum_{m=1}^M  \sqrt{\eta_{mk'} } c_{mk'}  \delta_{mk} \right)^2.
\end{split}
\end{equation}
\end{theorem}
\begin{proof}
The main steps of the proof are similar to those of the proof of Theorem~\ref{theorem:ULMR}. However, there are also major differences that are due to the coherent joint data transmission scheme among the APs. The details of the proof are available in Appendix~\ref{appendix:DLMR}.
\end{proof}
From Theorem~\ref{theorem:DLMR}, we observe that the effective downlink SINR has some similarities and differences as compared with its uplink counterpart in Theorem~\ref{theorem:ULMR}.  Similar to the uplink, the numerator of \eqref{eq:DLSINRMRT} is a quadratic function that depends on the channel estimation quality and the coherent joint transmission processing. Differently from the uplink, the transmit power coefficients appear explicitly in the numerator of \eqref{eq:DLSINRMRT} as a result of the cooperation among the APs. The  impact of pilot contamination in \eqref{eq:MIdk}, in contrast to the uplink, depends on all the transmit power coefficients. In addition, we observe that the impact of pilot contamination scales up with the number of APs and with the number of elements of the RIS. When the $K$ users employ orthogonal pilot sequences, the downlink SINR expression of the user~$k$ simplifies to \eqref{eq:DLOrt}.

\begin{table*}[t]
	\caption{Comparison of the Downlink SINR between Cell-Free Massive MIMO and RIS-Assisted Cell-Free Massive MIMO (some parameters are defined in Table~\ref{Table:CompareUL})} \label{Table:CompareDL}
	\centering
     \begin{tabular}{|c|c|c|c|c|c|c|c|c|}
		\hline
		\multicolumn{2}{|c|}{Downlink SINR}  & Cell-Free Massive MIMO    & RIS-Assisted Cell-Free Massive MIMO  \\ 
		\hline
		\multirow{5}{*}{\eqref{eq:DLSINRMRT}} 
		& $|\mathsf{DS}_{dk} |^2$ & $\rho_{d} \left(   \sum\limits_{m=1}^M  \sqrt{\eta_{mk}} \gamma_{mk} \right)^2$ & $\rho_{d} \left(   \sum\limits_{m=1}^M  \sqrt{\eta_{mk}} \gamma_{mk} \right)^2$ \\
		\cline{2-4}
		&  $\mathsf{MI}_{dk}$ & \makecell{$\rho_d \sum\limits_{k' =1}^K \sum\limits_{m=1}^M \eta_{mk'}  \gamma_{mk'} \delta_{mk} +$ \\ $ p \tau_p \rho_d  \sum\limits_{k' \in \mathcal{P}_k \setminus \{ k\}} \left(\sum\limits_{m=1}^M  \sqrt{\eta_{mk'} } c_{mk'}  \delta_{mk} \right)^2$} & \makecell{$\rho_d \sum\limits_{k' =1}^K \sum\limits_{m=1}^M \eta_{mk'}  \gamma_{mk'} \delta_{mk} + $ \\ $ p \tau_p \rho_d \sum\limits_{k' =1}^K \sum\limits_{k'' \in \mathcal{P}_{k'}} \sum\limits_{m=1}^M \sum\limits_{m'=1}^M \sqrt{\eta_{mk'} \eta_{m'k'}} c_{mk'}  c_{m'k'} \mathrm{tr} (\pmb{\Theta}_{mk} \pmb{\Theta}_{m'k''} ) + $\\$ p \tau_p \rho_d  \sum\limits_{k' \in \mathcal{P}_k } \sum\limits_{m=1}^M \eta_{mk'}  c_{mk'}^2 \mathrm{tr}(\pmb{\Theta}_{mk}^2) + p \tau_p \rho_d  \sum\limits_{k' \in \mathcal{P}_k \setminus \{ k\}} \left(\sum\limits_{m=1}^M  \sqrt{\eta_{mk'} } c_{mk'}  \delta_{mk} \right)^2$} \\  
		\cline{2-4}
		& $\mathsf{NO}_{dk}$ & 1 & 1 \\
		\cline{1-4}
		\multirow{5}{*}{\eqref{eq:DLOrt}}
		& $|\mathsf{DS}_{dk} |^2$ & $\rho_{d} \left(   \sum\limits_{m=1}^M  \sqrt{\eta_{mk}} \gamma_{mk}^o \right)^2$ & $\rho_{d} \left(   \sum\limits_{m=1}^M  \sqrt{\eta_{mk}} \gamma_{mk}^o \right)^2$ \\
		\cline{2-4}
		& $\mathsf{MI}_{dk}$ &  \makecell{$  \rho_d \sum\limits_{k' =1}^K \sum\limits_{m=1}^M \eta_{mk'} \gamma_{mk'}^o \delta_{mk} +1$} &  \makecell{$\rho_d \sum\limits_{k' =1}^K \sum\limits_{m=1}^M \eta_{mk'}  \gamma_{mk'}^o \delta_{mk} + $ \\ $ p \tau_p \rho_d \sum\limits_{k' =1}^K \sum\limits_{m=1}^M \sum\limits_{m'=1}^M \sqrt{\eta_{mk'} \eta_{m'k'}} c_{mk'}  c_{m'k'} \mathrm{tr} (\pmb{\Theta}_{mk} \pmb{\Theta}_{m'k'} ) +$ \\ $p \tau_p \rho_d  \sum\limits_{m=1}^M \eta_{mk}  c_{mk}^2 \mathrm{tr}(\pmb{\Theta}_{mk}^2)$}    \\  
		\cline{2-4}
		& $\mathsf{NO}_{dk}$ & $1$ & $1$\\
		\cline{1-4}
	\end{tabular}
	\vspace*{-0.2cm}
\end{table*} 
A comparison of the obtained analytical expressions of the downlink SINR for Cell-Free Massive MIMO and RIS-assisted Cell-Free Massive MIMO systems is given in Table~\ref{Table:CompareDL}. By comparing Table~\ref{Table:CompareUL} and Table~\ref{Table:CompareDL}, the difference and similarities between the uplink and downlink transmission phases can be identified as well. With the aid of numerical results, in Section~\ref{Sec:NumRes}, we will illustrate the advantages of RIS-assisted Cell-Free Massive MIMO especially if the direct links are not sufficiently reliable (e.g., they are blocked) with high probability.
\begin{remark} \label{RemarkPhase}
We observe that the ergodic net throughput in the uplink (Theorem~\ref{theorem:ULMR}) and downlink (Theorem~\ref{theorem:DLMR}) data transmission phases depend only on the large-scale fading statistics and on the channel covariance matrices, while they are independent of the instantaneous CSI. This simplifies the deployment and optimization of RIS-assisted Cell-Free Massive MIMO systems. As anticipated in Remark~\ref{RemarkProb}, in fact, the phase shifts of the RIS can be optimized based on the (long-term) analytical expressions of the ergodic net throughputs in Theorem~\ref{theorem:ULMR} and Theorem~\ref{theorem:DLMR}, which are independent of the instantaneous CSI. In this paper, we have opted for optimizing the phase shifts of the RIS in order to minimize the channel estimation error, which determines the performance of both the uplink and downlink transmission phases. The optimization of the phase shifts of the RISs based on the closed-form expressions in Theorem~\ref{theorem:ULMR} and Theorem~\ref{theorem:DLMR} is postponed to a future research work.
\end{remark}

\vspace*{-0.2cm}
\section{Numerical Results} \label{Sec:NumRes}
\vspace*{-0.1cm}
In this section, we report some numerical results in order to illustrate the performance of the RIS-assisted Cell-Free Massive MIMO system introduced in the previous sections. We consider a geographic area of size $1.5 \times 1.5$~km$^2$, where the locations of the APs and users are given in terms of $(x,y)$ coordinates. The four vertices of the considered region are $[-0.75, -0.75]$ km, $[-0.75, 0.75]$ km, $[0.75, 0.75]$ km, $[0.75, -0.75]$ km. 
To simulate a harsh communication environment, the $M$ APs are uniformly distributed in the sub-region $x,y \in [-0.75, -0.5]$~km, while the $K$ users are uniformly distributed in the sub-region  $x,y \in [0.375, 0.75]$~km. The RIS is located at the origin, i.e., $(x,y) = (0,0)$. The carrier frequency is $1.9$~GHz and the system bandwidth is $20$~MHz. Each coherence interval comprises $\tau_c = 200$ symbols, which may correspond to a coherence bandwidth equal to $B_c = 200$~KHz and a coherence time equal to $T_c = 1$~ms, except in  Fig.~\ref{FigDifferentBounds}$(b)$ where $\tau_c = 5000$ symbols. We assume $\tau_p =5$ orthonormal pilot sequences that are shared by all the users. The large-scale fading coefficients in dB are generated according to the three-slope propagation model in \cite[(51)--(53)]{ngo2017cell}, where the path loss exponent depends on the distance between the transmitter and the receiver. The shadow fading has a log-normal distribution with standard deviation equal to $8$~dB. The distance thresholds for the three slopes are $10$~m and $50$~m. The height of the APs, RIS, and users is $15$~m, $30$~m, and $1.65$~m, respectively. The direct links, $g_{mk}, \forall m,k$, are assumed to be unblocked with a given probability. More specifically, the large-scale fading coefficient $\beta_{mk}$ is formulated as follows
\setcounter{eqnback}{\value{equation}} \setcounter{equation}{62}
\begin{equation}
\beta_{mk} = \bar{\beta}_{mk}a_{mk},
\end{equation}
where $\bar{\beta}_{mk}$ accounts for the path loss due to the transmission distance and the shadow fading according to the three-slope propagation model in \cite{ngo2017cell}. The binary variables $a_{mk}$ accounts for the probability that the direct links are unblocked, and it is defined as
\begin{equation} \label{eq:amk}
a_{mk} = \begin{cases}
1, & \mbox{ with a probability } \tilde{p},\\
0, & \mbox{ with a probability } 1- \tilde{p},
\end{cases}
\end{equation}
where $\tilde{p} \in [0,1]$ is the probability that the direct link is not blocked. The noise variance is $-92$~dBm, which corresponds to a noise figure of $9$~dB. The covariance matrices are generated according to  the spatial correlation model in \eqref{eq:CovarMa}. The power of the pilot sequences is $100$~mW and the power budget of each AP is $200$~mW. The time intervals of the data transmission phase, in each coherence time, that are allocated to the uplink and downink transmissions are,  $\nu_u = \nu_d = 0.5$. The uplink and downlink power control coefficients are $\eta_k = 1, \forall k,$ and $\eta_{mk} = ( \sum_{k'=1}^K \gamma_{mk'} )^{-1},\forall m,k$, which is directly obtained from \eqref{eq:PowerConst} to satisfy the limited power budget per AP.\footnote{In the worst case, if all the direct links are blocked, we introduce a damping constant  when Cell-Free Massive MIMO systems in the absence of the RIS are considered, since in those cases we have $\sum_{m=1}^M \gamma_{mk'} = 0$.}  As far as the optimization of the phase shifts of the RIS elements are concerned, we assume that they are optimized according to the sum-NMSE minimization criterion in the absence of direct links, according to Corollary~\ref{corollary:EqualPhase} and Remark~\ref{RemarkProb}. Without loss of generality, in particular, the $N$ phase shifts in $\pmb{\Phi}$ are all set equal to $\pi/4$, except in Figs.~\ref{FigSumUserDiffPhaseShiftsUL} and \ref{FigSumUserDiffPhaseShiftsDL} where different phase shifts are considered for comparison. In order to evaluate the advantages and limitations of RIS-assisted Cell-Free Massive MIMO systems, three system configurations are considered for comparison:
\begin{itemize}
\item[$i)$] \textit{RIS-Assisted Cell-Free Massive MIMO}: This is the proposed system model in which the direct links are blocked according to \eqref{eq:amk}. This setup is denoted by ``RIS-CellFree''.
\item[$ii)$] \textit{(Conventional) Cell-Free Massive MIMO}: This setup is the same as the previous one with the only exception that the RIS is not deployed in the network. This setup is denoted by ``CellFree''.
\item[$iii)$] \textit{RIS-Assisted Cell-Free Massive MIMO with blocked direct links}: This is the worst case study in which the direct links are blocked with unit probability and the uplink and downlink transmission phases are ensured only through the RIS. This setup is denoted by ``RIS-CellFree-NoLOS''.
\end{itemize}
\begin{figure*}[t]
	\begin{minipage}{0.48\textwidth}
		\centering
		\includegraphics[trim=0.9cm 0cm 0.9cm 0.6cm, clip=true, width=3.2in]{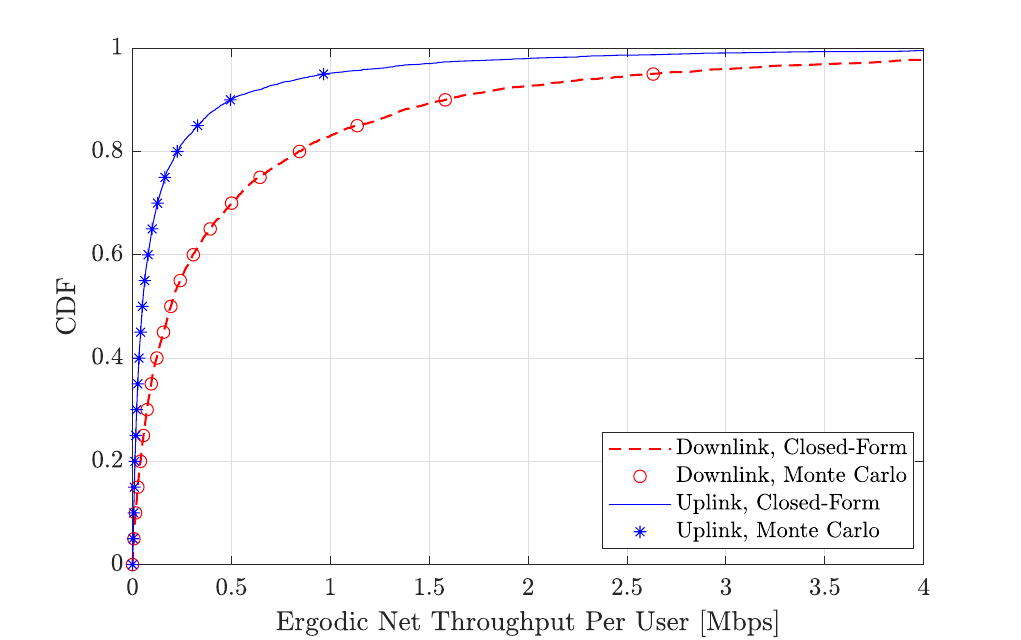} \vspace*{-0.2cm}
		\caption{Monte Carlo simulations versus the analytical frameworks with $M= 20$, $K=5$, $N= 64$, $\tau_p = 2,$ and $d_H= d_V = \lambda/4$. The unblocked probability of the direct links is $\tilde{p} = 1.0$.}
		\label{FigMonteCarloClosedForm}
		\vspace*{-0.2cm}
	\end{minipage}
	\vspace*{-0.0cm}
 \hfill
\begin{minipage}{0.48\textwidth}
	\centering
	\includegraphics[trim=0.9cm 0cm 0.9cm 0.6cm, clip=true, width=3.2in]{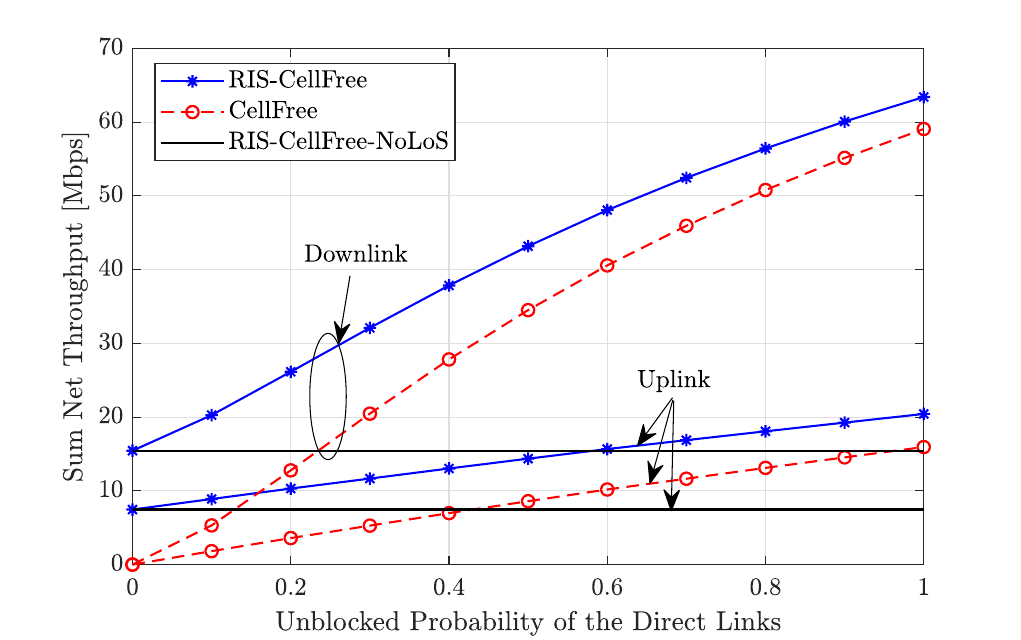} \vspace*{-0.2cm}
	\caption{Average sum net throughput [Mbps] versus the unblocked probability of the direct links $\tilde{p}$ with $M= 100$, $K=10$, $N= 900, \tau_p = 5,$ and $d_H = d_V = \lambda/4$.}
	\label{FigActiveProbSumUser}
	\vspace*{-0.2cm}
\end{minipage}
\end{figure*}

In Fig.~\ref{FigMonteCarloClosedForm}, we illustrate the cumulative distribution function (CDF) of the net throughput by using Monte Carlo simulations and the proposed analytical framework. The CDF is computed with respect to the locations of the APs and users in the considered area. The Monte Carlo simulation results are obtained by using \eqref{eq:ULRate}  and \eqref{eq:DLRate} with the SINRs given in \eqref{eq:ULSINR} and \eqref{eq:DLSINR}, while the analytical results are obtained by using Theorem~\ref{theorem:ULMR} and Theorem~\ref{theorem:DLMR}. We observe a very good overlap between the numerical simulations and the obtained analytical expressions. From Fig.~\ref{FigMonteCarloClosedForm}, we evince that the downlink net throughput per user is about $2.6\times$ better than the uplink net throughput. This is due to the higher transmission power of the APs and the gain of the joint processing of the APs. Since the Monte Carlo simulations are not simple to obtain for larger values of the simulation parameters, the rest of the figures are obtained by using the closed-form expressions of the net throughput derived in Theorem~\ref{theorem:ULMR} and Theorem~\ref{theorem:DLMR}.

\begin{figure*}[t]
	\begin{minipage}{0.48\textwidth}
		\centering
		\includegraphics[trim=0.9cm 0cm 0.9cm 0.6cm, clip=true, width=3.2in]{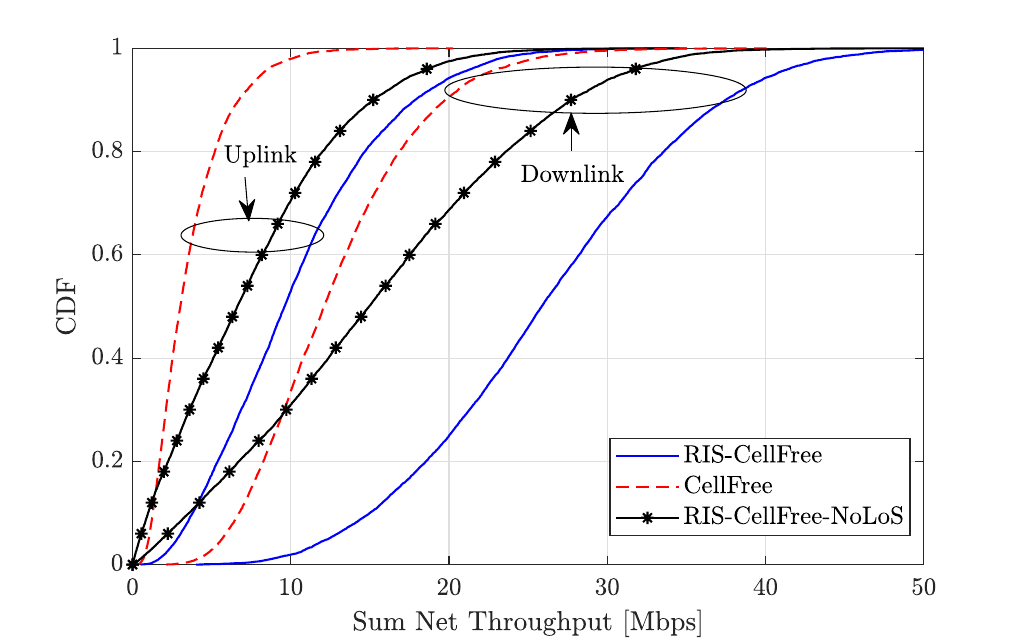} \vspace*{-0.2cm}
		\caption{CDF of the sum net throughput [Mbps] with $M= 100$, $K=10$, $N=900$, $\tau_p = 5,$ and $d_H = d_V = \lambda/4$. The unblocked probability of the direct links is $\tilde{p}=0.2$. }
		\label{FigCDFSumUser}
		\vspace*{-0.2cm}
	\end{minipage}
	\hfill
	\begin{minipage}{0.48\textwidth}
		\centering
		\includegraphics[trim=0.9cm 0cm 0.9cm 0.6cm, clip=true, width=3.2in]{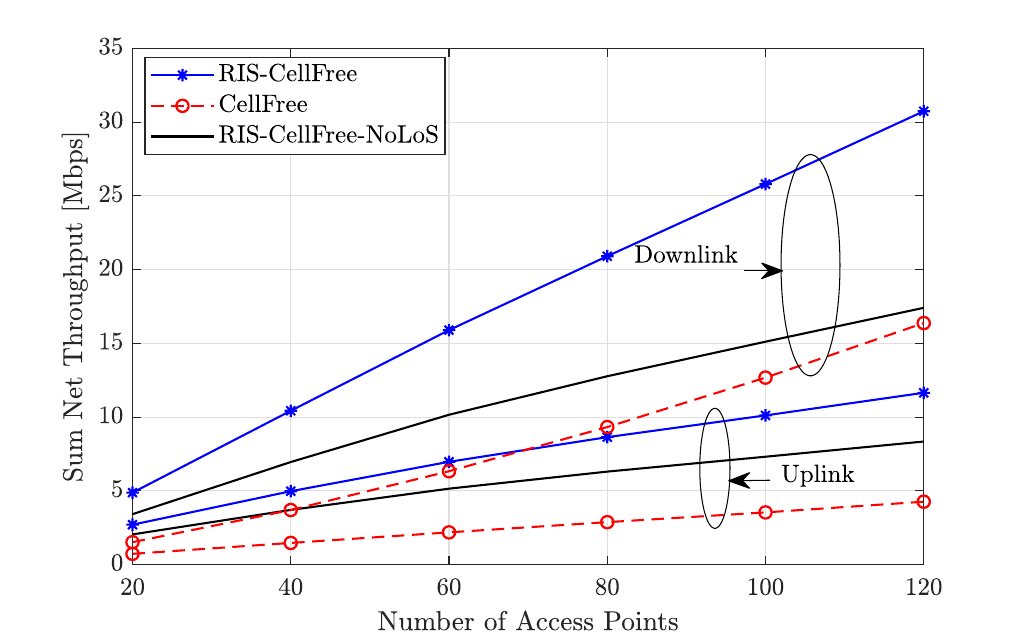} \vspace*{-0.2cm}
		\caption{Average sum net throughput [Mbps] versus the number of APs with $K= 10$, $N= 900$, $\tau_p = 5$, and $d_H = d_V = \lambda/4$. The unblocked probability of the direct links is $\tilde{p}=0.2$.}
		\label{FigDiffAPs}
		\vspace*{-0.2cm}
	\end{minipage}
\end{figure*}
\begin{figure*}[t]
	\begin{minipage}{0.48\textwidth}
		\centering
		\includegraphics[trim=0.9cm 0cm 0.9cm 0.6cm, clip=true, width=3.2in]{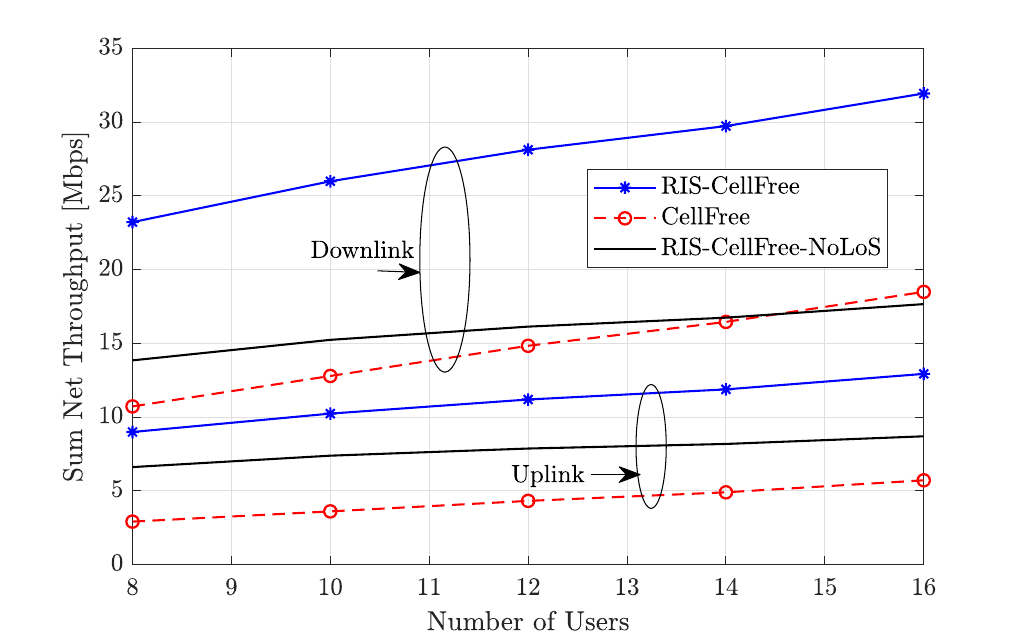} \vspace*{-0.2cm}
		\caption{Average sum net throughput [Mbps] versus the number of users with $M= 100$, $N= 900$, $\tau_p = 5$, and $d_H = d_V = \lambda/4$. The unblocked probability of the direct links is $\tilde{p}=0.2$. }
		\label{FigDiffUsersR1}
		\vspace*{-0.2cm}
	\end{minipage}
	\hfill
	\begin{minipage}{0.48\textwidth}
		\centering
		\includegraphics[trim=0.9cm 0cm 0.9cm 0.6cm, clip=true, width=3.2in]{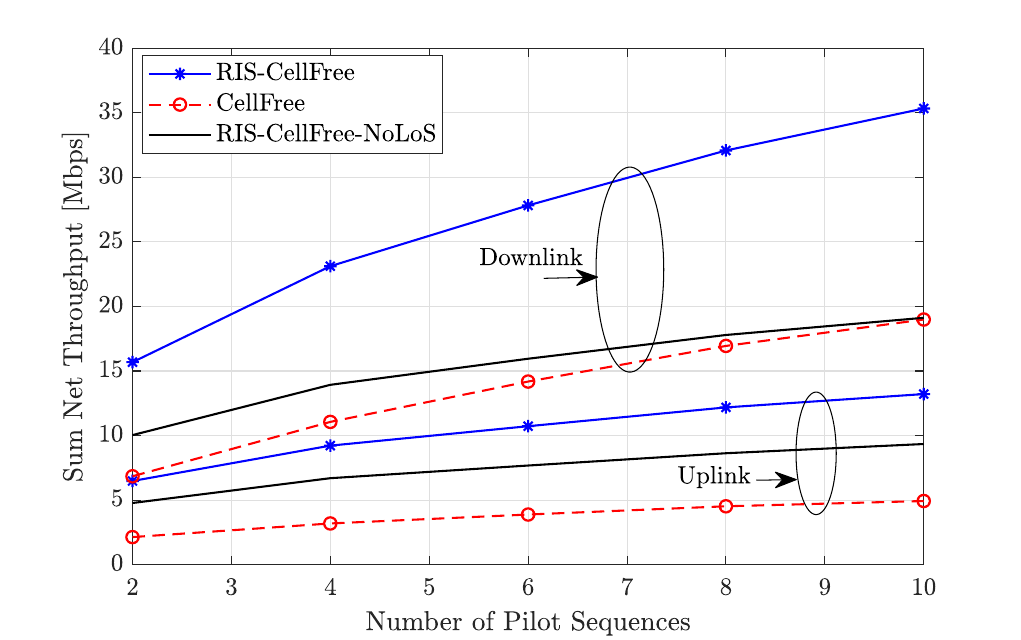} \vspace*{-0.2cm}
		\caption{Average sum net throughput [Mbps] versus the number of pilot sequences with $M= 100$, $K=10$, $N= 900$, and $d_H = d_V = \lambda/4$. The unblocked probability of the direct links is $\tilde{p}=0.2$. }
		\label{FigDiffPilotSequencesR1}
		\vspace*{-0.2cm}
	\end{minipage}
\end{figure*}
\begin{figure*}[t]
	\begin{minipage}{0.48\textwidth}
		\centering
		\includegraphics[trim=0.9cm 0cm 0.9cm 0.6cm, clip=true, width=3.2in]{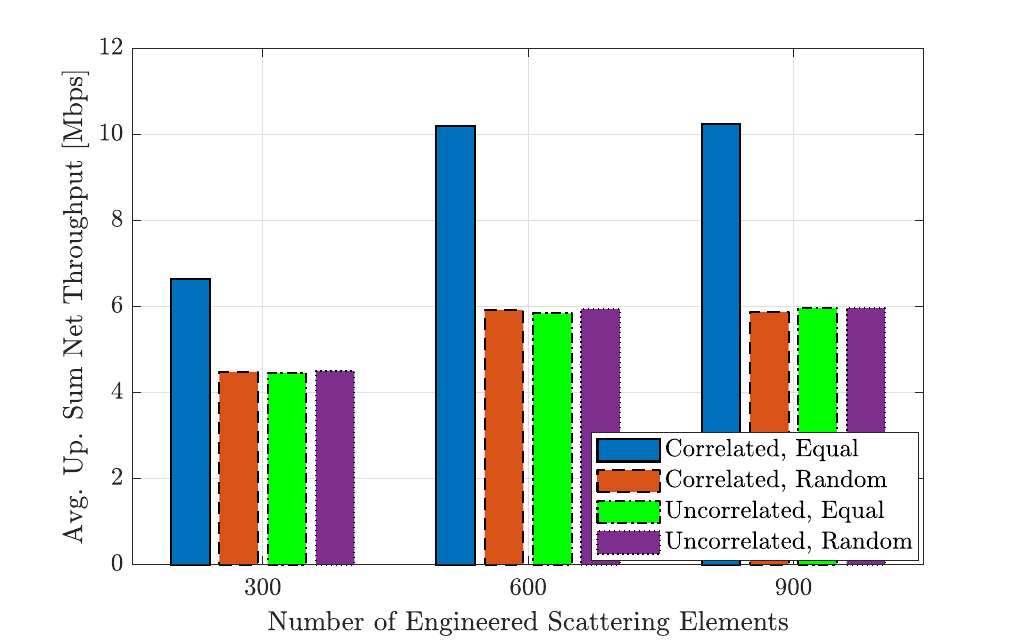} \vspace*{-0.2cm}
		\caption{Average uplink sum net throughput [Mbps] versus the number of engineered scattering elements with $M= 100$, $K= 10$, $\tau_p = 5$, and $d_H = d_V = \lambda/4$. The unblocked probability of the direct links is $\tilde{p}=0.2$. }
		\label{FigSumUserDiffPhaseShiftsUL}
		\vspace*{-0.2cm}
	\end{minipage}
	\hfill
	\begin{minipage}{0.48\textwidth}
		\centering
		\includegraphics[trim=0.9cm 0cm 0.9cm 0.6cm, clip=true, width=3.2in]{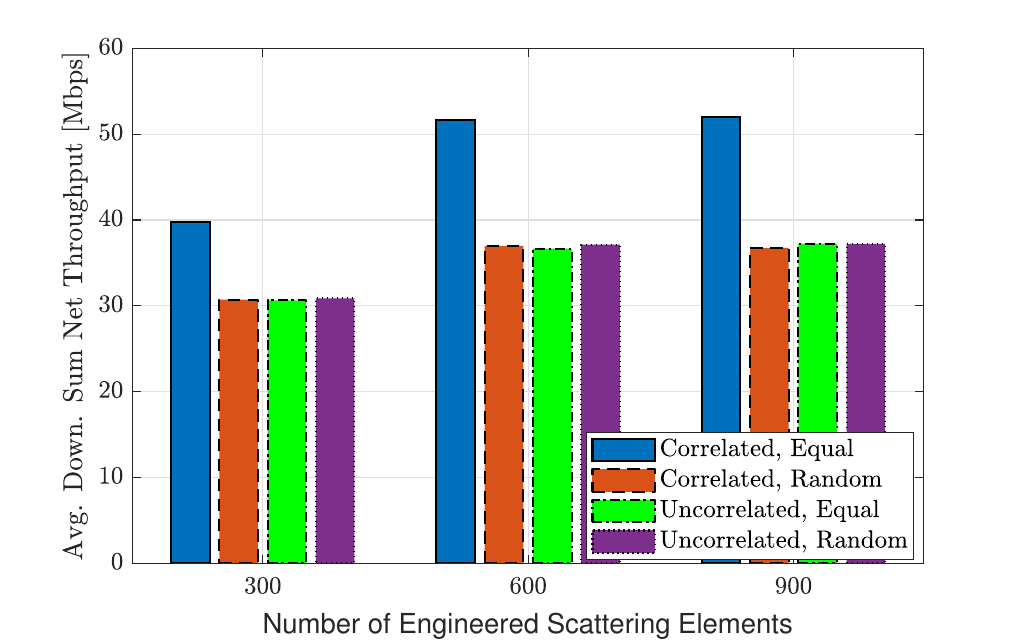} \vspace*{-0.2cm}
		\caption{Average downlink sum net throughput [Mbps] versus the number of engineered scattering elements with $M= 100$, $K=10$, $\tau_p = 5$, and $d_H = d_V = \lambda/4$. The unblocked probability of the direct links is $\tilde{p}=0.2$. }
		\label{FigSumUserDiffPhaseShiftsDL}
		\vspace*{-0.2cm}
	\end{minipage}
\end{figure*}
In Fig.~\ref{FigActiveProbSumUser}, we illustrate the average sum net throughput as a function of the probability $\tilde{p}$ in \eqref{eq:amk}. In particular, the average uplink sum net throughput is defined as $\sum_{k=1}^K \mathbb{E} \{ R_{uk} \}$ and the downlink sum net throughput is defined as $\sum_{k=1}^K \mathbb{E} \{ R_{dk} \}$, where the uplink and downlink SINRs are obtained by  using Theorem~\ref{theorem:ULMR} and Theorem~\ref{theorem:DLMR}. In particular, the expectation is computed with respect to the locations of the APs and users in the considered area. From the obtained results, we evince that Cell-Free Massive MIMO provides the worst performance if the blocking probability is large ($\tilde{p}$ is small). As expected, the average net throughput offered by Cell-Free Massive MIMO tends to zero if $\tilde{p} \to 0$ (the direct links are unreliable). For example, at $\tilde{p}=0.1$, the average sum net throughput of Cell-Free Massive MIMO is approximately $2.0\times$ and $1.4\times$ smaller, in the uplink and downlink, respectively, than the average sum net throughput of the worst-case RIS-assisted Cell-Free Massive MIMO setup (i.e., RIS-CellFree-NoLOS). In the considered case study, in addition, we note that the proposed RIS-assisted Cell-Free Massive MIMO setup offers the best average net throughput, since it can overcome the unreliability of the direct links thanks to the presence of the RIS. The presence of the RIS is particularly useful if $\tilde{p}$ is small, i.e., $\tilde{p}<0.2$ in Fig.~\ref{FigActiveProbSumUser}, since the direct links are not able to support a high throughput. In this case, the combination of Cell-Free Massive MIMO and RIS is capable of providing a high throughput and signal reliability.

In Fig.~\ref{FigCDFSumUser}, we compare the three considered systems in terms of average sum net throughput when $\tilde{p}=0.2$. We observe the net advantage of the proposed RIS-assisted Cell-Free Massive MIMO system, especially in the downlink. In the uplink, in addition, even the worst-case RIS-assisted Cell-Free Massive MIMO system setup (i.e., $\tilde{p}=0$) outperforms the Cell-Free Massive MIMO setup in the absence of an RIS. In Figs.~\ref{FigDiffAPs}--\ref{FigDiffPilotSequencesR1}, we show the average sum net throughput as a function of the number of APs, the number of users, and the number of orthonormal pilot signals, respectively. We observe that the average sum net throughput increases with the number of APs, with the number or users, and with the number of pilot sequences, and that the RIS-assisted Cell-Free Massive MIMO setup outperforms, especially in the downlink, the other benchmark schemes. Gains of the order of $1.7\times$ and $2.6\times$ are obtained in the considered setups.
 
\begin{figure*}[t]
	\begin{minipage}{0.48\textwidth}
		\centering
		\includegraphics[trim=0.9cm 0cm 0.9cm 0.6cm, clip=true, width=3.2in]{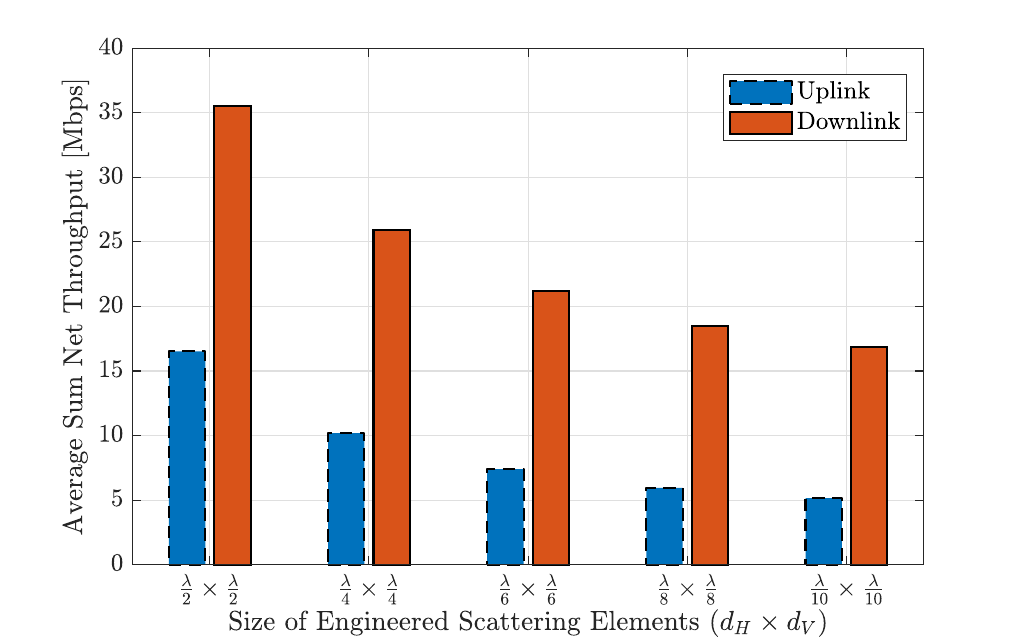} \vspace*{-0.2cm}
		\caption{Average sum net throughput [Mbps] versus the size of engineered scattering elements  $\{d_H, d_V \}$, but for a different size of the RIS with $M= 100$, $K=10$, $N= 900$, and $\tau_p = 5$. The unblocked probability of the direct links is $\tilde{p}=0.2$.}
		\label{FigDifferentdHdVDownlink}
		\vspace*{-0.2cm}
	\end{minipage}
	\hfill
	\begin{minipage}{0.48\textwidth}
		\centering
		\includegraphics[trim=0.9cm 0cm 0.9cm 0.6cm, clip=true, width=3.2in]{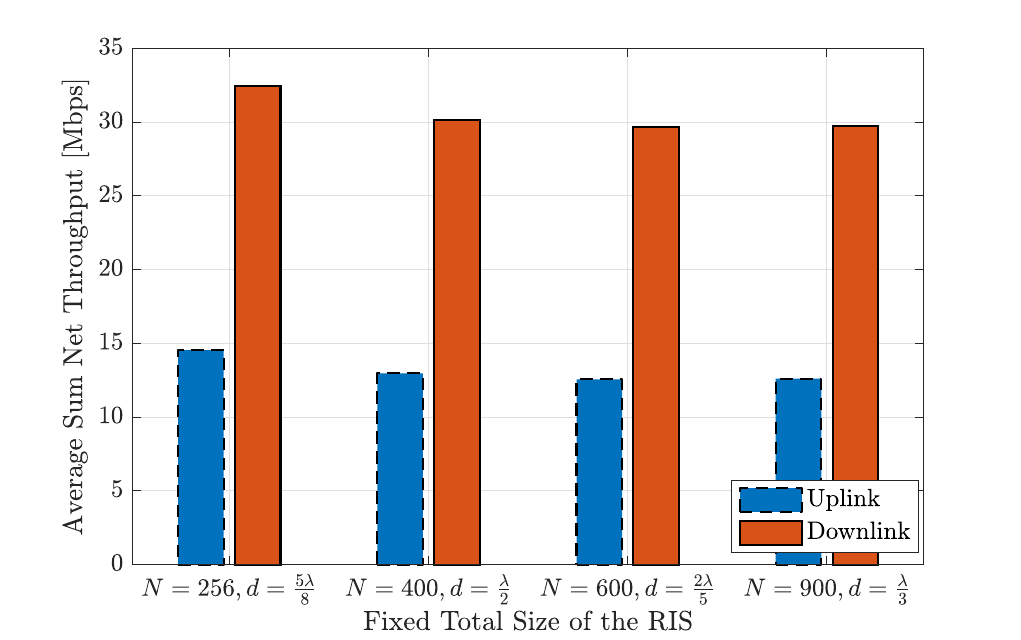} \vspace*{-0.2cm}
		\caption{Average sum net throughput [Mbps] for a fixed total size of the RIS but for a different number of RIS elements  with $M= 100 $, $K= 10$, $\tau_p = 5,$ and $d=d_H=d_V$. The unblocked probability of the direct links is $\tilde{p}=0.2$. }
		\label{FigFixedArea}
		\vspace*{-0.2cm}
	\end{minipage}
\end{figure*}

In the following figures, we focus our attention only on the RIS-assisted Cell-Free Massive MIMO setup, since it provides the best performance in the analyzed setups. In Figs.~\ref{FigSumUserDiffPhaseShiftsUL} and \ref{FigSumUserDiffPhaseShiftsDL}, we report the uplink and downlink average sum net throughput as a function of number of engineered scattering elements of the RIS. In particular, we compare the average sum net throughput when the phase shifts of the RIS are randomly chosen and are optimized according to Corollary~\ref{corollary:EqualPhase} over spatially-independent and spatially-corrrelated fading channels according to \eqref{eq:CovarMa}. In the presence of spatial correlation, the channel correlation matrices are $\mathbf{R}_m = \alpha_m d_H d_V \mathbf{I}_N$ and $\widetilde{\mathbf{R}}_{mk} = \tilde{\alpha}_{mk} d_H d_V \mathbf{I}_N, \forall m,k$. We see different performance trends over spatially-independent and spatially-correlated fading channels. If the spatial correlation is not considered, we observe that there is no significant difference between the random and uniform phase shifts setup. In the presence of spatial uncorrelation, on the other hand, the uniform phase shift design  obtained from Corollary~\ref{corollary:EqualPhase} provides a much higher average throughput. This result highlights the relevance of using even simple optimization designs for RIS-assisted communications over spatially-correlated fading channels. 

\begin{figure*}[t]
	\begin{minipage}{0.48\textwidth}
	\centering
	\includegraphics[trim=0.9cm 0cm 0.9cm 0.6cm, clip=true, width=3.2in]{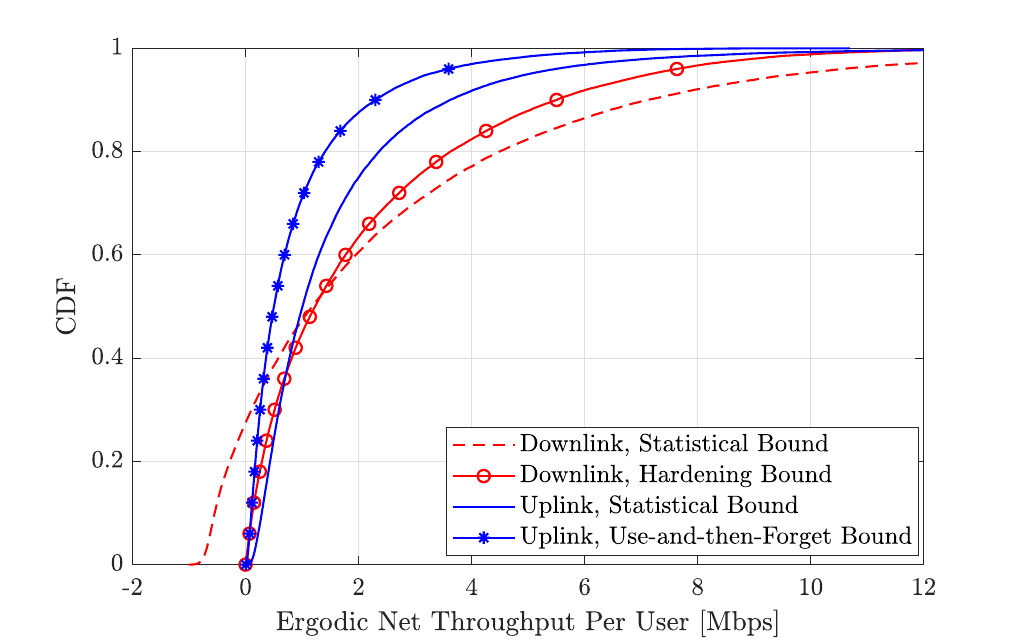} \vspace*{-0.1cm} \\
	$(a)$
	\vspace*{-0.0cm}
\end{minipage}
		\begin{minipage}{0.48\textwidth}
		\centering
		\includegraphics[trim=0.9cm 0cm 0.9cm 0.6cm, clip=true, width=3.2in]{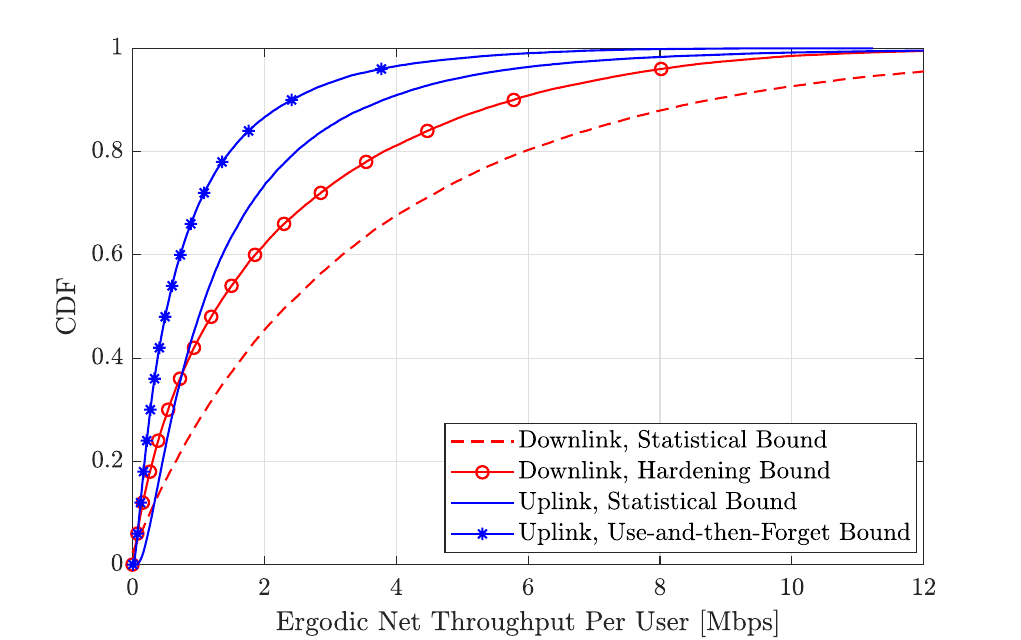} \vspace*{-0.1cm} \\
		$(b)$
		\vspace*{-0.0cm}
	\end{minipage}
\caption{A comparison of different channel capacity bounding techniques with $M= 20$, $K=10$, $N= 64$, $\tau_p = 10$, $d_H= d_V = \lambda/4$,  $\tilde{p} = 1.0$, and two coherence intervals: $(a)$ $\tau_c = 200$ symbols; $(b)$ $\tau_c =5000$ symbols.}
    \label{FigDifferentBounds}
	\vspace*{-0.2cm}
\end{figure*}
In Fig.~\ref{FigDifferentdHdVDownlink}, we analyze the impact of the size of the engineered scattering elements of the RISs on the uplink and downlink average net throughput, while keeping the total number of RIS elements $N$ fixed. The size of the considered RIS, which is a compact surface, is $N d_H d_V$, which implies that it increases as the size $d_H d_V$ of each element of the RIS increases. In this setup, we observe that the average net throughput increases as the physical size of each element of the RIS increases. In Fig.~\ref{FigFixedArea}, on the other hand, we analyze a setup in which the total size of the RIS is kept constant and equal to $N d_H d_V = 10 \lambda \times 10 \lambda$ while the triplet $(N, d=d_H=d_V)$ is changed accordingly. With the considered fading spatial correlation model and for a size of the RIS elements no smaller than $\lambda/3$, we do not observe a significant difference on the average net throughput. Further studies are, however, necessary for deep sub-wavelength RIS structures, for different optimization criteria of the phase shifts of the RIS, and in the presence of mutual coupling in addition to the fading spatial correlation \cite{gradoni2021end,qian2021mutual}.

In Fig.~\ref{FigDifferentBounds}, we plot the ergodic net throughput per user by utilizing different channel capacity bounding techniques, by assuming $\tau_c = 200$ symbols and $\tau_c = 5000$ symbols. The main benefit of the use-and-then-forget bound in \eqref{eq:ULRate} and the hardening bound in \eqref{eq:DLRate} is the possibility of obtaining a closed-form expression for the net throughput by capitalizing on  the fundamentals properties of Massive MIMO communications, as demonstrated in Theorems~\ref{theorem:ULMR} and \ref{theorem:DLMR}. However, the channel hardening capability reduces in the presence of an RIS \cite{van2021reconfigurable}. Thus, other channel capacity bounding techniques may result in a better estimate of the net throughput as compared to the actual channel capacity. The statistical bound in \cite{caire2018ergodic,interdonato2019downlink} was originally derived for application to the downlink data transmission when no instantaneous CSI is available at the users and the channel vectors may be less hardened \cite{interdonato2021enhanced}. In order to apply the same bound to the uplink data transmission, we assume that only the channel statistics are available at the CPU. Even though the statistical bound results in a better ergodic net throughput per user, some realizations of user locations and shadow fading may lead to negative values of the ergodic net throughput. This may occur in high mobility scenarios, which correspond to small values of $\tau_c$, as illustrated in Fig.~\ref{FigDifferentBounds}$(a)$. The statistical bound provides consistent values of the ergodic net throughput in  low mobility scenarios when $\tau_c$ is large, as displayed in Fig.~\ref{FigDifferentBounds}$(b)$. The numerical results unveils the need of developing different and more accurate channel bounding methods for evaluating the net throughput of RIS-assisted Cell-Free Massive MIMO systems. 
\vspace*{-0.2cm}
\section{Conclusion}\label{Sec:Conclusion}
\vspace*{-0.1cm}
Cell-Free Massive MIMO and RIS are two disruptive technologies for boosting the system performance of future wireless networks. These two technologies are not competing with each other, but have complementary features that can be integrated and leveraged for enhancing the system performance in harsh communication environments. Therefore, we have considered an RIS-assisted Cell-Free Massive MIMO system that operates according to the TDD mode. An efficient channel estimation scheme has been introduced to overcome the high overhead that may be associated with the estimation of the individual channels of the RIS elements. Based on the proposed channel estimation scheme, an optimal design for the phase shifts of the RIS that minimizes the channel estimation error has been devised and has been used for system analysis. Based on the proposed channel estimation method, closed-form expressions of the ergodic net throughput for the uplink and downlink data transmission phases have been proposed. Based on them, the performance of RIS-assisted Cell-Free Massive MIMO has been analyzed as a function of the fading spatial correlation and the blocking probability of the direct AP-user links. The numerical results have shown that the presence of an RIS is particularly useful if the AP-user links are unreliable with high probability.

Possible generalizations of the results illustrated in this paper include the optimization of the phase shifts of the RIS that maximize the uplink or downlink throughput, the analysis of the impact of the fading spatial correlation for non-compact and deep sub-wavelength RIS structures, and the analysis and optimization of RIS-assisted systems in the presence of mutual coupling.

\vspace*{-0.2cm}
\appendix
\vspace*{-0.1cm}
\subsection{Useful Lemmas} \label{Appendix:UsefulLemmas}
\vspace*{-0.1cm}
This section reports three useful lemmas that are utilized for asymptotic analysis.  
\begin{lemma}{\cite[Lemma B.$7$]{massivemimobook}} \label{lemma:trace}
For an arbitrary matrix $\mathbf{X} \in \mathbb{C}^{N\times N}$ and a positive semi-definite matrix $\mathbf{Y} \in \mathbb{C}^{N\times N}$, it holds that $|\mathrm{tr} (\mathbf{X}\mathbf{Y})| \leq \|  \mathbf{X} \|_2 \mathrm{tr}(\mathbf{Y})$. If $\mathbf{X}$ is also a positive semi-definite matrix, then $\mathrm{tr} (\mathbf{X}\mathbf{Y}) \leq \|  \mathbf{X} \|_2 \mathrm{tr}(\mathbf{Y})$.
\end{lemma}
\begin{lemma}{\cite[Lemma~$9$]{Chien2020book}} \label{Lemma:Sup1}
For a random variable $\mathbf{x} \in \mathbb{C}^N$ distributed as $\mathcal{CN}(\mathbf{0}, \bar{\mathbf{R}})$ with $\bar{\mathbf{R}} \in \mathbb{C}^{N \times N}$ and a given deterministic matrix $\mathbf{M} \in \mathbb{C}^{N \times N}$, it holds that
$\mathbb{E} \big\{ |\mathbf{x}^H \mathbf{M} \mathbf{x} |^2 \big\} = \big| \mathrm{tr}( \bar{\mathbf{R}} \mathbf{M} ) \big|^2 + \mathrm{tr} \big( \bar{\mathbf{R}} \mathbf{M} \bar{\mathbf{R}} \mathbf{M}^H \big)$.
\end{lemma}
\begin{lemma}\label{lemmaMN}
For a random vector $\mathbf{x} \in \mathbb{C}^N$ distributed as $\mathbf{x} \sim \mathcal{CN}(\mathbf{0}, \bar{\mathbf{R}})$ with $\bar{\mathbf{R}} \in \mathbb{C}^{N \times N}$ and  two deterministic matrices $\mathbf{M},\mathbf{N} \in \mathbb{C}^{N \times N}$, it holds that
\begin{equation} \label{eq:xAxxBx}
\mathbb{E} \{ \mathbf{x}^H \mathbf{M} \mathbf{x} \mathbf{x}^H \mathbf{N} \mathbf{x}  \} = \mathrm{tr}(\bar{\mathbf{R}}\mathbf{M}\bar{\mathbf{R}}\mathbf{N}) + \mathrm{tr}(\bar{\mathbf{R}}\mathbf{M}) \mathrm{tr}(\bar{\mathbf{R}}\mathbf{N}).
\end{equation}
\end{lemma}
\begin{proof}
Consider $\mathbf{x} = \bar{\mathbf{R}}^{1/2} \tilde{\mathbf{x}}$ with $\tilde{\mathbf{x}} \sim \mathcal{CN}(\mathbf{0}, \mathbf{I}_N)$. Let us further denote $\widetilde{\mathbf{M}} = \bar{\mathbf{R}}^{1/2} \mathbf{M} \bar{\mathbf{R}}^{1/2}$ and $\widetilde{\mathbf{N}} = \bar{\mathbf{R}}^{1/2} \mathbf{N} \bar{\mathbf{R}}^{1/2}$,  where  $[\widetilde{\mathbf{M}}]_{mn}$ and $[\widetilde{\mathbf{N}}]_{mn}$ are the $(m,n)-$th elements of matrix $\widetilde{\mathbf{M}}$ and $\widetilde{\mathbf{N}}$, respectively.  Then, the expectation on the left-hand side of \eqref{eq:xAxxBx} is
\begin{equation} \label{eq:xMxxNx}
\begin{split}
&\mathbb{E} \{ \mathbf{x}^H \mathbf{M} \mathbf{x} \mathbf{x}^H \mathbf{N} \mathbf{x}  \} \\
&= \mathbb{E} \left\{ \left( \sum_{m=1}^N \sum_{n=1}^N \tilde{x}_m^\ast [\widetilde{\mathbf{M}}]_{mn} \tilde{x}_n \right)\left( \sum_{m=1}^N \sum_{n=1}^N \tilde{x}_m^\ast [\widetilde{\mathbf{N}}]_{mn} \tilde{x}_n \right) \right\}\\
& = \sum_{m=1}^N \sum_{n=1}^N \sum_{m'=1}^N \sum_{n'=1}^N  [\widetilde{\mathbf{M}}]_{mn}  [\widetilde{\mathbf{N}}]_{m'n'} \mathbb{E} \{ \tilde{x}_m^\ast \tilde{x}_n  \tilde{x}_{m'}^\ast \tilde{x}_{n'}\}.
\end{split}
\end{equation}
where $\tilde{x}_n$ is the $n-$th element of vector $\tilde{\mathbf{x}}$. By noting that $\mathbb{E}\{ |\tilde{x}_m|^4 \} = 2$ from Lemma~\ref{Lemma:Sup1} and $\mathbb{E}\{ |\tilde{x}_m|^2 |\tilde{x}_n|^2 \} = 1$ if $m \neq n$, we obtain the following
\begin{equation} \label{eq:Ecase}
\mathbb{E} \{ \tilde{x}_m^\ast \tilde{x}_n  \tilde{x}_{m'}^\ast \tilde{x}_{n'}\} = \begin{cases}
2, & \mbox{if } m=n=m'=n', \\
1, & \mbox{if } (m=n) \neq (m'=n'), \\
1, & \mbox{if } (m=n') \neq (m'=n), \\
0, & \mbox{otherwise}. 
\end{cases}
\end{equation}
Consequently, \eqref{eq:xMxxNx} can be further simplified as
\begin{equation}
\mathbb{E} \{ \mathbf{x}^H \mathbf{M} \mathbf{x} \mathbf{x}^H \mathbf{N} \mathbf{x}  \} = \sum_{m=1}^N \sum_{n=1}^N  [\widetilde{\mathbf{M}}]_{mm} [\widetilde{\mathbf{N}}]_{nn} + \sum_{m=1}^N \sum_{n=1}^N  [\widetilde{\mathbf{M}}]_{mn} [\widetilde{\mathbf{N}}]_{nm},
\end{equation}
which, with the aid of some algebraic manipulations, coincides with \eqref{eq:xAxxBx}.
\end{proof}
\vspace*{-0.2cm}
\subsection{Proof of Lemma~\ref{lemma:ChannelProperty}} \label{appendix:ChannelProperty}
\vspace*{-0.1cm}
We first compute the second moment of the aggregated channel $u_{mk}$ by capitalizing on the statistical independence of the direct and indirect channels, as follows
\begin{equation}
\begin{split}
\mathbb{E} \{ |u_{mk}|^2 \} &=   \mathbb{E} \{ |g_{mk}|^2 \} + \mathbb{E} \big\{ |\mathbf{h}_{m}^H \pmb{\Phi} \mathbf{z}_{k}|^2 \big\} \\
&\stackrel{(a)}{=} \beta_{mk} +  \mathbb{E} \left\{ \mathrm{tr}\big(  \pmb{\Phi}^H \mathbf{h}_{m} \mathbf{h}_{m}^H \pmb{\Phi} \mathbf{z}_{k} \mathbf{z}_{k}^H \big)  \right\}\\
&\stackrel{(b)}{=} \beta_{mk} +  \mathrm{tr}\left(  \pmb{\Phi}^H \mathbb{E} \big\{ \mathbf{h}_{m} \mathbf{h}_{m}^H  \big\}  \pmb{\Phi} \mathbb{E} \big\{ \mathbf{z}_{k} \mathbf{z}_{k}^H \big\}  \right) \\
&= \beta_{mk} + \mathrm{tr} ( \pmb{\Theta}_{mk}) = \delta_{mk},
\end{split}
\end{equation}
where $(a)$ follows by applying the trace of product property $\mathrm{tr}(\mathbf{X}\mathbf{Y})= \mathrm{tr}(\mathbf{Y}\mathbf{X})$ for some given size-matched matrices $\mathbf{X}$ and $\mathbf{Y}$; and $(b)$ is obtained thanks to the independence of the cascaded channels $\mathbf{h}_{m}$ and $\mathbf{z}_{k}$. The fourth moment of the aggregated channel can be written, from \eqref{eq:umk}, as follows
\begin{equation} \label{eq:Eumk4}
\begin{split}
&\mathbb{E} \{ |u_{mk}|^4 \} = \\ 
&\mathbb{E} \left\{ \left| | g_{mk}|^2 +  g_{mk}^\ast \mathbf{h}_{m}^H \pmb{\Phi} \mathbf{z}_{k} +  g_{mk} \mathbf{z}_{k}^H  \pmb{\Phi}^H \mathbf{h}_{m} + \big| \mathbf{h}_{m}^H \pmb{\Phi} \mathbf{z}_{k} \big|^2 \right|^2 \right\}
\end{split}
\end{equation}
By setting $a= | g_{mk}|^2$, $b= g_{mk}^\ast \mathbf{h}_{m}^H \pmb{\Phi} \mathbf{z}_{k},$ $c= g_{mk} \mathbf{z}_{k}^H  \pmb{\Phi}^H \mathbf{h}_{m},$ and $d=\big| \mathbf{h}_{m}^H \pmb{\Phi} \mathbf{z}_{k} \big|^2,$ \eqref{eq:Eumk4} can be equivalently written as follows
\begin{equation} \label{eq:Eumk4v1}
\mathbb{E} \{|u_{mk}|^4\} = \mathbb{E}\{ |a|^2 \} + \mathbb{E}\{ |b|^2 \} + \mathbb{E}\{ |c|^2 \} + 2 \mathbb{E}\{ a d \} + \mathbb{E}\{ |d|^2 \}.
\end{equation}
By applying Lemma~\ref{Lemma:Sup1} with $g_{mk} \sim \mathcal{CN}(0, \beta_{mk})$, the first expectation on the right-hand side of \eqref{eq:Eumk4v1} is equal to
\begin{equation} \label{eq:a}
\mathbb{E}\{|a|^2\} = 2 \beta_{mk}^2.
\end{equation}
By exploiting the independence between the direct and RIS-assisted links, the next three expectations on the right-hand side of \eqref{eq:Eumk4v1} are equal to
\begin{equation}
\mathbb{E}\{ |b|^2 \}= \mathbb{E}\{ |c|^2 \}= \mathbb{E}\{ ad \} = \beta_{mk} \mathrm{tr} ( \pmb{\Theta}_{mk} ).
\end{equation}
By introducing the normalized variable $\tilde{z} = \mathbf{h}_{m}^H \pmb{\Phi} \mathbf{z}_{k}/\big\| \mathbf{R}_m^{1/2} \pmb{\Phi} \mathbf{z}_{k} \big\|$ with $\tilde{z} \sim \mathcal{CN}(0,1)$, the last expectation on the right-hand side of  \eqref{eq:Eumk4v1} is equal to
\begin{equation} \label{eq:d}
\begin{split}
&\mathbb{E}\{ |d|^2 \} = \mathbb{E}\left\{ \left\| \mathbf{R}_{m}^{1/2} \pmb{\Phi} \mathbf{z}_{k}  \right\|^4 |\tilde{z}|^4 \right\} \\
&\stackrel{(a)}{=} \mathbb{E}\left\{ \left\| \mathbf{R}_{m}^{1/2} \pmb{\Phi}  \mathbf{z}_{k}  \right\|^4 \right\} \mathbb{E}\{ |\tilde{z}|^4 \} = 2 \left( \mathrm{tr} ( \pmb{\Theta}_{mk} ) \right)^2 + 2\mathrm{tr}\big( \pmb{\Theta}_{mk}^2 \big) ,
\end{split}
\end{equation}
where $(a)$ follows because $\tilde{z}$ is independent of the remaining random variables; and $(b)$ is obtained by virtue of Lemma~\ref{Lemma:Sup1}. Inserting \eqref{eq:a}--\eqref{eq:d} into \eqref{eq:Eumk4v1}, the proof follows with the aid of some algebraic manipulations.

Also, by exploiting the second moment in \eqref{eq:2Order}, the expectation of the two independent aggregated channels can be formulated as shown in \eqref{eq:2IndeAggre}. The correlation between the two aggregated channels $u_{mk}$ and $u_{m'k}$, $m \neq m'$, is, by definition, as follows
\begin{equation}
\begin{split}
&\mathbb{E} \{ u_{mk} u_{m'k}^{\ast} \} = \mathbb{E} \{ g_{mk} g_{m'k}^{\ast} \} + \mathbb{E} \{ g_{mk} (\mathbf{h}_{m'}^H \pmb{\Phi} \mathbf{z}_k)^\ast \}\\
& + \mathbb{E} \{\mathbf{h}_{m}^H \pmb{\Phi} \mathbf{z}_k  g_{m'k}^\ast \}  + \mathbb{E}\{ \mathbf{h}_{m}^H \pmb{\Phi} \mathbf{z}_k (\mathbf{h}_{m'}^H \pmb{\Phi} \mathbf{z}_k)^\ast \} \stackrel{(a)}{=} 0,
\end{split}
\end{equation}
where $(a)$ follows because the propagation channels are independent.

The expectation in \eqref{eq:2UncorreChan} can be written as follows
\begin{equation} \label{eq:product2u}
\begin{split}
&\mathbb{E} \{ | u_{mk} u_{m'k}^\ast|^2 \} =  \mathbb{E} \{ | u_{mk}|^2  |u_{m'k}|^2 \} \\
&= \mathbb{E} \{ |g_{mk}|^2 |g_{m'k}|^2\} +  \mathbb{E} \{ |g_{mk}|^2 \mathbf{z}_k^H \pmb{\Phi}^H \mathbf{h}_{m'} \mathbf{h}_{m'}^H \pmb{\Phi} \mathbf{z}_k \}\\
& \quad + \mathbb{E} \{ |g_{m'k}|^2 \mathbf{z}_k^H \pmb{\Phi}^H \mathbf{h}_m \mathbf{h}_m^H \pmb{\Phi} \mathbf{z}_k \} \\
&\quad + \mathbb{E}\{ \mathbf{z}_k^H \pmb{\Phi}^H \mathbf{h}_m \mathbf{h}_m^H \pmb{\Phi} \mathbf{z}_k \mathbf{z}_k^H \pmb{\Phi}^H \mathbf{h}_{m'} \mathbf{h}_{m'}^H \pmb{\Phi} \mathbf{z}_k \}\\
& = \beta_{mk} \beta_{m'k} + \beta_{mk} \mathrm{tr}( \pmb{\Theta}_{m'k} ) + \beta_{m'k} \mathrm{tr}( \pmb{\Theta}_{mk} ) \\
&\quad + \mathbb{E}\{ \mathbf{z}_k^H \pmb{\Phi}^H \mathbf{R}_m \pmb{\Phi} \mathbf{z}_k \mathbf{z}_k^H \pmb{\Phi}^H \mathbf{R}_{m'} \pmb{\Phi} \mathbf{z}_k \} \\
&\stackrel{(a)}{=} \big(\beta_{mk} + \mathrm{tr}( \pmb{\Theta}_{mk}) \big)\big( \beta_{m'k} + \mathrm{tr}( \pmb{\Theta}_{m'k} )  \big) + \mathrm{tr}( \pmb{\Theta}_{mk} \pmb{\Theta}_{m'k} ),
\end{split}
\end{equation}
where $(a)$ is obtained by utilizing Lemma~\ref{lemmaMN}. Similar steps can be applied to the two aggregated channels $u_{mk}$ and $u_{mk'}$ with $k \neq k'$. 

The expectation in \eqref{eq:4Chan} can be written as follows
\begin{equation}
\begin{split}
&\mathbb{E}\{ u_{mk}^\ast u_{mk'} u_{m'k'}^{\ast} u_{m'k}  \} \\
& \stackrel{(a)}{=} \mathbb{E} \{ \mathbf{z}_k^H \pmb{\Phi}^H \mathbf{h}_m \mathbf{h}_m^H \pmb{\Phi} \mathbf{z}_{k'} \mathbf{z}_{k'}^H \pmb{\Phi}^H \mathbf{h}_{m'} \mathbf{h}_{m'}^H \pmb{\Phi} \mathbf{z}_k \} \\
& = \mathrm{tr}(\pmb{\Phi}^H \mathbf{R}_m \pmb{\Phi} \widetilde{\mathbf{R}}_{k'} \pmb{\Phi}^H \mathbf{R}_{m'}\pmb{\Phi} \widetilde{\mathbf{R}}_k ),
\end{split}
\end{equation}
where $(a)$ follows from the independence of the direct links. 
\vspace*{-0.2cm}
\subsection{Proof of Corollary~\ref{CorollaryReal}} \label{AppendixReal}
\vspace*{-0.1cm}
By utilizing the identities $\mathbb{E} \{ (X - \mathbb{E}\{X\}) (Y - \mathbb{E}\{Y\}) \} = \mathbb{E}\{XY\} - \mathbb{E}\{X\} \mathbb{E}\{Y\}$ and $\mathbb{E}\{ | X - \mathbb{E}\{X \}|^2 \} = \mathbb{E}\{ |X|^2 \} - |\mathbb{E}\{ X \}|^2$, the expectation in \eqref{eq:omkmk} can be formulated as follows
\begin{multline} \label{eq:omkomk}
\mathbb{E} \{ o_{mk} o_{m'k}^{\ast}  \} = \sqrt{\alpha_{mk} \alpha_{m'k}} \underbrace{\mathbb{E}\{ \hat{u}_{mk}^\ast u_{mk} \hat{u}_{m'k} u_{m'k}^{\ast}  \}}_{= Q_{mm'k}} \\
  - \sqrt{\alpha_{mk} \alpha_{m'k}} \gamma_{mk} \gamma_{m'k}.
\end{multline}
By using the analytical expressions of the projected training signal in \eqref{eq:ReceivedPilotv1} and the channel estimate in \eqref{eq:ChannelEst}, $Q_{mm'k}$ defined in \eqref{eq:omkomk} can be formulated in a closed-form expression, as follows
\begin{equation} \label{eq:Q}
\begin{split}
& Q_{mm'k} \mathop  = \limits^{\left( a \right)} c_{mk} c_{m'k}  \mathbb{E} \left\{ \left( \sqrt{p\tau_p} \sum_{k' \in \mathcal{P}_k} u_{mk'}^{\ast} + w_{pmk}  \right) u_{mk} \times \right.  \\
&\quad  \left.  \left( \sqrt{p\tau_p} \sum_{k' \in \mathcal{P}_k} u_{m'k'} + w_{pm'k}  \right) u_{m'k}^{\ast} \right\} \\
&\mathop  = \limits^{\left( b \right)} c_{mk} c_{m'k} p\tau_p \mathbb{E}\big\{ |u_{mk}|^2 |u_{m'k}|^2 \big\}  +   c_{mk} c_{m'k} p\tau_p   \times \\
& \sum_{k' \in \mathcal{P}_k \setminus \{ k\}} \mathbb{E} \{ u_{mk'}^\ast u_{mk} u_{m'k}^\ast u_{m'k'} \}\\
& = c_{mk} c_{m'k} p\tau_p \delta_{mk} \delta_{m'k} + c_{mk} c_{m'k} p\tau_p \sum_{k' \in \mathcal{P}_k} \mathrm{tr}( \pmb{\Theta}_{mk} \pmb{\Theta}_{m'k'} ),
\end{split}
\end{equation}
where $(a)$ is obtained by retaining only the terms whose expectation is not zero based on \eqref{eq:Uncorrelated} and $(b)$ follows by utilizing \eqref{eq:2UncorreChan}. From \eqref{eq:gammamk}, finally, we obtain
\begin{equation} \label{eq:Qv1}
	Q_{mm'k} = \gamma_{mk} \gamma_{m'k} + c_{mk} c_{m'k} p \tau_p \sum_{k' \in \mathcal{P}_k} \mathrm{tr}( \pmb{\Theta}_{mk} \pmb{\Theta}_{m'k'}  ). 
\end{equation}
The proof follows by inserting \eqref{eq:omkomk} in \eqref{eq:Qv1}. 

\vspace*{-0.2cm}
\subsection{Proof of Corollary~\ref{corollary:EqualPhase}} \label{appendix:CorMSEk}
\vspace*{-0.1cm}
We introduce the shorthand notation
$b_{mk} = p \tau_p \alpha_m \tilde{\alpha}_{k} d_H^2 d_V^2 $ and
$d_{mk} = p \tau_p d_H^2 d_V^2 \alpha_m \sum_{k' \in \mathcal{P}_k} \tilde{\alpha}_{k'}$. 
When the direct links are weak enough to be negligible, the NMSE of the channel estimate of the user~$k$ at the AP~$m$ can be reformulated as follows
\begin{equation}
\mathrm{NMSE}_{mk} = 1  - \frac{ b_{mk} \mathrm{tr}\big(\pmb{\Phi}^H \mathbf{R} \pmb{\Phi} \mathbf{R} \big)}{1 + d_{mk} \mathrm{tr}\big(\pmb{\Phi}^H \mathbf{R} \pmb{\Phi} \mathbf{R} \big)}.
\end{equation}
Let us denote by $f \big( \mathrm{tr}\big(\pmb{\Phi}^H \mathbf{R} \pmb{\Phi} \mathbf{R} \big) \big)= \sum_{m=1}^M \sum_{k=1}^K \mathrm{NMSE}_{mk}$ the objective function of the problem in \eqref{Prob:NMSEk}, the first-order derivative of $\mathrm{NMSE}_{mk}$ with respect to $\mathrm{tr}\big(\pmb{\Phi}^H \mathbf{R} \pmb{\Phi} \mathbf{R} \big)$ is
\begin{equation} \label{eq:1stDerivativev1}
\frac{ d f \big( \mathrm{tr}\big(\pmb{\Phi}^H \mathbf{R} \pmb{\Phi} \mathbf{R} \big) \big) }{d \mathrm{tr}\big(\pmb{\Phi}^H \mathbf{R} \pmb{\Phi} \mathbf{R} \big)} = - \sum_{m=1}^M \sum_{k=1}^K \frac{b_{mk}}{\left(1 + d_{mk} \mathrm{tr}\big(\pmb{\Phi}^H \mathbf{R} \pmb{\Phi} \mathbf{R} \big) \right)^2} < 0,
\end{equation}
which implies that the objective function is a monotonically decreasing function of $\mathrm{tr}\big(\pmb{\Phi}^H \mathbf{R} \pmb{\Phi} \mathbf{R} \big)$ since $b_{mk} \geq 0, \forall m,k$. Moreover, $\pmb{\Phi}^H \mathbf{R} \pmb{\Phi} \mathbf{R} $ is similar to $\mathbf{R}^{1/2} \pmb{\Phi}^H \mathbf{R} \pmb{\Phi} \mathbf{R}^{1/2} $, which is a positive semidefinite matrix. Thus, we obtain
\begin{equation} \label{eq:Tracev1}
\begin{split}
\mathrm{tr}\big(\pmb{\Phi}^H \mathbf{R} \pmb{\Phi} \mathbf{R} \big) &= \left| \mathrm{tr}\big(\pmb{\Phi}^H \mathbf{R} \pmb{\Phi} \mathbf{R} \big)  \right| =\left| \sum_{n=1}^N \sum_{n'=1}^N (r^{nn'})^2 e^{j(\theta_n - \theta_{n'})} \right|.
\end{split}
\end{equation}
Let us introduce the two vectors $\mathbf{a}, \mathbf{b} \in \mathbb{C}^{N^2}$ defined as follows
\begin{align} 
\mathbf{a} &= [r^{11}, \ldots, r^{nn'}, \ldots, r^{NN}]^T, \label{eq:avec} \\
\mathbf{b} &= [r^{11}e^{j (\theta_1 - \theta_1)}, \ldots, r^{nn'}e^{j (\theta_n - \theta_{n'})}, \ldots,  r^{NN}e^{j (\theta_N - \theta_N)}]^T. \label{eq:bvec}
\end{align}
With the aid of Cauchy-Schwarz's inequality, we obtain the following upper bound for \eqref{eq:Tracev1}
\begin{equation} \label{eq:CauchySwatz}
\begin{split}
\mathrm{tr}\big(\pmb{\Phi}^H \mathbf{R} \pmb{\Phi} \mathbf{R} \big) &= |\mathbf{a}^H \mathbf{b}| \leq \|\mathbf{a}\|  \|\mathbf{b}\| =  \sum_{n=1}^N \sum_{n'=1}^N (r^{nn'})^2 ,
\end{split}
\end{equation}
which holds with equality  if and only if the two vectors $\mathbf{a}$  and $\mathbf{b}$ in \eqref{eq:avec} and \eqref{eq:bvec} are parallel. This implies $\theta_{n} = \theta_{n'}, \forall n,n'$. By combining \eqref{eq:1stDerivativev1} and \eqref{eq:CauchySwatz}, the proof is concluded.
\vspace*{-0.2cm}
\subsection{Proof of Theorem~\ref{theorem:ULMR}} \label{appendix:ULMR}
\vspace*{-0.1cm}
To obtain the closed-form expression of the uplink SINR in \eqref{eq:ULSINR}, we first compute $|\mathsf{DS}_k|^2$ by using the definition of $u_{mk}$ in \eqref{eq:umk} as
\begin{equation} \label{eq:DSkv1}
\begin{split}
|\mathsf{DS}_{uk}|^2 & \stackrel{(a)}{=}  \rho_u\eta_k \left| \mathbb{E}  \left\{ \sum_{m=1}^M  \hat{u}_{mk}^\ast (\hat{u}_{mk} + e_{mk} )\right\} \right|^2 \\
& \stackrel{(b)}{=} \rho_u\eta_k \left|  \sum_{m=1}^M \mathbb{E}  \left\{ | \hat{u}_{mk}|^2 \right\} \right|^2=\rho_u \eta_{k} \left( \sum_{m=1}^M \gamma_{mk} \right)^2,
\end{split}
\end{equation}
where $(a)$ is obtained by expressing the original channel $u_{mk}$ into the summation of its channel estimate and its estimation error as stated in Lemma~\ref{lemma:ChannelEst}; and $(b)$ follows because the channel estimate and the channel estimation error are uncorrelated. Since the aggregated channels sharing the same AP index are correlated, the first expectation in the denominator of \eqref{eq:ULSINR} is equal to
\begin{equation} \label{eq:1Deno}
\begin{split}
&\mathbb{E} \{|\mathsf{BU}_{uk}|^2 \}
=  \rho_u \eta_k \mathbb{E} \left\{ \left| \sum_{m=1}^M o_{umk} \right|^2 \right\} \\
&= \underbrace{\rho_u \eta_k  \sum_{m=1}^M \sum_{m'=1, m' \neq m}^M \mathbb{E} \{ o_{umk} o_{um'k}^\ast\}}_{= \bar{T}_{u0} } +  \underbrace{\rho_u \eta_k  \sum_{m=1}^M \mathbb{E} \{ | o_{umk} |^2 \}}_{= \bar{T}_{u1}},
\end{split}
\end{equation}
where $o_{umk} = \hat{u}_{mk}^\ast u_{mk} - \mathbb{E}\{ \hat{u}_{mk}^\ast u_{mk}\}$ with $\omega_{mk} =1$.  The closed-form expression of $\bar{T}_{u0}$ defined in \eqref{eq:1Deno} is as follows
\begin{equation}
\bar{T}_{u0} = p \tau_p \rho_u \eta_k  \sum_{k' \in \mathcal{P}_k} \sum_{m=1}^M \sum_{m'=1, m' \neq m}^M c_{mk} c_{m'k} \mathrm{tr}(\pmb{\Theta}_{mk}\pmb{\Theta}_{m'k'}),
\end{equation}
thanks to \eqref{eq:omkmk} in Corollary~\ref{CorollaryReal}. The expectation of $\bar{T}_{u1}$ defined in \eqref{eq:1Deno} can be rewritten as follows
\begin{equation}
\begin{split}
&\bar{T}_{u1} = \rho_u\eta_k   \sum_{m=1}^M   \mathbb{E} \left\{ \left| \hat{u}_{mk}^\ast u_{mk}  -  \mathbb{E}  \left\{  \hat{u}_{mk}^\ast u_{mk}   \right\}\right|^2 \right\} \\
& \stackrel{(a)}{=} \rho_u\eta_k   \sum_{m=1}^M   \mathbb{E}  \left\{ \left| \hat{u}_{mk}^\ast u_{mk} \right|^2 \right\}  - \rho_u\eta_k  \sum_{m=1}^M  \left| \mathbb{E}  \left\{  \hat{u}_{mk}^\ast u_{mk}   \right\}\right|^2 \\
& \stackrel{(b)}{=}  \rho_u\eta_k   \sum_{m=1}^M   \mathbb{E}  \left\{ \left| \hat{u}_{mk}^\ast u_{mk} \right|^2 \right\}  - \rho_u\eta_k \sum_{m=1}^M  \gamma_{mk}^2,
\end{split}
\end{equation}
where $(a)$ is obtained by using the identity $\mathbb{E}\{|X - \mathbb{E} \{X\} |^2 \} = \mathbb{E}\{|X|^2\} - |\mathbb{E}\{X\}|^2$; and $(b)$ is obtained from $u_{mk}^\ast = \hat{u}_{mk}^\ast + e_{mk}^\ast$, by taking into account that the channel estimate and the channel estimation error are uncorrelated random variables as stated in Lemma~\ref{lemma:ChannelEst}. By replacing $\hat{u}_{mk}^\ast = c_{mk} y_{pmk}^\ast$ in the first expectation of \eqref{eq:1Deno}, we obtain
\begin{equation} \label{eq:BUkv1}
\begin{split}
&\mathbb{E}  \left\{ \left| \hat{u}_{mk}^\ast u_{mk} \right|^2 \right\} = c_{mk}^2  \mathbb{E}  \left\{ \left| y_{pmk}^\ast u_{mk} \right|^2 \right\} \\
&= c_{mk}^2  \mathbb{E}  \left\{ \left| \left( \sum_{k' \in \mathcal{P}_k} \sqrt{p\tau_p} u_{mk'}^\ast + w_{pmk}^\ast   \right)  u_{mk} \right|^2 \right\} \\
&=\underbrace{c_{mk}^2 p \tau_p \mathbb{E} \{ |u_{mk}|^4 \}}_{=T_{11}} + \underbrace{c_{mk}^2 p \tau_p \sum_{k' \in \mathcal{P}_k  \setminus \{ k\}} \mathbb{E} \{ |u_{mk'}^\ast u_{mk} |^2 \}}_{=T_{12}} \\
&\quad +  \underbrace{c_{mk}^2  \mathbb{E} \big\{ |w_{pmk}^\ast u_{mk} |^2 \big\}}_{= T_{13}}.
\end{split}
\end{equation}
Let us analyze the three terms, $T_{11}$, $T_{12}$, and $T_{13}$, in the last equality of \eqref{eq:BUkv1}. The first term  can be computed by exploiting the forth moment given in \eqref{eq:4Order}, as follows
\begin{equation} \label{eq:T1}
\begin{split}
& T_{11} = 2c_{mk}^2 p\tau_p \delta_{mk}^2 + 2 c_{mk}^2 p\tau_p \mathrm{tr}(\pmb{\Theta}_{mk}^2) \\
&\stackrel{(a)}{=} \gamma_{mk}^2+ c_{mk}^2 p\tau_p \delta_{mk}^2 +  2 c_{mk}^2 p\tau_p \mathrm{tr}(\pmb{\Theta}_{mk}^2),
\end{split}
\end{equation}
where $(a)$ is obtained by using the variance of the channel estimate in \eqref{eq:gammamk}, with $\delta_{mk}$ and $\xi_{mk}$ that are defined in the statement of the theorem. The second term can be computed by exploiting the uncorrelation of the two cascaded channels, as follows
\begin{equation}
\begin{split}
& T_{12} = c_{mk}^2 p \tau_p \sum_{k' \in \mathcal{P}_k  \setminus \{ k\}} \mathbb{E} \{ |u_{mk'}|^2  |u_{mk} |^2 \} \\
& = c_{mk}^2 p \tau_p \sum_{k' \in \mathcal{P}_k  \setminus \{ k\}} \delta_{mk'} \delta_{mk} +  c_{mk}^2 p \tau_p \sum_{k' \in \mathcal{P}_k  \setminus \{ k\}} \mathrm{tr}(\pmb{\Theta}_{mk'} \pmb{\Theta}_{mk}).
\end{split}
\end{equation}
The last term can be computed by exploiting the independence of the channel and noise, as follows
\begin{equation}\label{eq:T3}
\begin{split}
T_{13} &= c_{mk}^2 \mathbb{E} \{ |w_{pmk}|^2\}  \mathbb{E} \{ | u_{mk} |^2 \} \\
&= c_{mk}^2 \delta_{mk}.
\end{split}
\end{equation}
By inserting \eqref{eq:T1}--\eqref{eq:T3} into \eqref{eq:BUkv1} and with the aid of some algebraic steps, we obtain
\begin{equation} \label{eq:uhatu}
\begin{split}
&\mathbb{E}  \{ | \hat{u}_{mk}^\ast u_{mk} |^2 \} =  c_{mk}^2 p\tau_p \mathrm{tr}(\pmb{\Theta}_{mk}^2) + \gamma_{mk}^2  \\
&\quad + c_{mk} \sqrt{p\tau_p}  \delta_{mk}^2 + c_{mk}^2 p \tau_p \sum_{k' \in \mathcal{P}_k } \mathrm{tr}(\pmb{\Theta}_{mk'} \pmb{\Theta}_{mk}) \\
&   = c_{mk}^2 p\tau_p \mathrm{tr}(\pmb{\Theta}_{mk}^2) + \gamma_{mk}^2 + \gamma_{mk} \delta_{mk} \\
&\quad +  c_{mk}^2 p \tau_p \sum_{k' \in \mathcal{P}_k } \mathrm{tr}(\pmb{\Theta}_{mk'} \pmb{\Theta}_{mk}),
\end{split}
\end{equation}
where the final identity is obtained by using the relationship between $\gamma_{mk}$ and $c_{mk}$ in \eqref{eq:gammamk}. Combining \eqref{eq:1Deno} and \eqref{eq:uhatu}, the first term in the denominator of \eqref{eq:ULSINR} simplifies to
\begin{equation}\label{eq:BUkvc2}
\begin{split}
& \mathbb{E} \{|\mathsf{BU}_{uk}|^2 \} =   \rho_u \eta_k \sum_{m=1}^M \gamma_{mk} \delta_{mk} +  p \tau_p \rho_u \eta_k \sum_{k'' \in \mathcal{P}_k}\sum_{m=1}^M \sum_{m'=1}^M c_{mk} \\
& \times  c_{m'k} \mathrm{tr}(\pmb{\Theta}_{mk}\pmb{\Theta}_{m'k''}) +   p \tau_p \rho_u \eta_k \sum_{m=1}^M c_{mk}^2 \mathrm{tr}(\pmb{\Theta}_{mk}^2).
\end{split}
\end{equation}
The second term in the denominator of \eqref{eq:ULSINR} can be  split into the two terms based on the pilot reuse pattern defined in $\mathcal{P}_k$, as follows
\begin{multline} \label{eq:2Termv1}
\sum_{k'=1, k' \neq k}^K \mathbb{E} \{|\mathsf{UI}_{uk'k}|^2 \} = \sum_{k'\notin \mathcal{P}_k } \mathbb{E} \{|\mathsf{UI}_{uk'k}|^2 \} \\
 + \sum_{k' \in \mathcal{P}_k \setminus \{k\}} \mathbb{E} \{|\mathsf{UI}_{uk'k}|^2\}.
\end{multline}
The first term on the right-hand side of \eqref{eq:2Termv1} is the non-coherent interference, which is equal to
\begin{equation} \label{eq:UIkkv1}
\begin{split}
& \sum_{k'\notin \mathcal{P}_k } \mathbb{E} \{|\mathsf{UI}_{uk'k}|^2 \} = \rho_u \sum_{k'\notin \mathcal{P}_k } \eta_{k'} \mathbb{E} \left\{\left| \sum_{m=1}^M \hat{u}_{mk}^\ast u_{mk'} \right|^2 \right\} \\
&= \rho_u \sum_{k'\notin \mathcal{P}_k } \sum_{m=1}^M \sum_{m'=1}^M \eta_{k'} \underbrace{\mathbb{E} \left\{  \hat{u}_{mk}^\ast u_{mk'} \hat{u}_{m'k} u_{m'k'}^\ast  \right\}}_{= T_{mk'm'k}}.
\end{split}
\end{equation}
We compute each expectation $T_{mk'm'k}$ by utilizing the channel estimate in \eqref{eq:ChannelEst} with the aid of some algebraic manipulations, as follows
\begin{equation} \label{eq:Tmkprmprk}
\begin{split}
T_{mk'm'k} =& p \tau_p c_{mk} c_{m'k} \sum_{k'' \in \mathcal{P}_k}   \mathbb{E} \{u_{mk''}^{\ast} u_{mk'} u_{m' k''} u_{m' k'}^{\ast} \} \\
&+ c_{mk} c_{m'k}\mathbb{E} \{w_{pmk}^{\ast} u_{mk'} w_{pm' k} u_{m' k'}^{\ast} \}\\
=& \begin{cases}
  p \tau_p c_{mk} c_{m'k} \sum\limits_{k'' \in \mathcal{P}_k}  \bar{T}_{mk'm'k''}, & \mbox{if } m \neq m',\\
   \bar{T}_{mk'mk}, & \mbox{if } m= m'.
\end{cases}
\end{split}
\end{equation}
where the following definitions hold
\begin{align}
& \bar{T}_{mk'm'k''} = \mathbb{E} \{u_{mk''}^{\ast} u_{mk'} u_{m' k''} u_{m' k'}^{\ast} \},\\
& \bar{T}_{mk'mk} = p \tau_p c_{mk}^2 \sum_{k'' \in \mathcal{P}_k}   \mathbb{E} \{|u_{mk''}^{\ast} u_{mk'}|^2 \} + c_{mk}^2 \mathbb{E} \{ | u_{mk'}|^2 \} .
\end{align}
The closed-form expression of $\bar{T}_{mk'm'k''}$ is obtained by utilizing the uncorrelated property of the quadruple aggregated channels in \eqref{eq:4Chan}, as follows
\begin{equation} \label{eq:Tmkprmk}
\bar{T}_{mk'm'k''} = \mathrm{tr}(\pmb{\Theta}_{mk'} \pmb{\Theta}_{m'k''}).
\end{equation}
In addition, the closed-form expression of $\bar{T}_{mk'mk}$ can be computed by utilizing the results in Lemma~\ref{lemma:ChannelProperty}, as follows
\begin{equation} \label{eq:Tmkk}
\begin{split}
&\bar{T}_{mk'mk} = c_{mk}^2  p \tau_p \sum_{k'' \in \mathcal{P}_k} \delta_{mk''}\delta_{mk'}  + c_{mk}^2  p \tau_p \times \\
&\quad \sum_{k'' \in \mathcal{P}_k} \mathrm{tr}(\pmb{\Theta}_{mk''} \pmb{\Theta}_{mk'})+ c_{mk}^2 \delta_{mk'} \\
&  = \gamma_{mk} \delta_{mk'} + c_{mk}^2  p \tau_p \sum_{k'' \in \mathcal{P}_k} \mathrm{tr}( \pmb{\Theta}_{mk'} \pmb{\Theta}_{mk''}).
\end{split}
\end{equation}
Inserting \eqref{eq:Tmkprmk} and \eqref{eq:Tmkk} into \eqref{eq:Tmkprmprk}, and with the aid of some algebraic manipulations,  \eqref{eq:UIkkv1} can be equivalently formulated as follows
\begin{multline} \label{eq:EUIv1}
\sum_{k'\notin \mathcal{P}_k } \mathbb{E} \{|\mathsf{UI}_{uk'k}|^2 \} = \rho_u \sum_{k'\notin \mathcal{P}_k }\sum_{m=1}^M \eta_{k'} \gamma_{mk} \delta_{mk'} +   p \tau_p \rho_u \times \\  \sum_{k'\notin \mathcal{P}_k } \sum_{k'' \in \mathcal{P}_k}\sum_{m=1}^M  \sum_{m'=1}^M \eta_{k'} c_{mk} c_{m'k} \mathrm{tr}(\pmb{\Theta}_{mk'} \pmb{\Theta}_{m'k''} ).
\end{multline}
The second term on the right-hand side of \eqref{eq:2Termv1} is the coherent interference. By using \eqref{eq:ReceivedPilotv1} and  \eqref{eq:ChannelEst}, it simplifies as follows
\begin{equation} \label{eq:UIv1}
\begin{split}
&\mathbb{E} \{|\mathsf{UI}_{uk'k}|^2\} =  \mathbb{E} \left\{\left| \sum_{m=1}^M c_{mk} y_{pmk}^\ast u_{mk'} \right|^2 \right\} \\
&=   \mathbb{E}  \left\{ \left|\sum_{m=1}^M c_{mk} \left( \sum_{k'' \in \mathcal{P}_k} \sqrt{p\tau_p} u_{mk''}^\ast + w_{pmk}^\ast   \right)  u_{mk'} \right|^2 \right\} \\
&= \rho_u \eta_{k'} \mathbb{E}  \left\{ \left|\sum_{m=1}^M c_{mk} w_{pmk}^\ast u_{mk'} \right|^2 \right\} \\
&\quad + \rho_u \eta_{k'} p\tau_p \underbrace{\mathbb{E}  \left\{ \left|\sum_{m=1}^M c_{mk} \left( \sum_{k'' \in \mathcal{P}_k \setminus \{k'\}}  u_{mk''}^\ast \right)  u_{mk'} \right|^2 \right\}}_{=T_{u21}}\\
&\quad +  \rho_u \eta_{k'} p\tau_p \underbrace{ \mathbb{E}  \left\{ \left|\sum_{m=1}^M c_{mk}  |u_{mk'}|^2 \right|^2 \right\}}_{=T_{u22}}\\
& = \rho_u \eta_{k'} \sum_{m=1}^M c_{mk}^2 \delta_{mk'} + \rho_u \eta_{k'} p\tau_p (T_{u21} + T_{u22}) ,
\end{split} 
\end{equation}
which is obtained by using the identity $\mathbb{E}\{ |X+Y|^2 \} =\mathbb{E}\{ |X|^2 \} + \mathbb{E}\{ |Y|^2 \}$ that holds true for zero-mean and uncorrelated random variables. The expectation $T_{u21}$ can be simplified as follows
\begin{equation} \label{eq:T21}
\begin{split}
 & T_{u21} = \sum_{k'' \in \mathcal{P}_k \setminus \{k'\} } \sum_{m=1}^M  c_{mk}^2 \mathbb{E}\{ |u_{mk''}|^2 |u_{mk'}|^2 \} + \\
 &\sum_{k'' \in \mathcal{P}_k \setminus \{k'\} } \sum_{m=1}^M  \sum_{m'=1, m' \neq m }^M c_{mk}  c_{m'k} \mathbb{E}\{ u_{mk''}^\ast u_{mk'}  u_{m'k''} u_{m'k'}^\ast \} \\
&= \sum_{k'' \in \mathcal{P}_k \setminus \{k'\} } \sum_{m=1}^M  c_{mk}^2 \delta_{mk''} \delta_{mk'} + \\
& \sum_{k'' \in \mathcal{P}_k \setminus \{k'\} } \sum_{m=1}^M \sum_{m'=1}^M c_{mk} c_{m'k} \mathrm{tr}(\pmb{\Theta}_{mk'} \pmb{\Theta}_{m'k''} ).
\end{split}
\end{equation}
Similarly, the expectation $T_{u22}$ in \eqref{eq:UIv1} can be simplified as follows 
\begin{equation*} 
\begin{split}
&T_{u22} =  \sum_{m=1}^M \sum_{m'=1}^M c_{mk} c_{m'k} \mathbb{E} \{ |u_{mk'}|^2 |u_{m'k'}|^2\} \\
&= \sum_{m=1}^M  c_{mk}^2 \mathbb{E} \{ |u_{mk'}|^4\} + \sum_{m=1}^M \sum_{m'=1, m' \neq m}^M c_{mk} c_{m'k} \mathbb{E} \{ |u_{mk'}|^2 |u_{m'k'}|^2\} \\
&= 2 \sum_{m=1}^M c_{mk}^2 \delta_{mk'}^2 + 2 \sum_{m=1}^M c_{mk}^2 \mathrm{tr}(\pmb{\Theta}_{mk'}^2) + \sum_{m=1}^M \sum_{m'=1, m' \neq m}^M c_{mk} \times\\
&  c_{m'k} \delta_{mk'} \delta_{m'k'}+ \sum_{m=1}^M \sum_{m'=1, m' \neq m}^M c_{mk} c_{m'k} \mathrm{tr} (\pmb{\Theta}_{mk'} \pmb{\Theta}_{m'k'})
\end{split}
\end{equation*}
\begin{equation} \label{eq:T22}
\begin{split}
&= \sum_{m=1}^M c_{mk}^2 \delta_{mk'}^2 + \sum_{m=1}^M c_{mk}^2\mathrm{tr}(\pmb{\Theta}_{mk'}^2)  + \left(\sum_{m=1}^M c_{mk} \delta_{mk'} \right)^2\\
& + \sum_{m=1}^M \sum_{m'=1}^M c_{mk} c_{m'k} \mathrm{tr} (\pmb{\Theta}_{mk'} \pmb{\Theta}_{m'k'}).
\end{split}
\end{equation}
By inserting \eqref{eq:T21} and \eqref{eq:T22} into \eqref{eq:UIv1}, and with the aid of some algebraic manipulations, we obtain
\begin{equation}
\begin{split}
& \rho_{u} \eta_{k'} \sum_{m=1}^M c_{mk}^2 \delta_{mk'} + \rho_u \eta_{k'} p \tau_p \sum_{k'' \in \mathcal{P}_{k}} \sum_{m=1}^M c_{mk}^2 \delta_{mk''} \delta_{mk'} \\
& = \rho_{u} \eta_{k'} \sum_{m=1}^M c_{mk}^2 \delta_{mk'} \left(1 + p \tau_p \sum_{k'' \in \mathcal{P}_{k}} \delta_{m''k} \right)\\
& \stackrel{(a)}{=} \rho_u \eta_{k'} \sum_{m=1}^M c_{mk} \sqrt{p\tau_p} \delta_{mk} \delta_{mk'} \stackrel{(b)}{=} \rho_u \eta_{k'} \sum_{m=1}^M \gamma_{mk} \delta_{mk'},
\end{split}
\end{equation}
 where $(a)$ is obtained by using \eqref{eq:cmk} and $(b)$ by using \eqref{eq:gammamk}. Therefore, the total mutual interference between the user~$k'$ and the user~$k$ who share the same pilot sequence can be written as follows
\begin{equation} \label{eq:UIkprk}
\begin{split}
&\sum_{k' \in \mathcal{P}_k \setminus \{k \} } \mathbb{E} \{|\mathsf{UI}_{uk'k}|^2\} = \rho_u \sum_{k' \in  \mathcal{P}_k \setminus \{k \} } \sum_{m=1}^M \eta_{k'} \gamma_{mk} \delta_{mk'}   \\
&+ p \tau_p \rho_u \sum_{k' \in  \mathcal{P}_k \setminus \{k \} } \sum_{m=1}^M  \eta_{k'} c_{mk}^2 \mathrm{tr}(\pmb{\Theta}_{mk'}^2) \\
&+ p \tau_p \rho_u \sum_{k' \in  \mathcal{P}_k \setminus \{k \} } \sum_{k'' \in \mathcal{P}_k}  \sum_{m=1}^M  \sum_{m'=1}^M \eta_{k'} c_{mk} c_{m'k} \times \\
&\mathrm{tr}( \pmb{\Theta}_{mk'} \pmb{\Theta}_{m'k''})+ p \tau_p \rho_u \sum_{k' \in \mathcal{P}_k \setminus \{ k\} } \eta_{k'} \left(\sum_{m=1}^M c_{mk} \delta_{mk'} \right)^2.
\end{split}
\end{equation}
Let us denote $\mathsf{I}_{uk} = \sum_{k'=1, k' \neq k}^K \mathbb{E} \{ |\mathsf{UI}_{uk'k}|^2 \}  $. By combing \eqref{eq:EUIv1} and \eqref{eq:UIkprk}, the mutual interference at the user~$k$ can be formulated in a closed-form expression as follows
\begin{equation}
\begin{split}
&\mathsf{I}_{uk} = \rho_u \sum_{k'=1, k' \neq k}^K \sum_{m=1}^M \eta_{k'} \gamma_{mk} \delta_{mk'} + p \tau_p \rho_u \times \\
& \sum_{k' =1, k' \neq k }^K \sum_{k'' \in \mathcal{P}_k}    \sum_{m=1}^M  \sum_{m'=1}^M\eta_{k'} c_{mk}c_{m'k} \mathrm{tr}( \pmb{\Theta}_{mk'} \pmb{\Theta}_{m'k''}) +\\
&+  p \tau_p \rho_u \sum_{k' \in  \mathcal{P}_k \setminus \{ k \} } \sum_{m=1}^M  \eta_{k'} c_{mk}^2 \mathrm{tr}(\pmb{\Theta}_{mk'}^2) + p \tau_p \rho_u \times \\
& \sum_{k' \in \mathcal{P}_k \setminus \{ k\} } \eta_{k'} \left(\sum_{m=1}^M c_{mk} \delta_{mk'} \right)^2.
\end{split}
\end{equation}
Finally, the expectation of the additive noise after MR processing can be written as follows
\begin{equation} \label{eq:ULNoise}
\begin{split}
&\mathbb{E}\{|\mathsf{NO}_{uk}|^2 \}= \sum_{m=1}^M \mathbb{E} \{|\hat{u}_{mk}^\ast w_{um}|^2 \}\\
& = \sum_{m=1}^M \mathbb{E} \{|\hat{u}_{mk}|^2\} \mathbb{E} \{|w_{um}|^2 \} = \sum_{m=1}^M \gamma_{mk},
\end{split}
\end{equation}
thanks to the independence between the channel estimate and the noise. The proof follows by inserting \eqref{eq:DSkv1}, \eqref{eq:BUkv1}, \eqref{eq:UIkprk}, and \eqref{eq:ULNoise} into the  SINR in \eqref{eq:ULSINR} with the aid of some algebraic manipulations.
\vspace*{-0.2cm}
\subsection{Proof of Theorem~\ref{theorem:DLMR}}\label{appendix:DLMR}
\vspace*{-0.1cm}
Consider the downlink SINR in \eqref{eq:DLSINR}. Thanks to the uncorrelation between the channel estimate and the channel estimation error, the numerator simplifies to
\begin{equation} \label{eq:DSdk}
|\mathsf{DS}_{dk}|^2 = \rho_{d} \left|   \sum_{m=1}^M  \sqrt{\eta_{mk}} \mathbb{E}\{ |\hat{u}_{mk}|^2 \} \right|^2 = \rho_{d} \left(   \sum_{m=1}^M  \sqrt{\eta_{mk}} \gamma_{mk} \right)^2.
\end{equation}
The beamforming uncertainty term in the denominator of \eqref{eq:DLSINR} can be simplified by using \eqref{eq:ChannelEst}, as follows
\begin{equation} \label{eq:BUdk}
\begin{split}
& \mathbb{E} \{|\mathsf{BU}_{dk}|^2 \} = \rho_d \mathbb{E} \left\{\left| \sum_{m=1}^M o_{dmk} \right|^2 \right\} \\
& = \underbrace{\rho_d \sum_{m=1}^M \sum_{m'=1, m' \neq m}^M \mathbb{E} \{ o_{dmk} o_{dm'k}^{\ast} \}}_{= \bar{T}_{d0}} + \underbrace{\rho_d \sum_{m=1}^M \mathbb{E} \{ |o_{dmk}|^2 \}}_{= \bar{T}_{d1}},
\end{split}
\end{equation} 
where $o_{dmk} = \sqrt{\eta_{mk}} u_{mk} u_{mk}^{\ast} - \sqrt{\eta_{mk}} \mathbb{E} \{ u_{mk} u_{mk}^{\ast} \}$ with $\omega_{mk} = \eta_{mk}$ (see Corollary~\ref{CorollaryReal} for  $\omega_{mk}$). The closed-form expression of $\bar{T}_{d0}$ in \eqref{eq:BUdk} can be formulated as follows
\begin{equation} \label{eq:Td0}
\begin{split}
&\bar{T}_{d0} = \\
& p \tau_p \rho_d  \sum_{k'' \in \mathcal{P}_k} \sum_{m=1}^M \sum_{m'=1, m' \neq m}^M \sqrt{\eta_{mk} \eta_{m'k}}  c_{mk} c_{m'k} \mathrm{tr}(\pmb{\Theta}_{mk} \pmb{\Theta}_{m'k''}),
 \end{split}
\end{equation}
by utilizing \eqref{eq:omkmk} in Corollary~\ref{CorollaryReal}. The expectation $\bar{T}_{d1}$ in \eqref{eq:Td0} can be rewritten as follows
\begin{equation} \label{eq:Td1bar}
\begin{split}
\bar{T}_{d1} &= \underbrace{\rho_d   \sum_{m=1}^M \eta_{mk}  \mathbb{E}\big\{ |u_{mk} \hat{u}_{mk}^\ast |^2 \big\}}_{=T_{d1}}  - \rho_{d}   \sum_{m=1}^M  \eta_{mk} \big|\mathbb{E}\{ |\hat{u}_{mk}|^2\}\big|^2 \\
&=  T_{d1} - \rho_d \sum_{m=1}^M \eta_{mk} \gamma_{mk}^2,
\end{split}
\end{equation}
where we have used the identities $\mathbb{E}\{ |X - \mathbb{E}\{ X\} |^2 \} = \mathbb{E}\{|X|^2 \} - |\mathbb{E}\{ X\}|^2$ and $\mathbb{E}\{ |X+Y|^2 \} = \mathbb{E}\{ |X|^2\} + \mathbb{E}\{ Y|^2 \}$ for zero-mean uncorrelated random variables. By using the identities in \eqref{eq:BUkv1} and \eqref{eq:uhatu}, $T_{d1}$ in \eqref{eq:Td1bar} can be formulated as follows
\begin{equation} \label{eq:Td1}
\begin{split}
& T_{d1} = p\tau_p\rho_d \sum_{m=1}^M \eta_{mk} c_{mk}^2 \mathrm{tr}(\pmb{\Theta}_{mk}^2)  + \rho_d \sum_{m=1}^M \eta_{mk}\gamma_{mk}^2 +  \rho_d \times\\
&  \sum_{m=1}^M \eta_{mk} \gamma_{mk} \delta_{mk}  + p \tau_p \rho_d \sum_{k'' \in \mathcal{P}_k} \sum_{m=1}^M \eta_{mk} c_{mk}^2 \mathrm{tr}(\pmb{\Theta}_{mk}\pmb{\Theta}_{mk''} ) ,
  \end{split}
\end{equation}
by using \eqref{eq:ReceivedPilotv1}. Inserting \eqref{eq:Td0} and \eqref{eq:Td1} into \eqref{eq:BUdk}, we obtain
\begin{equation}
\begin{split}
&\mathbb{E} \{|\mathsf{BU}_{dk}|^2 \} = \rho_d \sum_{m=1}^M \eta_{mk} \gamma_{mk} \delta_{mk} + p \tau_p \rho_d \times \\
& \sum_{k'' \in \mathcal{P}_k} \sum_{m=1}^M \sum_{m'=1}^M \sqrt{\eta_{mk} \eta_{m'k}}  c_{mk} c_{m'k} \mathrm{tr}(\pmb{\Theta}_{mk} \pmb{\Theta}_{m'k''})\\
& + p \tau_p \rho_d \sum_{m=1}^M \eta_{mk} c_{mk}^2 \mathrm{tr}(\pmb{\Theta}_{mk}^2).
\end{split}
\end{equation}
The mutual interference term in the denominator of \eqref{eq:DLSINR} can be rewritten as follows
\begin{equation} \label{eq:UIdkprk}
\begin{split}
&\sum_{k'=1, k' \neq k}^K \mathbb{E} \{|\mathsf{UI}_{dk'k}|^2 \} =  \sum_{k'\in \mathcal{P}_k \setminus \{k\}} \underbrace{\rho_d \mathbb{E} \{|\mathsf{UI}_{dk'k}|^2 \}}_{=T_d} + \\
&\rho_d \sum_{k' \notin \mathcal{P}_k } \sum_{m=1}^M \eta_{mk'} \gamma_{mk'}\delta_{mk} +  p \tau_p \rho_d \times \\
& \sum_{k' \notin \mathcal{P}_k } \sum_{k'' \in \mathcal{P}_{k'}} \sum_{m=1}^M \sum_{m'=1}^M \sqrt{\eta_{mk'} \eta_{m'k'} } c_{mk'} c_{m'k'} \mathrm{tr} ( \pmb{\Theta}_{mk}\pmb{\Theta}_{m'k''}),
\end{split}
\end{equation}
where the first term on the right-hand side of \eqref{eq:UIdkprk} is obtained by exploiting the orthogonality of the pilot sequences. In the second summation of \eqref{eq:UIdkprk}, $T_d = \mathbb{E} \{ |\mathsf{UI}_{dkk'}|^2 \}$ can be rewritten as follows
\begin{equation} \label{eq:Tdv1}
\begin{split}
&T_d =  \rho_{d} \mathbb{E} \left\{\left|\sum_{m=1}^M \sqrt{\eta_{mk'}} u_{mk}\hat{u}_{mk'}^\ast \right|^2 \right\} \\
&= \rho_{d} \mathbb{E} \left\{\left|\sum_{m=1}^M \sqrt{\eta_{mk'}} c_{mk'} u_{mk}y_{pmk'}^\ast \right|^2 \right\}\\
&= \rho_{d}\mathbb{E} \left\{\left|\sum_{m=1}^M \sqrt{\eta_{mk'}} c_{mk'} u_{mk} \left( \sum_{k'' \in \mathcal{P}_{k'}} \sqrt{p\tau_p} u_{mk''}^\ast + w_{pmk'}^\ast   \right) \right|^2 \right\} \\
&= \rho_d \sum_{m=1}^M \eta_{mk'} c_{mk'}^2 \mathbb{E}\{|u_{mk} w_{pmk'}^\ast |^2 \} 
\\
&+ \rho_d p \tau_p  \underbrace{\mathbb{E}\left\{\left|\sum_{m=1}^M \sqrt{\eta_{mk'}} c_{mk'} u_{mk} \left( \sum_{k'' \in \mathcal{P}_{k'} \setminus \{ k\}}  u_{mk''}^\ast   \right) \right|^2 \right\}}_{=T_{d21}} + \\
& \rho_d p \tau_p \underbrace{\mathbb{E}\left\{\left|\sum_{m=1}^M \sqrt{\eta_{mk'}} c_{mk'} |u_{mk}|^2 \right|^2 \right\}}_{=T_{d22}}\\
&= \rho_d \sum_{m=1}^M \eta_{mk'} c_{mk'}^2 \delta_{mk}+ p \tau_p \rho_d  ( T_{d21} + T_{d22}),
\end{split}
\end{equation}
which is obtained by taking into account that the noise is circularly symmetric. The term $T_{d21}$ depends on the non-coherent interference and can be simplified as follows
\begin{equation} \label{eq:Td21}
\begin{split}
& T_{d21} = \sum_{k'' \in \mathcal{P}_{k'} \setminus \{ k\}}  \sum_{m=1}^M  \eta_{mk'} c_{mk'}^2 \mathbb{E}\{|u_{mk}   u_{mk''}^\ast  |^2 \} \\
&+  \sum_{k'' \in \mathcal{P}_{k'} \setminus \{ k\}}  \sum_{m=1}^M \sum_{m'=1, m' \neq m}^M  \sqrt{\eta_{mk'} \eta_{m'k'}} c_{mk'}  c_{m'k'}  \\
& \times \mathbb{E}\{   u_{mk''}^\ast u_{mk}  u_{m'k}^\ast  u_{m'k''}  \} \\
 = & \sum_{k'' \in \mathcal{P}_{k'} \setminus \{ k\}}  \sum_{m=1}^M  \eta_{mk'} c_{mk'}^2 \delta_{mk}\delta_{mk''} + \\
 &\sum_{k'' \in \mathcal{P}_{k'} \setminus \{ k\}}  \sum_{m=1}^M \sum_{m'=1}^M  \sqrt{\eta_{mk} \eta_{m'k'}} c_{mk'} c_{m'k'} \mathrm{tr}(\pmb{\Theta}_{mk}\pmb{\Theta}_{m'k''}),
\end{split}
\end{equation}
by using the second moment in \eqref{eq:2Order} and the uncorrelation among the aggregated channels. The last term $T_{d22}$ in \eqref{eq:Tdv1} can be simplified as follows
\begin{equation} \label{eq:Td22}
\begin{split}
& T_{d22} = \sum_{m=1}^M \sum_{m'=1}^M \sqrt{\eta_{mk'} \eta_{m'k'}} c_{mk'}c_{m'k'} \mathbb{E}\left\{|u_{mk}|^2 |u_{m'k}|^2   \right\}\\
& = \sum_{m=1}^M \eta_{mk'} c_{mk'}^2 \mathbb{E}\left\{|u_{mk}|^4  \right\} + \\
&\sum_{m=1}^M \sum_{m'=1, m' \neq m}^M \sqrt{\eta_{mk'} \eta_{m'k'}} c_{mk'}c_{m'k'} \mathbb{E}\{|u_{mk}|^2 |u_{m'k}|^2\}\\
&=  \sum_{m=1}^M \eta_{mk'} c_{mk'}^2 \delta_{mk}^2 + \sum_{m=1}^M \eta_{mk'} c_{mk'}^2 \mathrm{tr}(\pmb{\Theta}_{mk}^2)\\
& + \left(\sum_{m=1}^M  \sqrt{\eta_{mk'} } c_{mk'}  \delta_{mk} \right)^2
\\
&+ \sum_{m=1}^M \sum_{m'=1}^M \sqrt{\eta_{mk'} \eta_{m'k'}} c_{mk'}c_{m'k'} \mathrm{tr}(\pmb{\Theta}_{mk}\pmb{\Theta}_{m'k})
\end{split}
\end{equation}
where the last equality follows from Lemma~\ref{lemma:ChannelProperty}. By inserting \eqref{eq:Td21} and \eqref{eq:Td22} into \eqref{eq:Tdv1}, and by denoting $\mathsf{I}_{dk} = \sum_{k'=1, k' \neq k}^K \mathbb{E} \{|\mathsf{UI}_{dk'k}|^2 \} $, \eqref{eq:UIdkprk} can be rewritten as follows
\begin{equation} \label{eq:UIdkprkv1}
\begin{split}
&\mathsf{I}_{dk} = \rho_d \sum_{k' =1, k' \neq k}^K \sum_{m=1}^M \eta_{mk'}  \gamma_{mk'} \delta_{mk} + p \tau_p \rho_d \times \\
&   \sum_{k' =1, k' \neq k }^K \sum_{k'' \in \mathcal{P}_{k'}} \sum_{m=1}^M \sum_{m'=1}^M \sqrt{\eta_{mk'} \eta_{m'k'}} c_{mk'}  c_{m'k'} \mathrm{tr} (\pmb{\Theta}_{mk} \pmb{\Theta}_{mk''} ) \\
&+ p \tau_p \rho_d  \sum_{k' \in \mathcal{P}_k \setminus \{ k \}} \sum_{m=1}^M \eta_{mk'}  c_{mk'}^2 \mathrm{tr}(\pmb{\Theta}_{mk}^2) + \\
& p \tau_p \rho_d  \sum_{k' \in \mathcal{P}_k \setminus \{ k\}} \left(\sum_{m=1}^M  \sqrt{\eta_{mk'} } c_{mk'}  \delta_{mk} \right)^2,
\end{split}
\end{equation}
The proof follows, by inserting \eqref{eq:DSdk}, \eqref{eq:BUdk}, and \eqref{eq:UIdkprkv1} into \eqref{eq:DLSINR} and by using some algebraic manipulations.
\vspace*{-0.2cm}
 \bibliographystyle{IEEEtran}
 \bibliography{IEEEabrv,refs}
 \begin{IEEEbiographynophoto} 
 	{Trinh Van Chien} (S'16-M'20) received the B.S. degree in Electronics and Telecommunications from Hanoi University of Science and Technology (HUST), Vietnam, in 2012. He then received the M.S. degree in Electrical and Computer Enginneering from Sungkyunkwan University (SKKU), Korea, in 2014 and the Ph.D. degree in Communication Systems from Link\"oping University (LiU), Sweden, in 2020. He is now a research associate at University of Luxembourg. His interest lies in convex optimization problems and machine learning applications for wireless communications and image \& video processing. He was an IEEE wireless communications letters exemplary reviewer for 2016 and 2017. He also received the award of scientific excellence in the first year of the 5Gwireless project funded by European Union Horizon's 2020.
 \end{IEEEbiographynophoto}

\begin{IEEEbiographynophoto}
	{Hien Quoc Ngo}  received the B.S. degree in electrical engineering from the Ho Chi Minh City University of Technology, Vietnam, in 2007, the M.S. degree in electronics and radio engineering from Kyung Hee University, South Korea, in 2010, and the Ph.D. degree in communication systems from Link\"oping University (LiU), Sweden, in 2015. In 2014, he visited the Nokia Bell Labs, Murray Hill, New Jersey, USA. From January 2016 to April 2017, Hien Quoc Ngo was a VR researcher at the Department of Electrical Engineering (ISY), LiU. He was also a Visiting Research Fellow at the School of Electronics, Electrical Engineering and Computer Science, Queen's University Belfast, UK, funded by the Swedish Research Council.
	
	Hien Quoc Ngo is currently a Reader (Associate Professor) at Queen's University Belfast, UK. His main research interests include massive (large-scale) MIMO systems, cell-free massive MIMO, physical layer security, and cooperative communications. He has co-authored many research papers in wireless communications and co-authored the Cambridge University Press textbook \emph{Fundamentals of Massive MIMO} (2016).
	
	Dr. Hien Quoc Ngo received the IEEE ComSoc Stephen O. Rice Prize in Communications Theory in 2015, the IEEE ComSoc Leonard G. Abraham Prize in 2017, and the Best PhD Award from EURASIP in 2018. He also received the IEEE Sweden VT-COM-IT Joint Chapter Best Student Journal Paper Award in 2015. He was an \emph{IEEE Communications Letters} exemplary reviewer for 2014, an \emph{IEEE Transactions on Communications} exemplary reviewer for 2015, and an \emph{IEEE Wireless Communications Letters} exemplary reviewer for 2016.  He was awarded the UKRI Future Leaders Fellowship in 2019.
	Dr. Hien Quoc Ngo currently serves as an Editor for the IEEE Transactions on Wireless Communications, IEEE Wireless Communications Letters, Digital Signal Processing, Elsevier Physical Communication (PHYCOM), and IEICE Transactions on Fundamentals of Electronics, Communications and Computer Sciences. He was a Guest Editor of IET Communications, special issue on ``Recent Advances on 5G Communications'' and a Guest Editor of  IEEE Access, special issue on ``Modelling, Analysis, and Design of 5G Ultra-Dense Networks'', in 2017. He has been a member of Technical Program Committees for several IEEE conferences such as ICC, GLOBECOM, WCNC, and VTC.
\end{IEEEbiographynophoto}

 \begin{IEEEbiographynophoto} 
{Symeon Chatzinotas} is Full Professor and Head of the SIGCOM Research Group at SnT, University of Luxembourg. He is coordinating the research activities on communications and networking, acting as a PI for more than 20 projects and main representative for 3GPP, ETSI, DVB. He is currently serving in the editorial board of the IEEE Transactions on Communications, IEEE Open Journal of Vehicular Technology and the International Journal of Satellite Communications and Networking.

In the past, he has been a Visiting Professor at the University of Parma, Italy and was involved in numerous R\&D projects for NCSR Demokritos, CERTH Hellas and CCSR, University of Surrey.

He was the co-recipient of the 2014 IEEE Distinguished Contributions to Satellite Communications Award and Best Paper Awards at EURASIP JWCN, CROWNCOM, ICSSC. He has (co-)authored more than 500 technical papers in refereed international journals, conferences and scientific books. 
\end{IEEEbiographynophoto}

\begin{IEEEbiographynophoto} 
	{Marco Di Renzo} (Fellow, IEEE) received the Laurea (cum laude) and
	Ph.D. degrees in electrical engineering from the University of
	L’Aquila, Italy, in 2003 and 2007, respectively, and the Habilitation
	à Diriger des Recherches (Doctor of Science) degree from University
	Paris-Sud (now Paris-Saclay University), France, in 2013. Since 2010,
	he has been with the French National Center for Scientific Research
	(CNRS), where he is a CNRS Research Director (Professor) with the
	Laboratory of Signals and Systems (L2S) of Paris-Saclay University –
	CNRS and CentraleSupelec, Paris, France. In Paris-Saclay University,
	he serves as the Coordinator of the Communications and Networks
	Research Area of the Laboratory of Excellence DigiCosme, and as a
	Member of the Admission and Evaluation Committee of the Ph.D. School
	on Information and Communication Technologies. He is the
	Editor-in-Chief of IEEE Communications Letters and a Distinguished
	Speaker of the IEEE Vehicular Technology Society. In 2017-2020, he was
	a Distinguished Lecturer of the IEEE Vehicular Technology Society and
	IEEE Communications Society. He has received several research
	distinctions, which include the SEE-IEEE Alain Glavieux Award, the
	IEEE Jack Neubauer Memorial Best Systems Paper Award, the Royal
	Academy of Engineering Distinguished Visiting Fellowship, the Nokia
	Foundation Visiting Professorship, the Fulbright Fellowship, and the
	2021 EURASIP Journal on Wireless Communications and Networking Best
	Paper Award. He is a Fellow of the UK Institution of Engineering and
	Technology (IET), a Fellow of the Asia-Pacific Artificial Intelligence
	Association (AAIA), an Ordinary Member of the European Academy of
	Sciences and Arts (EASA), and an Ordinary Member of the Academia
	Europaea (AE). Also, he is a Highly Cited Researcher.
\end{IEEEbiographynophoto}

\begin{IEEEbiographynophoto} 
{Bj\"orn Ottersten} (S'87–M'89–SM'99–F'04) received the M.S. degree in electrical engineering and applied physics from Linköping University, Linköping, Sweden, in 1986, and the Ph.D. degree in electrical engineering from Stanford University, Stanford, CA, USA, in 1990. He has held research positions with the Department of Electrical Engineering, Linköping University, the Information Systems Laboratory, Stanford University, the Katholieke Universiteit Leuven, Leuven, Belgium, and the University of Luxembourg, Luxembourg. From 1996 to 1997, he was the Director of Research with ArrayComm, Inc., a start-up in San Jose, CA, USA, based on his patented technology. In 1991, he was appointed Professor of signal processing with the Royal Institute of Technology (KTH), Stockholm, Sweden. Dr. Ottersten has been Head of the Department for Signals, Sensors, and Systems, KTH, and Dean of the School of Electrical Engineering, KTH. He is currently the Director for the Interdisciplinary Centre for Security, Reliability and Trust, University of Luxembourg. He is a recipient of the IEEE Signal Processing Society Technical Achievement Award, the EURASIP Group Technical Achievement Award, and the European Research Council advanced research grant twice. He has co-authored journal papers that received the IEEE Signal Processing Society Best Paper Award in 1993, 2001, 2006, 2013, and 2019, and 8 IEEE conference papers best paper awards. He has been a board member of IEEE Signal Processing Society, the Swedish Research Council and currently serves of the boards of EURASIP and the Swedish Foundation for Strategic Research. Dr. Ottersten has served as Editor in Chief of EURASIP Signal Processing, and acted on the editorial boards of IEEE Transactions on Signal Processing, IEEE Signal Processing Magazine, IEEE Open Journal for Signal Processing, EURASIP Journal of Advances in Signal Processing and Foundations and Trends in Signal Processing. He is a fellow of EURASIP.

\end{IEEEbiographynophoto}

\end{document}